\PassOptionsToPackage{dvipsnames,table}{xcolor}
\PassOptionsToPackage{hypertexnames=false}{hyperref}
\documentclass[acmsmall,screen,natbib]{acmart}
\AtBeginDocument{%
  }

\setcopyright{cc}
\setcctype{by-sa}
\acmDOI{10.1145/3729251}
\acmYear{2025}
\acmJournal{PACMPL}
\acmVolume{9}
\acmNumber{PLDI}
\acmArticle{152}
\acmMonth{6}
\received{2024-11-14}
\received[accepted]{2025-03-06}

\title{Random Variate Generation with Formal Guarantees}

\author{Feras A.~Saad}
\orcid{0000-0002-0505-795X}
\affiliation{%
  \institution{Carnegie Mellon University}
  \department{Computer Science Department}
  \city{Pittsburgh}
  \state{PA}
  \country{USA}}
\email{fsaad@cmu.edu}

\author{Wonyeol Lee}
\orcid{0000-0003-0301-0872}
\affiliation{%
  \institution{POSTECH}
  \department{Department of Computer Science and Engineering}
  \city{Pohang}
  \state{Gyeongbuk}
  \country{Republic of Korea}}
\email{wonyeol.lee@postech.ac.kr}



\usepackage{adjustbox}
\usepackage{afterpage}
\usepackage{amsmath}
\usepackage{amsthm}
\usepackage{mathtools}
\usepackage{thmtools}
\usepackage{stmaryrd}
\usepackage{bm}
\usepackage{booktabs}
\usepackage{comment}
\usepackage{graphicx}
\usepackage{multicol}
\usepackage{multirow}
\usepackage{nicematrix}
\usepackage{soul}
\usepackage{siunitx}
\usepackage[mathscr]{eucal}
\usepackage{accents} % for \accentset{...}
\usepackage{wrapfig}

\numberwithin{equation}{section}

\usepackage{listings}
\definecolor{codegray}{rgb}{0,0.6,0}
\definecolor{codegreen}{rgb}{0,0.6,0}
\lstdefinestyle{CC}{
  language=C,
  basicstyle=\ttfamily\scriptsize,
  numbers=left,
  numbersep=2.25pt,
  xleftmargin=8pt,
  framesep=1pt,
  columns=fixed,
  mathescape=true,
  escapechar={@},
  commentstyle=\color{codegray},
}

\usepackage{tikz}
\usepackage{tikz-qtree}
\usetikzlibrary{arrows}
\usetikzlibrary{arrows.meta}
\usetikzlibrary{backgrounds}
\usetikzlibrary{calc}
\usetikzlibrary{decorations.pathreplacing}
\usetikzlibrary{intersections}
\usetikzlibrary{fit}
\usetikzlibrary{patterns}
\usetikzlibrary{positioning}
\usetikzlibrary{shapes}
\usetikzlibrary{shadows}
\usetikzlibrary{shapes.geometric}
\usetikzlibrary{tikzmark}

\usepackage[inline]{enumitem}

\usepackage{newfloat}
\usepackage[subrefformat=parens]{subcaption}
\DeclareFloatingEnvironment[name=Listing]{listing}
\usepackage{doi}
\usepackage{url}
\setcitestyle{square,citesep={,},notesep={;\,}}

\usepackage{algorithm}
\usepackage[indLines=true,noEnd=true,commentColor=gray,endLComment={}]{algpseudocodex}
\algrenewcommand\algorithmicrequire{\textbf{Input:}}
\algrenewcommand\algorithmicensure{\textbf{Output:}}
\algrenewcommand\algorithmicthen{}

\newlist{contributions}{enumerate}{1}
\setlist[contributions,1]{wide=0pt, label={\bfseries (C\arabic*)}, ref=(C\arabic*)}

\usepackage[capitalize]{cleveref}
\newcommand{\crefrangeconjunction}{--}
\makeatletter
\DeclareRobustCommand{\labelcrefrange}[2]{\@crefrangenostar{labelcref}{#1}{#2}}
\makeatother
\AtEndPreamble{%
  \crefname{equation}{}{}
  \Crefname{equation}{Eq.}{Eqs.}
  \crefname{line}{line}{lines}%
  \Crefname{line}{Line}{Lines}%
  \crefname{definition}{Def.}{Defs.}%
  \Crefname{definition}{Def.}{Defs.}%
  \crefname{proposition}{Prop.}{Props.}%
  \Crefname{proposition}{Prop.}{Props.}%
  \crefname{corollary}{Cor.}{Cors.}%
  \Crefname{corollary}{Cor.}{Cors.}%
  \crefformat{section}{#2\S#1#3}
  \Crefformat{section}{#2\S#1#3}
  \crefmultiformat{section}{#2\S#1#3}{ and~#2\S#1#3}{, #2\S#1#3}{, and~#2\S#1#3}
  \Crefmultiformat{section}{#2\S#1#3}{ and~#2\S#1#3}{, #2\S#1#3}{, and~#2\S#1#3}
  \crefrangeformat{section}{#3\S#1#4\crefrangeconjunction{}#5\S#2#6}
  \crefname{contributionsi}{Contribution}{Contributions}
}

\declaretheorem[style=acmplain,numberwithin=section,qed=\textup{\guillemotleft}]{theorem}

\declaretheorem[style=acmplain,sibling=theorem,qed=\textup{\guillemotleft}]{proposition}
\declaretheorem[style=acmplain,sibling=theorem,qed=\textup{\guillemotleft}]{corollary}

\declaretheorem[style=acmdefinition,sibling=theorem,qed=\textup{\guillemotleft}]{definition}
\declaretheorem[style=acmdefinition,sibling=theorem,qed=\textup{\guillemotleft}]{remark}
\declaretheorem[style=acmdefinition,sibling=theorem,qed=\textup{\guillemotleft}]{example}

\numberwithin{theorem}{section}
\numberwithin{lemma}{section}
\numberwithin{proposition}{section}
\numberwithin{corollary}{section}
\numberwithin{remark}{section}
\numberwithin{definition}{section}
\numberwithin{example}{section}

\newcommand{\newbar}{%
  \text{\smash[b]{\scalebox{1.4}[0.85]{\trimbox{0pt .55ex}{\textup{\textminus}}}}}
}
\newcommand{\mybar}[1]{%
  \mathchoice
      {\accentset{\displaystyle\newbar}{#1}}
      {\accentset{\textstyle\newbar}{#1}}
      {\accentset{\scriptstyle\newbar}{#1}}
      {\accentset{\scriptscriptstyle\newbar}{#1}}
}

\newcommand{\noqed}{\let\qed\relax}

\DeclarePairedDelimiter{\abs}{\lvert}{\rvert}
\DeclarePairedDelimiter{\set}{\lbrace}{\rbrace}
\DeclarePairedDelimiter\ceil{\lceil}{\rceil}
\DeclarePairedDelimiter\floor{\lfloor}{\rfloor}

\newcommand{\defas}{\coloneqq}
\newcommand{\asdef}{\eqqcolon}
\newcommand{\diff}{\mathrm{d}}%

\newcommand{\bfmt}{\mathbb{B}} % \mathbb{F}

\newcommand{\posit}{\mathbb{P}}
\newcommand{\positn}{\posit_n}

\newcommand{\borel}[1]{\mathcal{B}\left(#1\right)}
\newcommand{\dom}{\mathrm{dom}}
\newcommand{\float}{\mathbb{F}}
\newcommand{\floatEm}{\float^E_m}
\newcommand{\real}{\mathbb{R}}
\newcommand{\nat}{\mathbb{N}}

\newcommand{\bool}{\set{0,1}}

\newcommand{\realext}{\mybar{\real}}

\newcommand{\round}{\mathrm{rnd}}
\newcommand{\roundfl}[1]{\round_{#1}}

\newcommand{\rtd}{{\downarrow}}

\newcommand{\SampleNaive}{\textsc{\textsc{GenCBS}}}

\newcommand{\SampleOpt}{\textsc{\textsc{GenOpt}}}
\newcommand{\SampleNaiveImpl}{\textsc{CBS}}
\newcommand{\SampleOptImpl}{\textsc{Opt}}

\newcommand{\GetBit}{\textsc{ExtractBit}}
\newcommand{\Flip}{\textsc{RandBit}}
\newcommand{\Quantile}{\textsc{Quantile}}
\newcommand{\QuantileExt}{\textsc{Quantile}} % ExactQuantile2
\newcommand{\Bernoulli}{\textsc{Bernoulli}}

\newcommand{\ExactSubtract}{\textsc{ExtractBitPreproc}}
\newcommand{\ExactSubtractOne}{\textsc{ExtractBitPreproc1}}
\newcommand{\ExactSubtractTwo}{\textsc{ExtractBitPreproc2}}

\newcommand{\ExactRatio}{\textsc{ExactRatio}}

\newcommand{\fL}{f_0}
\newcommand{\fR}{f_1}
\newcommand{\fM}{f_2} % {f_3}
\newcommand{\dL}{d_0}
\newcommand{\dR}{d_1}
\newcommand{\dM}{d_2} % {d_3}

\newcommand{\WCDF}{CDF}
\newcommand{\WSF}{SF}
\newcommand{\ECDF}{DDF}

\makeatletter
\renewcommand*\env@matrix[1][*\c@MaxMatrixCols c]{%
  \hskip -\arraycolsep
  \let\@ifnextchar\new@ifnextchar
  \array{#1}}
\makeatother

\allowdisplaybreaks

\newcommand\Tstrut{\rule{0pt}{2.25ex}}
\newcommand\Bstrut{\rule[-1.1ex]{0pt}{0pt}}

\makeatletter
\patchcmd{\NAT@test}{\else \NAT@nm }{\else \NAT@nmfmt{\NAT@nm}}{}{}
\makeatother

\makeatletter
\renewcommand\paragraph{\@startsection{paragraph}{4}{\parindent}%
  {-.25\baselineskip \@plus -2\p@ \@minus -.2\p@}%
  {-3.5\p@}%
  {\ACM@NRadjust{\@parfont\@adddotafter}}}
\renewcommand\noindentparagraph{\@startsection{paragraph}{4}{\z@}%
  {-.25\baselineskip \@plus -2\p@ \@minus -.2\p@}%
  {-3.5\p@}%
  {\ACM@NRadjust{\@parfont}}}
\makeatother

\makeatletter
\newtheoremstyle{acmplain}%
  {.25\baselineskip\@plus.1\baselineskip\@minus.1\baselineskip}% space above
  {.25\baselineskip\@plus.1\baselineskip\@minus.1\baselineskip}% space below
  {\@acmplainbodyfont}% body font
  {\@acmplainindent}% indent amount
  {\@acmplainheadfont}% head font
  {.}% punctuation after head
  {.5em}% spacing after head
  {\thmname{#1}\thmnumber{ #2}\thmnote{ {\@acmplainnotefont(#3)}}}% head spec

\newtheoremstyle{acmdefinition}%
  {.1\baselineskip\@plus.1\baselineskip\@minus.1\baselineskip}% space above
  {.1\baselineskip\@plus.1\baselineskip\@minus.1\baselineskip}% space below
  {\@acmdefinitionbodyfont}% body font
  {\@acmdefinitionindent}% indent amount
  {\@acmdefinitionheadfont}% head font
  {.}% punctuation after head
  {.5em}% spacing after head
  {\thmname{#1}\thmnumber{ #2}\thmnote{ {\@acmdefinitionnotefont(#3)}}}% head spec

\makeatother

\begin{document}

\begin{abstract}
Generating random variates is a fundamental operation in diverse areas of
computer science and is supported in almost all modern programming languages.
Traditional software libraries for random variate generation are grounded
in the idealized ``Real-RAM'' model of computation, where algorithms are
assumed to be able to access uniformly distributed real numbers from the unit interval
and compute with infinite-precision real arithmetic.
These assumptions are unrealistic, as any software implementation of
a Real-RAM algorithm on a physical computer can instead access a stream of
individual random bits and computes with finite-precision arithmetic.
As a result, existing libraries have few theoretical guarantees in practice.
For example, the actual distribution of a random variate generator
is generally unknown, intractable to quantify, and arbitrarily
different from the desired distribution; causing runtime errors,
unexpected behavior, and inconsistent APIs.

This article introduces a new approach to principled and practical random
variate generation with formal guarantees.
The key idea is to first specify the desired probability distribution in
terms of a finite-precision numerical program that defines its cumulative
distribution function (CDF), and then generate exact random variates
according to this CDF.
We present a universal and fully automated method to synthesize exact
random variate generators given any numerical CDF implemented in any binary
number format, such as floating-point, fixed-point, and posits.
The method is guaranteed to operate with the same precision used to specify
the CDF, does not overflow, avoids expensive arbitrary-precision
arithmetic, and exposes a consistent API.
The method rests on a novel space-time optimal implementation for the class
of generators that attain the
information-theoretically optimal \citeauthor{knuth1976} entropy rate,
consuming the least possible number of input random bits per output
variate.
We develop a random variate generation library using our method in C
and evaluate it on a diverse set of ``continuous'' and ``discrete''
distributions, showing competitive runtime with the
state-of-the-art GNU Scientific Library while delivering higher
accuracy, entropy efficiency, and automation.
\end{abstract}

\maketitle

%!TEX root=./paper.tex

\section{Introduction}
\label{sec:intro}
%!TEX root=./paper.tex
%!TEX root=./paper.tex
\begin{figure}[t]
\centering
\begin{tikzpicture}
\node[name=RV, minimum width=4.5cm, draw]{$X:[0,1] \to \real$};
\node[name=ES, draw, minimum width=4.75cm, left=5.1cm of RV.center, anchor=center, label={[name=ESL,font=\bfseries]above: Entropy Source}]{$\mathrm{Uniform}([0,1])$ real number};
\node[name=OT, right=.45cm of RV.east, minimum width=1.25cm, anchor=west, text width=1.5cm]{$x \in \real$};
\node[name=LB, left=3.25 of ES.center, anchor=center]{Infinite};
\draw[-latex] (ES) -- (RV);
\draw[-latex] (RV) -- (OT);

\node[anchor=center,at=(LB.center |- ESL.center),font=\bfseries]{Precision};
\node[anchor=center,at=(RV.center |- ESL.center),font=\bfseries]{Random Variate Generator};
\node[anchor=west,at=(OT.west |- ESL.center),font=\bfseries]{Output};

\node[name=RV, minimum width=4.5cm, draw, below=.25 of RV]{$X:\bool^* \rightharpoonup \set{0,1}^n$};
\node[name=ES, draw, minimum width=4.75cm,left=5.1cm of RV.center, anchor=center]{$\mathrm{Bernoulli}(1/2)$ i.i.d.~bit stream};
\node[name=OT, right=.45cm of RV.east, minimum width=1.5cm, anchor=west, text width=1.75cm]{$x \in \set{0,1}^n$};
\node[name=LB, left=3.25 of ES.center,anchor=center]{Finite};
\draw[-latex] (ES) -- (RV);
\draw[-latex] (RV) -- (OT);
\end{tikzpicture}
\captionsetup{skip=4pt,belowskip=-12pt}
\caption{Random variate generation with infinite-precision (Real-RAM) and
finite-precision computation.}
\label{fig:real-world}
\end{figure}
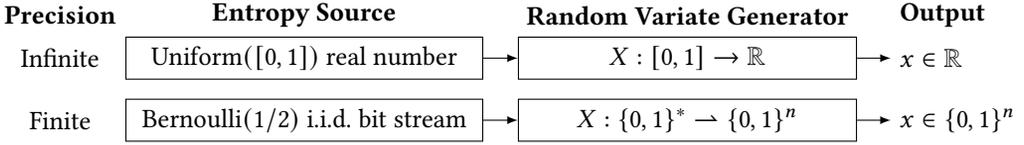

The purpose of a random variate generation algorithm is to produce random
numbers that adhere to a specified probability distribution.
These distributions can be discrete, such as the Poisson distribution over
the natural numbers, or continuous, like the Gaussian distribution over the
real numbers.
A fundamental result of probability theory establishes that, in purely
mathematical terms, any random variate generator corresponds to a
function that takes as input a real number $u \in [0,1]$ drawn uniformly
at random from the unit interval and returns as output a real number
$x \in \real$ (\cref{fig:real-world}, top row).
Software libraries for random variate generation use this insight to
develop numerical algorithms, whose correctness properties rest on a fictitious
``Real-RAM'' computer that can perform infinitely precise computation over
real numbers in constant time~\citep[\S1.1 Assumptions I---III]{devroye1986}.

\paragraph{Key Challenges}
The Real-RAM model was introduced in \citeauthor{shamos1978}'s
\citeyear{shamos1978} dissertation~\citep{shamos1978} on computational
geometry and adopted in \citeauthor{devroye1986}'s \citeyear{devroye1986}
random variate generation book~\citep{devroye1986} for its conceptual simplicity.
However, this model fails to characterize the error and complexity
properties of software that executes on finite-precision hardware,
as cautioned by \citeauthor{neumann1951} in \citeyear{neumann1951}:
\begin{quote}
\small
\itshape
Anyone who considers arithmetic methods of producing random digits is, of
course, in a state of sin\dots
If one considers arithmetic methods in detail, it is quickly found
that the critical thing about them is the very obscure, very imperfectly
understood behavior of round-off errors in mathematics\dots
One might as well admit that one can prove nothing, because the
amount of theoretical information about the statistical properties of the
round-off mechanism is nil.\hfill\normalfont ---\Citet{neumann1951}
\end{quote}
A modern illustration of these problems can be found in
differentially private algorithms~\citep{dwork2006} that use
random numbers to obfuscate sensitive datasets, such as
the 2020 US Census~\citep{Abowd2022}.
\Citet{Mironov2012} demonstrates that floating-point effects in the
Laplace random variate generator from existing software libraries can
entirely destroy the real-world privacy guarantees of algorithms that are
otherwise differentially private under the Real-RAM assumption.
Similar issues arise in areas such as lattice-based
cryptography~\citep{schryver2012,roy2013,dwarakanath2014,follath2014,du2015,Karmakar2018},
where developers of secure protocols need more-realistic computational models than
Real RAM that enable them to rigorously characterize
\begin{enumerate}[label=(\alph*),wide=0pt,noitemsep,nosep]
\item the approximation error on a finite-precision computer, to
establish security guarantees;

\item the required register size needed for a given accuracy level, to
design efficient hardware;

\item the entropy consumption per output variate, to predict runtime and
avoid side channel attacks.
\end{enumerate}
\vspace{\topsep}

\paragraph{This Work}

This article introduces a theoretically principled and practical
approach to random variate generation grounded in a finite-precision
model of computation (\cref{fig:real-world}, bottom row).
On the input side, a random variate generator has lazy access to a stream
of independent unbiased random bits.
It returns as output a finite $n$-bit string that represents a number in
some binary encoding, computed using finite-precision arithmetic and
finite memory.
An immediate consequence of this model is that the generator can produce at
most $2^n$ distinct outputs, each with a rational probability.

Our approach to random variate generation begins with a formal
specification of the desired probability distribution of the
random variates.
This specification takes the form of a numerical program
that implements a cumulative distribution function (CDF) $F$.
For any finite-precision number $x$ that can be represented on the
computer, $F(x)$ computes the probability that the random variate is less
than or equal to $x$.
Given such a CDF, our technique automatically synthesizes a random
variate generator that is guaranteed to exactly match $F$.
This generator is also guaranteed to be \textit{entropy-optimal}---it draws
the information-theoretically least number of random bits on the input side
to produce an output, achieving the optimal entropy rate
from \citet{knuth1976}.

\Cref{fig:overview} compares our approach to existing random
variate generation software libraries.
First consider the idealized Real-RAM model (\cref{fig:overview-real}).
Every probability measure has multiple equivalent representations in terms of a
unique cumulative distribution function (CDF), survival function (SF),
quantile function (QF), or infinite collection of identically distributed
random variables (i.e., measurable functions from $[0,1]$ into $\mathbb{R}$).
\Cref{fig:overview-standard} shows the approach
in standard software libraries, which
provide numerical programs for the CDF, SF, QF, and
random variate generator(s)
that approximate the corresponding Real-RAM functions.
The resulting diagram, however, does not ``commute''
(red cross marks in \cref{fig:overview-standard})---due to numerical errors, each
of these numerical programs actually defines a \textit{different} probability
measure, creating an inconsistent API for what was originally a
single coherent probability measure in the infinite-precision Real RAM model.
Further, the actual output distributions of the implemented random
variate generators are generally intractable to compute, making it
difficult to formally characterize their properties such as approximation
error or expected runtime.
\Cref{fig:overview-new} shows the proposed approach,
where the desired probability measure
is first specified using a numerical program that computes its CDF.
This CDF is then used to synthesize an exact random variate generator and
related functions.
The resulting API for the implemented probability measure
is mathematically coherent, perfectly mirroring the Real-RAM equivalences.

%!TEX root=./paper.tex
\begin{figure}[t]
\begin{tikzpicture}
\tikzstyle{boxed}=[fill=white,minimum height=1cm, minimum width=1cm,draw=none]
\tikzstyle{stackbox}=[boxed, draw=black, double copy shadow={shadow xshift=-.5ex, shadow yshift=.5ex}]
\tikzstyle{arrExact}=[latex-latex, color=blue]
\tikzstyle{exactSize}=[font=\footnotesize,inner sep=0pt]
\tikzset{arrowshadow/.style = {-latex,double copy shadow={shadow xshift=-.5ex, shadow yshift=.5ex}}}

\node[name=X, stackbox]{$X_{\mu}$};
\node[name=FX,boxed,below=.75 of X]{$F_\mu$};
\node[name=SX,boxed,below=.75 of FX]{$S_\mu$};
\node[name=QX,boxed,below=.75 of SX]{$Q_\mu$};
\node[name=MU,boxed,above=.15 of X]{$\mu$};

\node[name=X-Label, left=4.4 of X,anchor=west,font=\footnotesize]{\begin{tabular}{@{}l}Random Variable Representations\\(i.e., Random Variate Generators)\end{tabular}};
\node[name=FX-Label, at = (X-Label.west |- FX),anchor=west,font=\footnotesize]{\begin{tabular}{@{}l}Cumulative Distribution\\Function Representation\end{tabular}};
\node[name=SX-Label, at = (X-Label.west |- SX),anchor=west,font=\footnotesize]{\begin{tabular}{@{}l}Survival Function\\Representation\end{tabular}};
\node[name=QX-Label, at = (X-Label.west |- QX),anchor=west,font=\footnotesize]{\begin{tabular}{@{}l}Quantile Function\\Representation\end{tabular}};
\node[name=MU-Label,at = (X-Label.west |- MU),anchor=west,font=\footnotesize]{Probability Measure};

\node[name=Xb, right=2.6 of X, stackbox,label={[inner sep=0pt,anchor=west]above right:${\mathrlap{\textcolor{red}{\times}}\rbrace}\;\mbox{\scriptsize{inequivalent}}$}]{$\widehat{X_{\mu}}$};
\node[name=FXb, at=(FX.center -| Xb), anchor=center, boxed, thick,fill=blue!10!white]{$\widehat{F_\mu}$};
\node[name=SXb, at=(SX.center -| Xb), anchor=center, boxed]{$\widehat{S_\mu}$};
\node[name=QXb, at=(QX.center -| Xb), anchor=center, boxed]{$\widehat{Q_\mu}$};
\node[name=MUc, at=(MU.center -| Xb), anchor=center, boxed,fill=none,font=\footnotesize]{(no coherent measure)};

\node[name=Xc, right=2.6 of Xb,boxed,stackbox,draw=black,label={[inner sep=0pt,anchor=west]above right:${\rbrace}\mbox{\scriptsize{equivalent}}$}]{$X_{\widehat{F_\mu}}$};
\node[name=FXc, at=(FX.center -| Xc), anchor=center, boxed, fill=blue!10!white]{$\widehat{F_{\mu}}$};
\node[name=SXc, at=(SX.center -| Xc), anchor=center, boxed]{$S_{\widehat{F_{\mu}}}$};
\node[name=QXc, at=(QX.center -| Xc), anchor=center, boxed]{$Q_{\widehat{F_{\mu}}}$};
\node[name=MUc, at=(MU.center -| Xc), anchor=center, boxed,fill=none]{$\mu_{\widehat{F_{\mu}}}$};

% Arrows.
\draw[arrowshadow] (X) --node[pos=0.5,label={[inner ysep=0pt,font=\scriptsize,align=center]below:{Numerical Program\\Implementation}}]{} (Xb);
\node[name=X, stackbox,label={[inner sep=0pt,anchor=west]above right:${\rbrace}\mbox{\scriptsize{equivalent}}$}]{$X_{\mu}$};

\foreach \x in {FX, SX, QX} {
  \draw[-latex] (\x) --node[pos=0.5,label={[inner ysep=0pt,font=\scriptsize,align=center]below:{Numerical Program\\Implementation}}]{} (\x b);
}

\draw[latex-latex] (X) --node[pos=0.5,label={[exactSize]left:equivalent}]{} (FX);
\draw[latex-latex] (FX) --node[pos=0.5,label={[exactSize]left:equivalent}]{} (SX);
\draw[latex-latex] (SX) --node[pos=0.5,label={[exactSize]left:equivalent}]{} (QX);

\draw[latex-latex] (Xb) --node[pos=0.5,label={[exactSize,inner xsep=0pt]right:inequivalent}]{\textcolor{red}{$\times$}} (FXb);
\draw[latex-latex] (FXb) --node[pos=0.5,label={[exactSize,inner xsep=0pt]right:inequivalent}]{\textcolor{red}{$\times$}} (SXb);
\draw[latex-latex] (SXb) --node[pos=0.5,label={[exactSize,inner xsep=0pt]right:inequivalent}]{\textcolor{red}{$\times$}} (QXb);

\draw[latex-latex] (Xc) --node[pos=0.5,label={[exactSize]left:equivalent}]{} (FXc);
\draw[latex-latex] (FXc) --node[pos=0.5,label={[exactSize]left:equivalent}]{} (SXc);
\draw[latex-latex] (SXc) --node[pos=0.5,label={[exactSize]left:equivalent}]{} (QXc);

\draw[densely dashed] (FXb) -- (FXc);
\node[coordinate,name=Xd, right=.3 of Xc ]{};
\node[coordinate,name=SXd,right=.3 of SXc]{};
\node[coordinate,name=QXd,right=.3 of QXc]{};
\draw[gray,latex-] (SXc) -- (SXd);
\draw[gray,latex-] (QXc) -- (QXd);
\draw[gray,latex-,transform canvas={yshift=-0.25cm}] (Xc)  -- (Xd);
\draw[gray]        ([yshift=-0.25cm]Xd) --node[pos=0.5,exactSize,xshift=.4cm,rotate=-90,anchor=center]{\begin{tabular}[t]{@{}c@{}}automatically\\[-2pt] synthesized\end{tabular}} (QXd);
\draw[gray,latex-] (Xc)                 --node[pos=0.5,label={[exactSize]right:\begin{tabular}[t]{@{\;}c@{}}entropy\\[-3pt]optimal\end{tabular}}]{}            (Xd);

\node[name=real-ram, below=0 of QX, xshift=-.25cm, text width=4cm] {
  \subcaption{\begin{tabular}[t]{@{}c@{}}Real RAM \\(Consistent)\end{tabular}}
  \label{fig:overview-real}
  };

\node[name=standard, at=(real-ram.center -| QXb), text width=4cm] {
  \subcaption{\begin{tabular}[t]{@{}c@{}}Standard Approach\\(Inconsistent)\end{tabular}}
  \label{fig:overview-standard}
  };

\node[name=proposed, at=(real-ram.center -| QXc), text width=4cm] {
  \subcaption{\begin{tabular}[t]{@{}c@{}}Proposed Approach\\(Consistent)\end{tabular}}
  \label{fig:overview-new}
  };

\end{tikzpicture}
\captionsetup{skip=2pt, belowskip=-10pt}
\caption{Overview of the standard and proposed approaches to random variate generation software libraries.}
\label{fig:overview}
\end{figure}
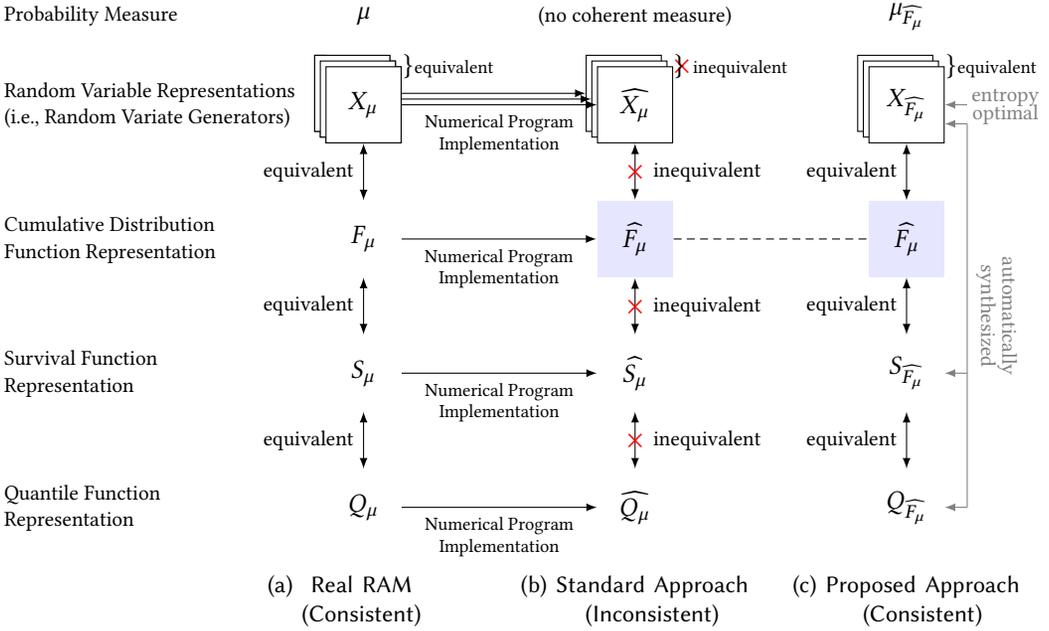

\paragraph{Contributions}
This article makes the following contributions.

\begin{contributions}

\item \label{contribution:formulation}
{\bf Formulation of exact random variate generation using finite precision.}
We rigorously formulate the problem of generating {exact} random variates
given a finite-precision implementation of a CDF.
This approach is fundamentally different from existing libraries (\cref{fig:overview}).
It guarantees a {coherent API} for the cumulative distribution,
survival, and quantile functions of the implemented generator that
all agree with one another.
It also enables the {fully automatic construction} of a
random variate generator given a {formal specification} of the desired
probability distribution, and allows for {strong theoretical guarantees} on  % maintains
exactness, entropy optimality, and practical efficiency.

% Alg 2.
\item \label{contribution:binary}
{\bf Exact and optimal random variate generators for binary-coded distributions (\cref{sec:binary}).}
We present a sound and entropy-optimal algorithm
(\cref{theorem:sampler-opt}) for generating random variates given any
binary-coded probability distribution, which is a universal mathematical
representation for probability measures over $\real$.
This algorithm improves upon a method described in \citet{knuth1976}:
it obtains optimal space-time complexity by exploiting
properties of binary expansions of real numbers that govern the
structure of entropy-optimal generators
(\cref{proposition:bit-patterns-simple,proposition:bit-patterns}).

% Alg 4 & 5-6.
\item \label{contribution:floating}
{\bf Exact and optimal random variate generators in finite precision (\cref{sec:floating}).}
We specialize the universal algorithm for binary-coded probability
distributions to soundly generate exact random variates given a finite-precision
implementation of a CDF over any {binary number format}
(e.g., floating-point, fixed-point, posits; \cref{theorem:sampler-impl-correct-opt}).
This algorithm is {information-theoretically} optimal
and {highly efficient} in software:
it is guaranteed to require the same precision
as the given CDF (\cref{prop:exactsubtract1}) and
uses fast integer arithmetic instead of
expensive arbitrary-precision arithmetic.

% Alg in Appendix.
\item \label{contribution:survival}
{\bf Extended-accuracy random variate generators that combine both CDF and SF (\cref{sec:survival}).}
We present an extension of the method from \cref{sec:floating}
that achieves {higher accuracy}, especially in the tails of a probability
%% distribution, using a principled combination of a Word-RAM CDF and
distribution, using a principled combination of a finite-precision CDF and
survival function (SF) implementation (\cref{theorem:ecda-default}).
This algorithm enjoys the same theoretical guarantees as before, while
being able to represent twice as many outcomes as compared to only a CDF or
SF.

% Results.
\item \label{contribution:evaluation}
{\bf Implementation and empirical evaluations (\cref{sec:evaluation}).}
We develop and evaluate a random variate generation library using our methods
in C.
The results show that, as compared to the state-of-the-art
GNU Scientific Library~\citep{galassi2009}, our
generators are \begin{enumerate*}[label=(\roman*)]
\item more entropy-efficient, consuming 2.6x--142x fewer random bits per
output variate;

\item more representative of the ideal distribution range, covering up to
$10^{35}$x wider intervals; and

\item more automated and amenable to program analysis, having known
output distributions.
\end{enumerate*}
The results also show that the extended-accuracy methods in \cref{sec:survival} incur
negligible overhead in entropy and runtime over the original versions in \cref{sec:floating}.

\end{contributions}

\section{Overview}
\label{sec:overview}
%!TEX root=./paper.tex

\subsection{Mathematical Representations of Probability Distributions}
\label{sec:overview-mathematical}

Let $\lambda$ denote the Lebesgue measure and $\mathcal{B}(\real)$
the Borel sigma-algebra.
Every probability measure $\mu$ over $\real$
(e.g., Gaussian, Gamma, Poisson, etc.)
has multiple equivalent mathematical representations:
\begin{align}
&\mbox{PM}  && \mbox{Probability Measure}              && \mu: \borel{\real} \to [0,1] \label{eq:dist-pm} \\[-2.5pt]
&\mbox{RV}  && \mbox{Random Variable}                  && X: ([0,1], \borel{[0,1]}) \to (\real, \borel{\real}) \label{eq:dist-rv}\\[-2.5pt]
&\mbox{CDF} && \mbox{Cumulative Distribution Function} && F(x) \defas \mu((-\infty, x]) \asdef \Pr(X \le x), x \in \real \label{eq:dist-cdf}\\[-2.5pt]
&\mbox{SF}  && \mbox{Survival Function}                && S(x) \defas \Pr(X > x) = 1 - F(x), x \in \real \label{eq:dist-sf}\\[-2.5pt]
&\mbox{QF}  && \mbox{Quantile Function}                && Q(u) \defas \inf\set{x \in \real \mid u \le F(x)}, u \in [0,1]. \label{eq:dist-qf}
\end{align}
\Cref{fig:overview-real} shows the correspondences among
\crefrange{eq:dist-pm}{eq:dist-qf} under the Real-RAM model of computation.
There is a bijection between the spaces of PM, CDF, SF, QF;
and a surjection between these spaces and the space RV.
We denote these correspondences as
$\mu \sim \set{F_\mu, S_\mu, Q_\mu, X_\mu}$,
$F \sim \set{\mu_F, S_F, Q_F, X_F}$,
$S \sim \set{\mu_S, F_S, Q_S, X_S}$,
$Q \sim \set{\mu_Q, F_Q, S_Q, X_Q}$,
$X \sim \set{\mu_X, F_X, S_X, Q_X}$.
The RV is the only non-unique representation---infinitely many random variables
(understood as measurable mappings from the underlying probability space $[0,1]$
into $\real$) can have the same distribution, i.e.,
$X \ne X'$ but $\mu_{X}(A) \defas \lambda(X^{-1}(A)) = \lambda(X'^{-1}(A)) \asdef \mu_{X'}(A)$
for all $A \in \borel{\real}$.
Random variables $X \ne X'$ that disagree on positive measure sets
suggest different random variate generation strategies for $\mu$.
\begin{example}
\label{example:normal-rvs}
The following random variables all
have a standard Gaussian distribution over $\real$:
\bgroup
\setlength{\abovedisplayskip}{0pt}
\setlength{\belowdisplayskip}{2pt}
\begin{align*}
\textstyle
X_1(\omega) = \inf\set*{x \in \real \mid \omega \le \int_{-\infty}^{x}\frac{e^{-\frac{t^2}{2}}}{\sqrt{2\pi}}\diff{t}},
&&
X_2(\omega) = \frac{\sqrt{-2\ln{u_1(\omega)}}}{1/\cos(2\pi u_2(\omega))},
&&
X_3(\omega) = \lim_{n \to \infty} \textstyle \frac{\sum_{i=1}^n u_i(\omega) -n/2}{\sqrt{n/12}},
\end{align*}
\egroup
where $\{u_i(\omega)\}_{i=1}^\infty$ denote countably many i.i.d.~uniform
numbers on $[0,1]$ ``split'' from $\omega$ (\cref{remark:random-variable-domain}).
These functions describe different implementations of Gaussian
random variate generators: the inverse transform method $X_1$,
Box-Muller transform $X_2$, and central limit approximation $X_3$.
\Cref{lst:gsl-samplers}
(\crefrange{lst:gsl-samplers-r1}{lst:gsl-samplers-r4})
lists four Gaussian random variate generators from the GNU
Scientific Library.
\end{example}

\subsection{Traditional Software Implementations of Probability Distributions}
\label{sec:overview-software}

This section illustrates problems with the approach in traditional
random variate generation libraries (e.g., \citep{matlab2024,numpy2020,scipy2020,pytorch2019,galassi2009})
that provide interfaces for~\crefrange{eq:dist-pm}{eq:dist-qf} as
shown schematically in \cref{fig:overview-standard}.
We present a small%
\footnote{\label{footnote:gsl-gamma}%
  Library implementations of random variate generators and associated
  functions can span 100s of lines of code, e.g., \citep{gslGammaCDF,gslGammaRV}.
  }
but realistic example for the $\mathrm{Exponential}(s)$
distribution (\cref{lst:exponential}) to highlight some issues
with finite-precision numerical implementations of idealized Real-RAM algorithms.%
\footnote{\label{foonote:numpy}%
A cursory inspection of three widely used Python
libraries (NumPy~\citep{numpy2020}, SciPy~\citep{scipy2020},
PyTorch~\citep{pytorch2019}) surfaced ${\sim}90$ user-reported issues in the
random variate generation algorithms, almost all
related to numerical error, shown in \cref{appx:survey}.}

\paragraph{Problems with Uniform Source}
\Crefrange{lst:exponential-u0}{lst:exponential-u1} of \cref{lst:exponential}
show common~\citep{press1992} numerical approximations of a uniform variable in
$[0,1]$, which serves as the primitive source of randomness for random
variate generation (cf.~\cref{fig:real-world,eq:dist-rv}).
The \texttt{uniform} function calls $\texttt{rand}$, which returns an
integer between 0 and $\texttt{RAND\_MAX}$, and divides the result by
$\texttt{RAND\_MAX}+1$ to give a float in $[0,1)$; \texttt{uniform\_pos}
returns a float in $(0,1)$ by rejection sampling.
Issues with generators of this sort are well-documented in the literature:
\citet{Goualard2020,Goualard2022} and \citet{lemire2019}
explore these drawbacks in detail.
For example, while including or omitting endpoints $\set{0,1}$ is of no
consequence in Real RAM, these values have nonzero probability when
using finite precision, causing many downstream errors depending
on the implementation.%
\footnote{\label{footnote:pytorch-exponential}%
  This specific issue arises in the PyTorch exponential generator, see
  \url{https://github.com/pytorch/pytorch/issues/22557}}
Moreover, typical implementations of \texttt{uniform} return a paltry $10^{-7}$\%--7\% of
all representable floating-point numbers in $[0,1]$ (\cref{prop:std-uniform-density}),
which further limits the accuracy of random variate generator
algorithms that invoke these functions.%
  \footnote{\label{footnote:pytorch-uniform}%
    An extensive discussion of this challenge among developers is
    found in \url{https://github.com/pytorch/pytorch/issues/16706}}

%!TEX root=./paper.tex
\begin{listing}[t]
\caption{Implementing an exponential probability distribution in C.}
\label{lst:exponential}
\begin{lstlisting}[style=CC,frame=tB]
// Typical random variate generators for the "uniform" distribution over unit interval.
@\label{lst:exponential-u0}@double uniform    (){return ((double)rand()) / (RAND_MAX + 1);}           // $\color{codegray}\smash{U\sim\mathrm{Uniform}([0,1))}$
@\label{lst:exponential-u1}@double uniform_pos(){double u; do {u=uniform();} while (u==0); return u;} // $\color{codegray}\smash{U\sim\mathrm{Uniform}((0,1))}$

// Cumulative distribution (CDF), survival (SF), & quantile (QF) function implementations.
@\label{lst:exponential-cf}@double exp_cdf (double x, double s) {return x <= 0 ? 0 : -expm1(-x/s);} // $\color{codegray}\smash{F(x; s)=1-e^{-x/s}}$
@\label{lst:exponential-sf}@double exp_sf  (double x, double s) {return x <= 0 ? 1 : exp(-x/s);}    // $\color{codegray}\smash{S(x; s)=e^{-x/s}}$
@\label{lst:exponential-qf}@double exp_qf  (double u, double s) {return -s * log1p(-u);}            // $\color{codegray}\smash{Q(u; s)=-s\ln(1-u)}$

// TRADITIONAL APPROACH: Implement ad-hoc random variate generators.
@\label{lst:exponential-g0}@double exp_generate     (double s) {return exp_qf(uniform(), s);}       // $\color{codegray} Q(U; s)$
@\label{lst:exponential-g1}@double exp_generate_alt (double s) {return -s * log(uniform_pos());}    // $\color{codegray} Q(1-U; s)$

// THIS WORK: Exact random variate generators given a formal CDF and/or SF specification.
@\label{lst:exponential-synth1}@GENERATE_FROM_CDF (exp_cdf, 1.0);
@\label{lst:exponential-synth2}@GENERATE_FROM_SF  (exp_sf , 1.0);
@\label{lst:exponential-synth3}@GENERATE_FROM_DDF (exp_cdf, exp_sf, 1.0);
\end{lstlisting}
\end{listing}

\paragraph{Inconsistencies in the CDF, SF, QF}
\Crefrange{lst:exponential-cf}{lst:exponential-qf} of \cref{lst:exponential}
show numerically stable implementations of \crefrange{eq:dist-cdf}{eq:dist-qf}.
Let us explore some inconsistencies among these representations.
First consider \texttt{exp\_cdf} and \texttt{exp\_sf}.
These functions formally define
distributions over floats (\cref{remark:CDA-sample-intractable}),
but they have entirely different properties.
When $\texttt{s}=1$,
the former's range (i.e., support) contains floats within $[7.01\times{10}^{-46}, 17.33]$
and the latter $[2.98 \times{10}^{-8}, 103.97]$.
Now consider $\texttt{exp\_cdf}$ and $\texttt{exp\_qf}$.
If \texttt{u} is the float immediately before 1,
then the return value of \texttt{exp\_qf(u,1)} is
$16.635$ whereas the exact $u$-quantile of \texttt{exp\_cdf($\cdot$,1)} is $16.230$,
yielding a large absolute error exceeding $0.405$.
Relative errors are also large: if
\texttt{u} is the float immediately after 0,
then the return value of \texttt{exp\_qf(u,1)} is \textit{twice} larger than
the exact $u$-quantile of \texttt{exp\_cdf($\cdot$,1)}.
Many similar issues can be surfaced.

\paragraph{Inconsistencies in the Generators}
\Crefrange{lst:exponential-g0}{lst:exponential-g1} of
\cref{lst:exponential} show two exponential random variate
generators, $Q(U)$ and $Q(1-U)$, based on the inverse-transform method
(\cref{remark:billingsley}).
However, \texttt{exp\_generate} is not consistent with the distribution
specified by \texttt{exp\_cdf} because it calls into $\texttt{exp\_qf}$,
which is itself not consistent with \texttt{exp\_cdf}.
Supposing that $\texttt{RAND\_MAX} = 2^{32}-1$,
the output of \texttt{exp\_generate} lies in the range $[0, 22.18]$ when $\texttt{s}=1$,
which is again different from the ranges of both
\texttt{exp\_cdf} and \texttt{exp\_sf}.
Further, \texttt{exp\_generate} is also not consistent with the
distribution specified by \texttt{exp\_qf} even though it directly invokes it,
because $\texttt{uniform}$ does not cover all floats in $[0,1]$.
The \texttt{exp\_generate\_alt} function has similar inconsistencies,
with the additional caveat that
implementing $Q(1-U)$ requires using $\texttt{uniform\_pos}$
instead of $\texttt{uniform}$,
because $\texttt{log}(0.)$ may
return $\texttt{-inf}$, $\texttt{nan}$, or even a domain error,
depending on the language (cf.~\Cref{footnote:pytorch-exponential}).

\subsection{This Work: Exact Random Variate Generators from Formal Specifications}
\label{sec:overview-this-work}

\Crefrange{lst:exponential-synth1}{lst:exponential-synth3}
of \cref{lst:exponential} show the approach to random variate
generation introduced in this work, which is based on coherent
implementations of \crefrange{eq:dist-pm}{eq:dist-qf}.
The \texttt{GENERATE\_FROM\_CDF} expression takes the name of any numerical
CDF implementation (i.e., \texttt{exp\_cdf}) and values of
distributional parameters (i.e., \texttt{s} is \texttt{1.0}), and returns a random % \texttt{mu}
floating-point number \texttt{x} \textit{precisely} with cumulative probability
\texttt{cdf(x)}.
An alternate generator is \texttt{GENERATE\_FROM\_SF}, which returns
\texttt{x} with cumulative probability $1-_{\real}\texttt{exp\_sf(x)}$, where the
subtraction is exact (i.e., no floating-point rounding error).
These generators reflect the fact that \texttt{exp\_cdf} and
\texttt{exp\_sf} define different distributions.
We also have an extended-accuracy generator, \texttt{GENERATE\_FROM\_DDF},
which \textit{combines} the \texttt{exp\_cdf} and \texttt{exp\_sf}
specifications into a single coherent generator.
Recall that, for $\texttt{s}=1$,
the former has range $[7.01\times{10}^{-46}, 17.33]$ and the
latter $[2.98 \times{10}^{-8}, 103.97]$,
so neither is globally more precise than the other.
The output of this generator agrees with \texttt{exp\_cdf} below its median,
where the CDF is more precise, and agrees with \texttt{exp\_sf} above its
median, where the SF is more precise.
Its largest value (103.97) is 4.7x higher than that of
\texttt{exp\_generate} (22.18), giving higher coverage of $(0, \infty)$.

%!TEX root=./paper.tex
\begin{listing}[t]
\caption{Using high-quality cumulative distribution and survival function implementations from
the GNU Scientific Library (GSL) to automatically generate ``Gaussian'' variates (only
function signatures are shown).}
\label{lst:gsl-samplers}
\begin{lstlisting}[style=CC,frame=tB]
// GSL: Existing random variate generators for Gaussian distribution (renamed for clarity).
@\label{lst:gsl-samplers-r1}@double gsl_ran_gaussian_inverse_cdf  (const gsl_rng *r, double sigma);
@\label{lst:gsl-samplers-r2}@double gsl_ran_gaussian_box_muller   (const gsl_rng *r, double sigma);
@\label{lst:gsl-samplers-r3}@double gsl_ran_gaussian_ziggurat     (const gsl_rng *r, double sigma);
@\label{lst:gsl-samplers-r4}@double gsl_ran_gaussian_ratio_method (const gsl_rng *r, double sigma);

// GSL: Numerical implementations of various distribution functions.
@\label{lst:gsl-samplers-cf}@double gsl_cdf_gaussian_P    (double x, double sigma); // Cumulative Dist. Function (CDF)
@\label{lst:gsl-samplers-sf}@double gsl_cdf_gaussian_Q    (double x, double sigma); // Survival Function (SF)
@\label{lst:gsl-samplers-qf}@double gsl_cdf_gaussian_P_inv(double u, double sigma); // Quantile Function (QF)

// THIS WORK: Exact random variate generators given a numerical CDF and/or SF specification.
@\label{lst:gsl-samplers-s0}@GENERATE_FROM_CDF (gsl_cdf_gaussian_P, 5.0);
@\label{lst:gsl-samplers-s1}@GENERATE_FROM_SF  (gsl_cdf_gaussian_Q, 5.0);
@\label{lst:gsl-samplers-s2}@GENERATE_FROM_DDF (gsl_cdf_gaussian_P, gsl_cdf_gaussian_Q, 5.0);
\end{lstlisting}
\vspace{-.5cm}
\end{listing}

These generators do not use ad-hoc floating-point approximations of random uniform
reals (e.g., \texttt{uniform} or \texttt{uniform\_pos} in
\cref{lst:exponential-u0,lst:exponential-u1} of \cref{lst:exponential})
as the primitive unit of entropy.
Instead, they operate directly over individual random bits from streams
such as \texttt{rand()} (which provides pseudorandom bits) or
\texttt{/dev/urandom} (which provides cryptographically secure random
bits).
For example, using our framework, implementing a \texttt{uniform} that
guarantees 100\% coverage of all floats in $(0,1]$
or $[0, 1)$ (cf.~\cref{footnote:pytorch-uniform};
\cref{table:float-coverage} in \cref{appx:overview})
is just a matter of specifying the CDF:\\
\begin{minipage}{.5\linewidth}
\begin{lstlisting}[style=CC,basicstyle=\ttfamily\scriptsize,numbers=none,frame=none]
// Uniform over all doubles in (0,1]
double cdf_uniform_round_up(double x) {
  if (isnan(x)) { return 1 }
  return min(max(x, 0), 1);
  }
GENERATE_FROM_CDF(cdf_uniform_round_up);
\end{lstlisting}
\end{minipage}\hfill
\begin{minipage}{.5\linewidth}
\begin{lstlisting}[style=CC,basicstyle=\ttfamily\scriptsize,numbers=none,frame=none]
// Uniform over all doubles in [0,1)
double cdf_uniform_round_dn(double x) {
  double z = nextafter(x, INFINITY);
  return cdf_uniform_round_up(z);
  }
GENERATE_FROM_CDF(cdf_uniform_round_dn);
\end{lstlisting}
\end{minipage}\\
These functions specify the exact distributions of a random variate obtained
by correctly rounding an infinitely precise real $U\sim\mathrm{Uniform}([0,1])$
to the next and previous IEEE-754 double-precision float, respectively.
While highly specialized random variate generation algorithms for these challenging types
of uniform distributions over floats exist in the
literature~\citep{grabowski2015,walker1974,campbell2014,Downey2007,Goualard2022},
using our approach, exact generators are automatically
obtained by specifying a desired CDF, which is often
more straightforward than implementing the generator and extends
directly to nonuniform distributions.

The generators described in this work are also \textit{entropy-optimal}
(\cref{theorem:sampler-opt}), i.e., they consume the
information-theoretically minimal number of random bits from the entropy
source to generate an output.
They incur no additional floating-point error with respect to the CDF or SF
implementation, and avoid arbitrary-precision arithmetic
(\cref{prop:exactsubtract1}).
It is also straightforward to compute exact quantiles (\cref{alg:quantile})
of these generators without using approximate implementations
(e.g., \texttt{exp\_qf} on \cref{lst:exponential-qf} of \cref{lst:exponential})
that have no theoretical relationship to the implemented generator.

\Cref{lst:gsl-samplers} shows how our random variate generation algorithms
(\crefrange{lst:gsl-samplers-s0}{lst:gsl-samplers-s2})
interoperate with existing software such as the GNU Scientific
Library (GSL) by reusing high-quality CDF and/or SF implementations
(\crefrange{lst:gsl-samplers-cf}{lst:gsl-samplers-sf}).
The built-in GSL Gaussian generators
(\crefrange{lst:gsl-samplers-r1}{lst:gsl-samplers-r4}) often have complex
implementations spanning hundreds of lines of code, and each specify
different output distributions which are all intractable to estimate.
Indeed, any GSL random variate generator that makes just two
(or more) calls to \texttt{uniform} is already intractable to analyze.
In contrast, our generators
(\crefrange{lst:gsl-samplers-s0}{lst:gsl-samplers-s2}) have known output
distributions that match the specified CDF or SF.
They are also automatically synthesized from these specifications,
requiring no additional implementation code.
The view taken by our approach is to first develop high-quality CDF and/or
SF implementations that provide transparent formal specifications of the
desired distribution the random variates should follow.
This specification can be carefully debugged for numerical problems or
other errors (e.g., those in \cref{appx:survey}), after which we automatically
synthesize exact generators that match the specification.

\section{Preliminaries}
\label{sec:prelim}
%!TEX root=./paper.tex

\begin{definition}
\label{def:random-variable}
A \textit{random variable} $X: [0,1] \to \real$ is a Borel measurable map
from the unit interval into the reals.
We may equivalently write $X: \set{0,1}^\nat \to \real$,
where $X(u_1u_2\dots) \defas X\left(\sum_{i=1}^\infty u_i2^{-i}\right)$.
The \textit{distribution} of a random variable $X$ is the probability measure
$\mu_X$ over $(\real, \mathcal{B}(\real))$ such that $\mu_X(A) \defas \lambda(X^{-1}(A)) \asdef \Pr(X \in A)$,
where $A \in \mathcal{B}(\real)$ is an event and $\lambda$ the Lebesgue measure.
\end{definition}

\begin{remark}
\label{remark:random-variable-domain}
Defining the domain of a random variable to be $[0,1]$ (instead of say
$\cup_{n=0}^{\infty}[0,1]^n$) is made without loss of generality.
A single uniformly distributed real $\omega \defas
(0.\omega_1\omega_2\ldots)_2 \in [0,1]$ can be ``split'' into a
countably infinite number of i.i.d.~uniform numbers
$(u_1(\omega), u_2(\omega), \ldots)$ where
$u_n(\omega) \defas (0. \omega_{p_n}\omega_{p^2_n}\omega_{p^3_n}\dots)_2$ for
$n \ge 1$, with $p_1 < p_2 < p_3 < \ldots$ an enumeration of the primes.
\end{remark}

\begin{definition}
\label{def:cdf}
A \textit{cumulative distribution function} (CDF)
$F: \real \to [0,1]$ is a right-continuous monotone mapping such that
% \begin{align}
$\lim_{x\to-\infty} F(x) = 0$
and
$\lim_{x\to\infty} F(x) = 1$.
A \textit{survival function} (SF) $S: \real \to [0,1]$ is a mapping such that
$S(x) = 1 - F(x)$ for some CDF $F$.
A \textit{quantile function} (QF) $Q: [0,1] \to \real$ is a mapping such
that $Q(u) = \inf\set{x \in \real \mid u \le F(x)}$ for some CDF $F$.
\end{definition}

\begin{proposition}[{\citet[Theorems 12.4, 14.1]{Billingsley1986}}]
\label{remark:billingsley}
Let $([0,1], \mathcal{B}([0,1]), \lambda)$ be the standard probability space.
Each random variable $X: [0,1] \to \real$ has a unique CDF $F_X \defas \Pr(X \le x) = \lambda(X^{-1}(-\infty,x])$.
Each CDF $F: \real \to [0,1]$ is in 1-1 correspondence with a unique probability measure $\mu_F$ on
$(\real, \mathcal{B}(\real))$.
There exist infinitely many random variables
$X_F$ on $([0,1], \mathcal{B}([0,1]), \lambda)$
that are identically distributed as $\mu_F$, e.g.,
$X_F(\omega) = Q_F(U(\omega))$ where $U \sim \mathrm{Uniform}([0,1])$.
\end{proposition}

For a probability measure $\mu$ over $\real$, the CDF representation
$F_\mu$ is ideal for \textit{specification}: it
concisely describes $\mu$ in terms of a univariate real function.
On the other hand, a random variable representation $X_\mu$ is ideal for
\textit{generation}: it describes a procedure for producing
$\mu$-distributed random variates.
We next define random-variate generators, which are computable,
finite-resource analogues of random variables that do not require infinite
memory or computation over the reals.

\begin{definition}
\label{definition:random-variate-generator}
A \textit{random variate generator} $X: \set{0,1}^* \rightharpoonup \set{0,1}^n$ is a partial map
from the set of all finite-length binary strings to $n$-bit binary strings such
that the following two conditions hold:
\begin{align}
\mbox{\underline{Prefix-Free}:}\ u \in \dom(X) \implies u\set{0,1}^+ \notin \mathrm{dom}(X);
&&
\mbox{\underline{Exhaustive}:}\
\textstyle\sum_{u \in \dom(X)}2^{-\abs{u}} = 1.
\rlap{\;\qedsymbol}
\label{eq:isothujone}
\end{align}\noqed
\end{definition}

\begin{proposition}[name=,restate=LiftRVG]
\label{proposition:lift-random-sampler}
A random variate generator $X$ describes a random variable
$X_\gamma: [0,1] \to \mathbb{R}$ over at most $2^n$ values,
where $\Pr(X_\gamma=x) = \sum_{u \in \dom(X)} 2^{-\abs{u}}\mathbf{1}[\gamma(X(u)) = x]$
for any $\gamma: \set{0,1}^n \to \real$.
\end{proposition}

\bgroup
\setlength{\textfloatsep}{0pt}
%!TEX root=./paper.tex

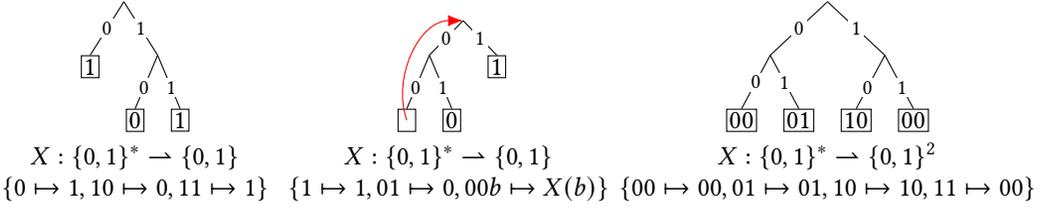
\begin{figure}[t]
\tikzset{level distance=20pt, sibling distance=10pt}
\tikzset{every tree node/.style={anchor=north}}
\tikzstyle{leaf}=[inner sep=1pt,draw]
\tikzstyle{branch}=[font=\footnotesize,fill=white,inner sep=1pt]
\begin{minipage}[b]{.275\linewidth}
\centering
\begin{tikzpicture}
\Tree[
  \edge node[branch,pos=.5]{0}; \node[leaf]{1};
  \edge node[branch,pos=.5]{1} ;
    [
    \edge node[branch,pos=.6]{0}; \node[leaf]{0};
    \edge node[branch,pos=.6]{1}; \node[leaf]{1};
    ] ]
\end{tikzpicture}

$X: \set{0,1}^* \rightharpoonup \set{0,1}$\\
$\set{0 \mapsto 1, 10 \mapsto 0, 11 \mapsto 1}$
\end{minipage}\hfill
\begin{minipage}[b]{.325\linewidth}
\centering
\begin{tikzpicture}
\Tree[.\node[name=root]{};
  \edge node[pos=.5,branch]{0};
    [
     \edge node[pos=.6,branch]{0}; \node[leaf,draw,name=back]{\phantom{0}};
     \edge node[pos=.6,branch]{1}; \node[leaf,draw]{0};
     ]
  \edge node[pos=.5,branch]{1}; \node[leaf,draw]{1};
  ]
\draw[red,-{Latex[length=2mm]},out=110,in=180] (back.center) to (root.south); % -latex
\end{tikzpicture}

$X: \set{0,1}^* \rightharpoonup \set{0,1}$\\
$\set{1 \mapsto 1, 01 \mapsto 0, 00b \mapsto X(b)}$
\end{minipage}\hfill
\begin{minipage}[b]{.4\linewidth}
\centering
\begin{tikzpicture}
\Tree[
  \edge node[pos=.5,branch]{0};
    [
      \edge node[pos=.5,branch]{0}; \node[leaf,draw]{00};
      \edge node[pos=.5,branch]{1}; \node[leaf,draw]{01};
    ]
  \edge node[pos=.5,branch]{1};
    [
    \edge node[pos=.6,branch]{0}; \node[leaf,draw]{10};
    \edge node[pos=.6,branch]{1}; \node[leaf,draw]{00};
    ] ]
\end{tikzpicture}

$X: \set{0,1}^* \rightharpoonup \set{0,1}^2$\\
$\set{00 \mapsto 00, 01 \mapsto 01, 10 \mapsto 10, 11 \mapsto 00}$
\end{minipage}
\captionsetup{skip=5pt}
\caption{Three random variate generators represented as discrete distribution generating (DDG) trees.
The 0/1 labels along the edges are omitted from the tree diagrams going forward.}
\label{fig:ddg-tree-examples}
\end{figure}

\egroup

\Cref{definition:random-variate-generator} corresponds to the
``discrete distribution generating'' (DDG) trees from \citet[Section 2]{knuth1976}.
Any random variate generator can be drawn as a tree where leaves have labels in $\set{0,1}^n$
and branches correspond to the decision on the next $0$ or $1$ input bit (\cref{fig:ddg-tree-examples}).

\begin{definition}
The \textit{entropy cost} of a random variate generator $X$ is a
random variable $C$ over $\mathbb{N}$ measuring the number of bits consumed, with
$\Pr(C = c) = \textstyle\sum_{u \in \dom(X)} \mathbf{1}[\abs{u}=c]2^{-c}$
for $c \ge 0$.
\end{definition}

\begin{definition}
A \textit{concise binary expansion} of a real number $x$ is a binary representation
that does not end in an infinite string of 1s.
Binary expansions hereon are always concise.
\end{definition}

\begin{theorem}[\Citet{knuth1976}]
\label{theorem:knuth-yao}
Let $p \defas \set{\ell_1 \mapsto p_1, \dots, \ell_m \mapsto p_m}$ be a discrete probability
distribution over $m \ge 1$ outcomes $\ell_1,\dots,\ell_m$.
Write the binary expansions as $p_i = (p_{i0}.p_{i1}p_{i2}\dots)_2$ for
$i=1,\dots,m$.
A random variate generator $X$ for $p$ has minimal expected
entropy cost $\mathbb{E}[C]$ (i.e., it is ``entropy-optimal'') if and
only if its DDG tree contains exactly $p_{ij}$
leaf nodes labeled $\ell_i$ at depth $j\ge0$.
Further, $H(p) \le \mathbb{E}[C] < H(p)+2$, where
$H(p)\defas \sum_{i=1}^{m}- p_i \log_2 (p_i)$ is the Shannon entropy of $p$.
\end{theorem}

\begin{example}
  In \cref{fig:ddg-tree-examples}, the first two DDG trees show entropy-optimal generators
  for $p = \set{0 \mapsto (0.01)_2, 1 \mapsto (0.11)_2}$
  and $p = \set{0 \mapsto (0.\overline{01})_2, 1 \mapsto (0.\overline{10})_2}$.
  The third DDG tree is an entropy-suboptimal generator for
  $p = \set{00 \mapsto (0.10)_2, 01 \mapsto (0.01)_2, 10 \mapsto (0.01)_2}$.
\end{example}

\section{Exact Random Variate Generators for Binary-Coded Probability Distributions}
\label{sec:binary}
%!TEX root=./paper.tex

Since the size of a DDG tree is lower bounded by the number $m$ of discrete
outcomes, explicitly constructing entropy-optimal DDG trees using
\cref{theorem:knuth-yao} is computationally intractable whenever $m$ is
combinatorially large.
For example, a discrete distribution with full support
over all IEEE-754 double-precision floats has $m = 2^{64}$ outcomes.
To address this challenge, we introduce
\textit{binary-coded probability distributions}, which are a universal
mathematical representation for lazily describing any computable
probability measure over the reals.
This powerful abstraction lets us develop entropy-optimal random variate
generation algorithms that avoid the combinatorial explosion in the DDG tree size.
The ``idealized'' algorithms in this section will be specialized in
\cref{sec:floating,sec:survival} to obtain efficient software
implementations of random variate generators on finite-precision computers.

%!TEX root=./paper.tex

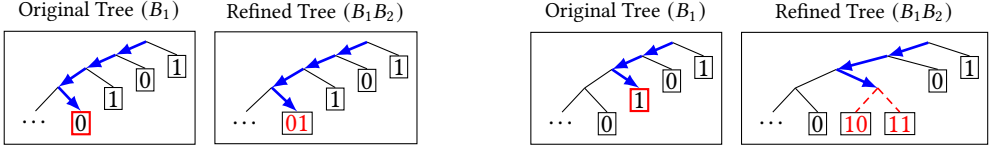
\begin{figure}[t]
\centering
\definecolor{mygreen}{RGB}{0,160,0}
\tikzset{level 1/.style={level distance=5pt}}
\tikzset{level 2/.style={level distance=5pt}}
\tikzset{level 3/.style={level distance=7pt}} % 7.5pt
\tikzset{level 4/.style={level distance=9pt}} % 12.5pt
\tikzset{sibling distance=5pt}
\tikzset{every tree node/.style={anchor=north}}

\tikzstyle{branch}=[shape=coordinate]
\tikzstyle{excE}=[color=blue, line width=1pt,-latex]    % executed edge
\tikzstyle{refE}=[color=red,  line width=0.6pt,densely dashed] % refined  edge
\tikzstyle{curN}=[color=black,draw=black,inner sep=1pt] % current  node
\tikzstyle{excN}=[color=black,draw=red,  inner sep=1pt, thick] % executed node
\tikzstyle{refN}=[color=red,  draw=black,inner sep=1pt] % refined node

\begin{subfigure}[b]{.45\linewidth}
\centering
\begin{tikzpicture}
\node[name=x,draw,label={[font=\footnotesize]above:Original Tree $(B_1)$}]{
  \begin{tikzpicture}
  \Tree
  [.\node[branch,name=root]{};
    \edge[excE];
    [
      \edge[excE];
      [ \edge[excE];[ [.\node[name=R]{\dots};] \edge[excE]; \node[excN]{0}; ] \node[curN]{1}; ] \node[curN]{0}; ]
    \node[curN]{1};
  ]
  \end{tikzpicture}
};

\node[name=y,draw, right=0.25 of x, label={[font=\footnotesize]above: Refined Tree $(B_1B_2)$}]{
  \begin{tikzpicture}
  \Tree
  [.\node[branch,name=root]{};
    \edge[excE];
    [
      \edge[excE];
      [ \edge[excE];[ [.\node[name=R]{\dots};] \edge[excE];\node[refN]{01}; ] \node[curN]{1}; ] \node[curN]{0}; ]
    \node[curN]{1};
  ]
  \end{tikzpicture}
};

\end{tikzpicture}
\caption{Refining a node into a leaf. After refinement, 01 is returned using zero additional flips.}
\label{fig:ddg-tree-refine-leaf}
\end{subfigure}\hfill
\begin{subfigure}[b]{.5\linewidth}
\centering
\begin{tikzpicture}
\node[name=x,draw,label={[font=\footnotesize]above:Original Tree $(B_1)$}]{
  \begin{tikzpicture}
  \Tree
  [.\node[branch,name=root]{};
    \edge[excE];
    [
      \edge[excE];
      [ [ [.\node[name=R]{\dots};] \node[curN]{0}; ] \edge[excE];\node[excN]{1}; ]
      \node[curN]{0}; ]
    \node[curN]{1};
  ]
  \end{tikzpicture}
};

\node[name=y, right=.25 of x, draw, label={[font=\footnotesize]above: Refined Tree $(B_1B_2)$}]{
  \begin{tikzpicture}
  \Tree
  [.\node[branch,name=root]{};
    \edge[excE];
    [
      \edge[excE];
      [
        [ [.\node[name=R]{\dots};] \node[curN]{0}; ]
        \edge[excE];
        [.\node[branch]{}; % \node[curN,draw=red,thick]{1};
          \edge[refE]; \node[refN]{10};
          \edge[refE]; \node[refN]{11}; ] ]
      \node[curN]{0};
      ]
    \node[curN]{1};
  ]
  \end{tikzpicture}
};

\end{tikzpicture}
\caption{Refining a node into a subtree.
After refinement, 10 or 11 is returned using one additional flip.}
\label{fig:ddg-tree-refine-subtree}
\end{subfigure}
\captionsetup{skip=4pt}
\caption{Dynamically refining the leaves of an optimal DDG tree for lazy random variate generation.}
\label{fig:ddg-tree-refine}
\vspace{-5pt}
\end{figure}

\begin{definition}
\label{definition:bcpd}
\label{eq:hemophobia}
A \textit{binary-coded probability distribution}
is a map $p: \bool^* \to [0,1]$ such that
$p(\varepsilon)=1$
and
$p(b_1\ldots{b_j}) = p(b_1\ldots{b_j}0) + p(b_1\ldots{b_j}1)$
for all $j \ge 1$ and $b_1,\dots,b_j \in \bool$.
\end{definition}

A binary-coded probability distribution $p$
defines a family of discrete probability distributions
\begin{equation}
p_n \,{\defas}\, \set{b \,{\mapsto}\, p(b) \,{\mid}\, b \,{\in}\, \bool^n}
\end{equation}
over $\bool^n$ ($n \ge 0$).
For example,
$p(00) = 0.5; p(01) = 0.2;
p(10) = 0.1; p(11) = 0.2$ defines the distributions
$p_1\defas\set{0\mapsto 0.7, 1 \mapsto 0.3}$,
$p_2\defas\set{00\mapsto 0.5, 01\mapsto0.2, 10 \mapsto 0.1, 11\mapsto0.2}$.

\setlength{\intextsep}{0pt}%
\setlength{\columnsep}{8pt}%
\begin{wrapfigure}{r}{0.35\textwidth}
%!TEX root=./paper.tex

\begin{adjustbox}{width=\linewidth}
\hfill\begin{tikzpicture}[yscale=3, xscale=4,thick]

\tikzset{lbl/.style={inner sep=0.25pt,pos=0.5,above,font=\footnotesize}}

\draw (0,0) -- (1,0) -- (1,1) -- (0,1) -- cycle;

\draw (-0.025, 0) -- (0.025,0) node[lbl,pos=0,left]{$0$};
\draw (-0.025, 1) -- (0.025,1) node[lbl,pos=0,left]{$1$};
\draw[name path=cdf] (0,0) .. controls (.292,.809) and (.892,.234) .. (1,1) node[pos=0.9,left,font=\footnotesize]{CDF $F$};

\foreach \i\istr in {0/{4}, 1/{5}, 2/{6}, 3/{7}, 4/{8}, 5/{9}, 6/{10}, 7/{11}, 8/{12}} {
  \draw[] (\i/8, -0.025) -- (\i/8, +0.025) node[font=\normalsize,anchor=north,pos=0,inner sep=1pt]{$\frac{\istr}{4}$};
  \draw[name path=path-\i,draw=none] (\i/8, 0) -- (\i/8, 1);
  \path[name intersections={of=cdf and path-\i,by=i\i}] node[name=n\i, at=(i\i),circle,fill=black, outer sep=0pt, inner sep=1.25pt]{};
}
\begin{scope}[on background layer]
\foreach \i in {1,...,7}{
  \draw[color=red,thick] (\i/8,0) -- (n\i) -- (n\i -| {(0,0)});
}
\end{scope}

\draw[decorate,decoration={brace,amplitude=5pt,raise=.5ex}]
  ([yshift=-.5pt]0,1) -- ([yshift=+.5pt]n7 -| {(0,0)})
  node[pos=0.5,right,font=\footnotesize,xshift=6pt]{$p(111)$};

\def\ys{0.28}
\def\tick{0.015 cm}

\draw[] (0,-\ys) -- (1,-\ys) node[lbl]{$\epsilon$};
\foreach \i in {0,1} {
  \draw[] ([yshift=-\tick]\i,-\ys) -- ([yshift=\tick]\i,-\ys);
}

\draw[yshift=-0.1cm] (0,-\ys) -- (0.500,-\ys) node[lbl]{$0$};
\draw[yshift=-0.1cm] (0.500,-\ys) -- (1,-\ys) node[lbl]{$1$};
\foreach \i in {0,0.500,0.500,1} {
  \draw[yshift=-0.1cm] ([yshift=-\tick]\i,-\ys) -- ([yshift=\tick]\i,-\ys);
}

\draw[yshift=-0.2cm] (0,-\ys) -- (0.250,-\ys) node[lbl]{$00$};
\draw[yshift=-0.2cm] (0.250,-\ys) -- (0.500,-\ys) node[lbl]{$01$};
\draw[yshift=-0.2cm] (0.500,-\ys) -- (0.750,-\ys) node[lbl]{$10$};
\draw[yshift=-0.2cm] (0.750,-\ys) -- (1,-\ys) node[lbl]{$11$};
\foreach \i in {0,0.250,0.500,0.500,0.750,1.} {
  \draw[yshift=-0.2cm] ([yshift=-\tick]\i,-\ys) -- ([yshift=\tick]\i,-\ys);
}

\draw[yshift=-0.3cm] (0,-\ys) -- (0.125,-\ys) node[lbl]{$000$} node[below]{$\dots$};
\draw[yshift=-0.3cm] (0.125,-\ys) -- (0.250,-\ys) node[lbl]{$001$} node[below]{$\dots$};
\draw[yshift=-0.3cm] (0.250,-\ys) -- (0.375,-\ys) node[lbl]{$010$} node[below]{$\dots$};
\draw[yshift=-0.3cm] (0.375,-\ys) -- (0.500,-\ys) node[lbl]{$011$} node[below]{$\dots$};
\draw[yshift=-0.3cm] (0.500,-\ys) -- (0.625,-\ys) node[lbl]{$100$} node[below]{$\dots$};
\draw[yshift=-0.3cm] (0.625,-\ys) -- (0.750,-\ys) node[lbl]{$101$} node[below]{$\dots$};
\draw[yshift=-0.3cm] (0.750,-\ys) -- (0.875,-\ys) node[lbl]{$110$} node[below]{$\dots$};
\draw[yshift=-0.3cm] (0.875,-\ys) -- (1,-\ys) node[lbl]{$111$};
\foreach \i in {0, 1, 2, 3, 4, 5, 6, 7, 8} {
  \draw[yshift=-0.3cm] ([yshift=-\tick]\i/8,-\ys) -- ([yshift=\tick]\i/8,-\ys);
}
\end{tikzpicture}
\end{adjustbox}
\vspace{-.8cm}
\end{wrapfigure}

Binary-coded probability distributions also correspond to
probability measures over $\set{0,1}^\nat \equiv \real$
\citep[Lemma~7.1.2]{bogachev2007}.
The plot to the right shows this idea for a CDF $F$ over $[1,3]$.
Each binary string $b$ defines an interval $[x_1(b), \allowbreak x_2(b)] \subset [1,3]$,
where $p(b) = F(x_2(b)) - F(x_1(b))$ is its probability under $F$, e.g., $p(01) = F(8/4)-F(6/4)$.
The same idea holds for unbounded domains,
by using binary-coded partitions of $\real$~\citep[\S3]{knuth1976}.

\subsection{Entropy-Suboptimal Generation}
\label{sec:binary-sampling-suboptimal}

Before considering optimal algorithms, an
entropy-suboptimal baseline for generating a random string $(B_1, B_2, \dots)\,{\sim}\,p$
is \textit{conditional bit sampling}~\citep[\S{II.B}]{Sobolewski1972}.
This method generates one bit at a time by using the chain rule of probability,
i.e.,
$B_1\,{\sim}\,\mathrm{Bernoulli}(p(1))$,
$B_2{\mid}{B_1}\,{\sim}\,\mathrm{Bernoulli}(p(B_11)/p(B_1))$, etc.
However, as generating each $B_n$ requires roughly two random bits in
expectation (\Cref{appx:naive-baseline-binary}), conditional bit sampling is
entropy-inefficient, consuming roughly $2n$ more bits than an entropy-optimal
sampler for generating a length-$n$ bit string in the worst case
(\cref{proposition:opt-naive-bound}). % cf.~
A second challenge is the high computational cost of computing
the conditional probabilities during generation.
The conditional bit sampling baseline is presented and analyzed in
\cref{appx:naive-baseline}; and evaluated in \cref{sec:evaluation}.

%!TEX root=./paper.tex

\begin{figure}[p]

\newcommand{\hhl}[2]{{\sethlcolor{#1}\hl{#2}}}
\newcommand{\hlp}[1]{\hhl{pink}{#1}}
\newcommand{\hlg}[1]{\hhl{Goldenrod}{#1}}
\newcommand{\hlb}[1]{\hhl{cyan!50!white}{#1}}
\newcommand{\gr}[1]{{\textcolor{gray}{#1}}}
\newcommand{\tm}[1]{\tikz[overlay, remember picture, baseline=(#1.south)] \node[name=#1,rectangle,draw] {};}

\setlength{\tabcolsep}{4pt}
\setlength{\intextsep}{0pt}

\begin{minipage}[t]{.46\linewidth}
\begin{algorithm}[H]
\captionsetup{hypcap=false}
\caption{Optimal Generation}
\label{alg:sampler-opt}
\algrenewcommand\algorithmicindent{1.0em}%
\begin{algorithmic}[1]
\Require{%
  Binary-coded probability distribution $p: \set{0,1}^* \to [0,1]$,
  cf.~\cref{definition:bcpd} \\
  \color{gray}{String $b \in \bool^*$ generated so far}\\
  \color{gray}{\#Flips $\ell \ge 0$ consumed so far}
  }
\Ensure{Random bitstream $b$ drawn from $p$,
  i.e., $b_1\dots{b_n}\sim{p}_n$ $(n \ge 0)$
}
\Function{\SampleOpt}{$p$, $b=\varepsilon$, $\ell=0$}
  \If{$[p({b0})]_{\ell} = 1 \wedge [p({b1})]_\ell = 0$}
    \Comment{Leaf}
    \label{algline:UnivSampler-If1-Start}
    \State \Return $\SampleOpt(p, b0, \ell)$
    \Comment{0}%
    \label{algline:UnivSampler-Print0-NoIter}
  \EndIf
  \If{$[p(b0)]_{\ell} = 0 \wedge [p({b1})]_\ell = 1$}
    \Comment{Leaf}
    \State \Return $\SampleOpt(p, b1, \ell)$
    \Comment{1}%
    \label{algline:UnivSampler-Print1-NoIter}
  \EndIf
  \While{$\textbf{true}$} \label{algline:UnivSampler-Loop}
    \Comment{Refine Subtree}
    \State $x \gets \Flip()$; $\ell \gets \ell+1$
    \If{$x = 0 \wedge [p(b0)]_{\ell} = 1$}
      \Comment{Leaf}
      \State \Return $\SampleOpt(p, b0, \ell)$
      \Comment{0}%
      \label{algline:UnivSampler-Print0}
    \EndIf
    \If{$x = 1 \wedge [p(b1)]_{\ell} = 1$}
      \Comment{Leaf}
      \State \Return $\SampleOpt(p, b1, \ell)$
      \Comment{1}%
      \label{algline:UnivSampler-Print1}
    \EndIf
  \EndWhile
\EndFunction
\end{algorithmic}
\end{algorithm}
\end{minipage}\hfill
\addtocounter{figure}{1}
\addtocounter{subfigure}{-2} % WL: This is an ad-hoc fix to the weird caption numbering.
\begin{subfigure}[t]{.51\linewidth}
\centering
\captionsetup{aboveskip=0pt,belowskip=0pt}
\caption{A binary-coded probability distribution $p$ unrolled to four bits,
giving a discrete distribution over $\set{0,1}^4$.}
\label{fig:optimal-sampler-dist}
\begin{adjustbox}{max width=\linewidth}
\begin{tikzpicture}
\node[]{$
\begin{NiceMatrix}[l]
 0000 \mapsto \frac{6}{137}
&0001 \mapsto \frac{12}{137}
&0010 \mapsto \frac{13}{137}
&0011 \mapsto \frac{9}{137}
\\[2pt]
 0100 \mapsto \frac{10}{137}
&0101 \mapsto \frac{12}{137}
&0110 \mapsto \frac{6}{137}
&0111 \mapsto \frac{1}{137}
\\[2pt]
 1000 \mapsto \frac{1}{137}
&1001 \mapsto \frac{2}{137}
&1010 \mapsto \frac{13}{137}
&1011 \mapsto \frac{8}{137}
\\[2pt]
 1100 \mapsto \frac{14}{137}
&1101 \mapsto \frac{13}{137}
&1110 \mapsto \frac{7}{137}
&1111 \mapsto \frac{10}{137}
\CodeAfter
  \SubMatrix\lbrace{1-1}{4-4}\rbrace
\end{NiceMatrix}$};
\end{tikzpicture}
\end{adjustbox}
\centering
\captionsetup{aboveskip=4pt,belowskip=2pt}
\caption{A trace of \cref{alg:sampler-opt} on the distribution $p$ from \subref{fig:optimal-sampler-dist}.}
\label{fig:optimal-sampler-trace}
\begin{adjustbox}{valign=t,max width=\linewidth}
\begin{tabular}{|c||p{1.25cm}@{=\,}p{.5cm}@{\,=\,}l@{\,}lllll|c@{}}
\cline{5-9}
\multicolumn{1}{c}{Recur.} & \multicolumn{3}{c|}{~} & \multicolumn{5}{c|}{\Flip\,$x$} & ~
\\
\multicolumn{1}{c}{Level}  & \multicolumn{3}{c|}{Probabilities}            & 1      & 0       & \hlb{1}    & 1       & \hlb{0}    & Output $b$ \\ \hline\hline
\multirow{2}{*}{0} & $p(0)$      & $\frac{69}{137}$ & \bf{0.} & \bf{1} & \bf{0} & \bf{0}     & \gr{0} & \gr{0}     & ~ \Tstrut \\
~                  & $p(1)$      & $\frac{68}{137}$ & \bf{0.} & \bf{0} & \bf{1} & \hlg{\bf1} & \gr{1} & \gr{1}     & \hlp{1} \Tstrut\Bstrut \\ \cline{1-9}
\multirow{2}{*}{1} & $p({10})$   & $\frac{24}{137}$ & \gr{0.} & \gr{0} & \gr{0} & \hlg{\bf1} & \gr{0} & \gr{1}     & \hlp{0} \Tstrut \\
~                  & $p({11})$   & $\frac{44}{137}$ & \gr{0.} & \gr{0} & \gr{1} & \bf{0}     & \gr{1} & \gr{0}     & ~ \Tstrut\Bstrut \\ \cline{1-9}
\multirow{2}{*}{2} & $p({100})$  & $\frac{3}{137}$  & \gr{0.} & \gr{0} & \gr{0} & \bf{0}     & \gr{0} & \gr{0}     & ~ \Tstrut \\
~                  & $p({101})$  & $\frac{21}{137}$ & \gr{0.} & \gr{0} & \gr{0} & \hlg{\bf1} & \gr{0} & \gr{0}     & \hlp{1} \Tstrut\Bstrut \\\cline{1-9}
\multirow{2}{*}{3} & $p({1010})$ & $\frac{13}{137}$ & \gr{0.} & \gr{0} & \gr{0} & \bf{0}     & \bf{1} & \hlg{\bf1} & \hlp{0} \Tstrut \\
~                  & $p({1011})$ & $\frac{8}{137}$  & \gr{0.} & \gr{0} & \gr{0} & \bf{0}     & \bf{0} & \bf{1}     & ~ \Tstrut\Bstrut \\\hline
\end{tabular}
\end{adjustbox}
\end{subfigure}

\vspace{4pt}
%!TEX root=./paper.tex

\begin{subfigure}{\linewidth}
\centering
\definecolor{mygreen}{RGB}{0,160,0}
\tikzset{level distance=8pt, sibling distance=5pt}
\tikzset{every tree node/.style={anchor=north}}

\tikzstyle{title}=[color=black,font=\normalsize,label distance=-5pt] % \normalsize
\tikzstyle{branch}=[shape=coordinate]
\tikzstyle{curE}=[color=black] % current  edge
\tikzstyle{prvE}=[color=gray ] % previous edge
\tikzstyle{excE}=[color=blue, line width=1pt,-latex]    % executed edge
\tikzstyle{refE}=[color=red,  line width=0.6pt,densely dashed] % refined  edge
\tikzstyle{curN}=[color=black,draw=black,inner sep=1pt] % current  node
\tikzstyle{prvN}=[color=gray, draw=gray, inner sep=1pt] % previous node
\tikzstyle{excN}=[color=black,draw=red,  inner sep=1pt, thick] % executed node
\tikzstyle{refN}=[color=red,  draw=black,inner sep=1pt] % refined node
\tikzstyle{curNnoline}=[color=black]
\tikzstyle{prvNnoline}=[color=gray ]

\begin{adjustbox}{max width=.8\linewidth}
\begin{tikzpicture}

%%%%%%%%%%%%%%%% Level 0 %%%%%%%%%%%%%%%%%%
\node[name=t0,label={[title]above:Original Tree ($B_1$)}]{
  \begin{tikzpicture}[remember picture]
  \Tree
  [.\node[branch,name=r0]{};
    \edge[curE]; \node[curN]{0};
    \edge[excE]; [
      \edge[excE]; [
        \edge[curE]; [
          \edge[curE]; [
            \edge[curE]; \node[curNnoline]{\dots};
            \edge[curE]; \node[curN]{1};
          ]
          \edge[curE]; \node[curN]{1};
        ]
        \edge[excE]; \node[excN]{1};
      ]
      \edge[curE]; \node[curN]{1};
    ]
  ]
  \end{tikzpicture}
};

\node[name=t0b,right=2.5 of t0,label={[title]above:Refined Tree ($B_1 B_2$)}]{ %  when $B_1 = 1$
  \begin{tikzpicture}[remember picture]
  \Tree
  [.\node[branch,name=r0b]{};
    \edge[curE]; \node[curN]{0};
    \edge[curE]; [
      \edge[curE]; [
        \edge[curE]; [
          \edge[curE]; [
            \edge[curE]; \node[curNnoline]{\dots};
            \edge[curE]; \node[curN]{1};
          ]
          \edge[curE]; \node[curN]{1};
        ]
        \edge[curE]; \node[refN]{10};
      ]
      \edge[curE]; \node[curN]{1};
    ]
  ]
  \end{tikzpicture}
};

%%%%%%%%%%%%%%%% Level 1 %%%%%%%%%%%%%%%%%%
\node[name=t1,below=0.675 of t0,label={[title]above:Original Tree $(B_1 B_2)$}]{
  \begin{tikzpicture}[remember picture]
  \Tree
  [.\node[branch,name=r0]{};
    \edge[prvE]; \node[prvN]{0};
    \edge[prvE]; [
      \edge[prvE]; [
        \edge[prvE]; [
          \edge[prvE]; [
            \edge[prvE]; \node[prvNnoline]{\dots};
            \edge[prvE]; \node[prvN]{1};
          ]
          \edge[prvE]; \node[prvN]{1};
        ]
        \edge[prvE]; \node[excN]{10};
      ]
      \edge[prvE]; \node[prvN]{1};
    ]
  ]
  \end{tikzpicture}
};

\node[name=t1b,at=(t1 -| t0b),label={[title]above:Refined Tree $(B_1 B_2 B_3)$}]{ % when $B_1 B_2 = 10$
  \begin{tikzpicture}[remember picture]
  \Tree
  [.\node[branch,name=r0]{};
    \edge[prvE]; \node[prvN]{0};
    \edge[prvE]; [
      \edge[prvE]; [
        \edge[prvE]; [
          \edge[prvE]; [
            \edge[prvE]; \node[prvNnoline]{\dots};
            \edge[prvE]; \node[prvN]{1};
          ]
          \edge[prvE]; \node[prvN]{1};
        ]
        \edge[prvE]; \node[refN]{101};
      ]
      \edge[prvE]; \node[prvN]{1};
    ]
  ]
  \end{tikzpicture}
};

%%%%%%%%%%%%%%%% Level 2 %%%%%%%%%%%%%%%%%%
\node[name=t2,below=0.675 of t1,label={[title]above:Original Tree $(B_1 B_2 B_3)$}]{
  \begin{tikzpicture}[remember picture]
  \Tree
  [.\node[branch,name=r0]{};
    \edge[prvE]; \node[prvN]{0};
    \edge[prvE]; [
      \edge[prvE]; [
        \edge[prvE]; [
          \edge[prvE]; [
            \edge[prvE]; \node[prvNnoline]{\dots};
            \edge[prvE]; \node[prvN]{1};
          ]
          \edge[prvE]; \node[prvN]{1};
        ]
        \edge[prvE]; \node[excN]{101};
      ]
      \edge[prvE]; \node[prvN]{1};
    ]
  ]
  \end{tikzpicture}
};

\node[name=t2b,at=(t2 -| t0b),label={[title]above:Refined Tree $(B_1 B_2 B_3 B_4)$}]{ % when $B_1 B_2 B_3 = 101$
  \begin{tikzpicture}[remember picture]
  \Tree
  [.\node[branch,name=r0]{};
    \edge[prvE]; \node[prvN]{0};
    \edge[prvE]; [
      \edge[prvE]; [
        \edge[prvE]; [
          \edge[prvE]; [
            \edge[prvE]; \node[prvNnoline]{\dots};
            \edge[prvE]; \node[prvN]{1};
          ]
          \edge[prvE]; \node[prvN]{1};
        ]
        \edge[prvE]; [
          \edge[refE]; \node[refN]{1010};
          \edge[refE]; [
            \edge[refE]; \node[refN]{1010};
            \edge[refE]; \node[refN]{1011};
          ]
        ]
      ]
      \edge[prvE]; \node[prvN]{1};
    ]
  ]
  \end{tikzpicture}
};

%%%%%%%%%%%%%%%% Level 3 %%%%%%%%%%%%%%%%%%
\node[name=t3,below=0.675 of t2,label={[title]above:Original Tree $(B_1 B_2 B_3 B_4)$}]{
  \begin{tikzpicture}[remember picture]
  \Tree
  [.\node[branch,name=r0]{};
    \edge[prvE]; \node[prvN]{0};
    \edge[prvE]; [
      \edge[prvE]; [
        \edge[prvE]; [
          \edge[prvE]; [
            \edge[prvE]; \node[prvNnoline]{\dots};
            \edge[prvE]; \node[prvN]{1};
          ]
          \edge[prvE]; \node[prvN]{1};
        ]
        \edge[prvE]; [
          \edge[curE]; \node[curN]{1010};
          \edge[excE]; [
            \edge[excE]; \node[excN]{1010};
            \edge[curE]; \node[curN]{1011};
          ]
        ]
      ]
      \edge[prvE]; \node[prvN]{1};
    ]
  ]
  \end{tikzpicture}
};

\node[name=t3b,at=(t3 -| t0b)]{$\mathbf{\dots}$};

\node[name=t0l,left=1.5 of t0,title]{\begin{tabular}{@{}l@{}}Recursion \\ Level 0:\end{tabular}};
\node[name=t1l,at=(t0l |- t1), title]{\begin{tabular}{@{}l@{}}Recursion \\ Level 1:\end{tabular}};
\node[name=t2l,at=(t0l |- t2), title]{\begin{tabular}{@{}l@{}}Recursion \\ Level 2:\end{tabular}};
\node[name=t3l,at=(t0l |- t3), title]{\begin{tabular}{@{}l@{}}Recursion \\ Level 3:\end{tabular}};

\draw[-latex,transform canvas={yshift=0pt},shorten <=55pt,shorten >=75pt] (t0.center) --node[midway,below=2pt,title,font=\itshape]{\hspace{-.75cm}Refine} (t0b.center);
\draw[-latex,transform canvas={yshift=0pt},shorten <=55pt,shorten >=75pt] (t1.center) --node[midway,below=2pt,title,font=\itshape]{\hspace{-.75cm}Refine} (t1b.center);
\draw[-latex,transform canvas={yshift=0pt},shorten <=55pt,shorten >=75pt] (t2.center) --node[midway,below=2pt,title,font=\itshape]{\hspace{-.75cm}Refine} (t2b.center);
\end{tikzpicture}
\end{adjustbox}

\captionsetup{skip=4pt}
\caption{DDG trees that \cref{alg:sampler-opt} explores for the example
trace in \subref{fig:optimal-sampler-trace}.}
\label{fig:optimal-sampler-ddg}
\end{subfigure}

% https://tex.stackexchange.com/q/716580
\addtocounter{figure}{-1}
\captionsetup{aboveskip=4pt,belowskip=0pt}
\caption{Entropy-optimal generation for a binary-coded probability
distribution $p: \set{0,1}^* \to [0,1]$.
In \cref{alg:sampler-opt}, the parameter $b$ (defaulted to
the empty string, $\varepsilon$) denotes a string that stores the bits
generated so far, and $\ell$ counts the number of calls to $\Flip$.
The notation $[z]_i$ denotes the $i$th bit in $z = (z_0.z_1z_2z_3\ldots)_2 \in [0,1]$.
In the sample trace \subref{fig:optimal-sampler-trace}, a blue
bit denotes a $\Flip$ that creates a new recursive call; a yellow bit
indicates the selected leaf in the DDG tree; and a pink bit is
the label of that leaf, which is appended to $b$.
Gray bits are not visited in this execution, as the
algorithm lazily explores a \textit{single path} through the DDG tree.}
\label{fig:optimal-sampler}

\end{figure}

\subsection{Entropy-Optimal Generation}
\label{sec:binary-sampling-optimal}

Our approach to efficiently sampling from a binary-coded probability distribution
$p$ is based on the idea of \textit{lazily refining} a DDG tree.
The key idea is as follows: we first generate $B_1$ using an
entropy-optimal DDG tree $T_1$ for $\set{0\mapsto p(0), 1\mapsto p(1)}$,
which corresponds to arriving at a leaf node $x$ at $T_1$.
Rather than generated by $B_2{\mid}B_1$ as in the chain rule
(\cref{sec:binary-sampling-suboptimal}), $B_2$ is instead determined by
expanding a subtree under the leaf node $x$ of $T_1$.
This subtree is a fragment of an entropy-optimal DDG tree $T_2$ for
$\set{00\mapsto p(00),01\mapsto p(01), 10\mapsto p(10),11\mapsto p(11)}$.
Repeating this process allows us to efficiently explore a single,
linear-memory path without building an exponentially large DDG tree $T_n$
over $\set{0,1}^n$.
\Cref{fig:ddg-tree-refine} shows two examples of refinement,
where diagrams labeled ``Original Tree''
show an entropy-optimal DDG tree $T_1$ for $B_1 \sim \mathrm{Bernoulli}(2/3)$.
Blue arrowed edges show a random execution path.
When halting at a leaf (red) in the Original Tree (i.e., $B_1=b_1$ is determined),
the leaf is \textit{refined} into a new subtree whose leaves are outcomes of $B_1B_2$ with $B_1 = b_1$.
The refined subtree could be a leaf
(\cref{fig:ddg-tree-refine-leaf}) or branch
(\cref{fig:ddg-tree-refine-subtree}) node.
While exactly one leaf in $T_1$ is refined in a given execution, refining
\textit{every} leaf of $T_1$ produces an entropy-optimal tree $T_2$ for
$B_1B_2$.

\Cref{alg:sampler-opt} presents an entropy-optimal generator for a
binary-coded probability distribution $p$ that uses refinement
to efficiently traverse an infinite-size DDG tree.
The algorithm itself recurses infinitely, generating a fresh % \Cref{alg:sampler-opt} itself
bit of the random stream $b \sim p$ at each step.
At the $n$th recursive call, the leaf node selected in the optimal tree for
$p_{n-1}$ over $\set{0,1}^{n-1}$ is refined into an optimal subtree for $p_{n}$.
\Crefrange{algline:UnivSampler-If1-Start}{algline:UnivSampler-Print1-NoIter}
occur when the refined subtree is a leaf (cf.~\cref{fig:ddg-tree-refine-leaf}).
\Crefrange{algline:UnivSampler-Loop}{algline:UnivSampler-Print1} occur when
the refined subtree is not a leaf (cf.~\cref{fig:ddg-tree-refine-subtree}).
The trees are such that a leaf labeled 0 (resp.~1) is always a left
(resp.~right) child, which is visited when $x=0$ (resp.\ $x=1$).
The algorithm is readily
implementable using lazy computation with guarded recursive calls
(\cref{lst:sampler-opt-haskell} in \cref{appx:binary}).

\Cref{fig:optimal-sampler-trace} shows an example trace of
\cref{alg:sampler-opt} on input $p$ from \cref{fig:optimal-sampler-dist}.
\Cref{fig:optimal-sampler-ddg} shows how this example trace
corresponds to a lazy exploration of an infinite-size optimal DDG tree.
Each row shows the original tree to be explored at the start of each
recursive call and the refined tree to be explored at the next
recursion.
For the original trees (\cref{fig:optimal-sampler-ddg}, left column), black
edges and nodes show the subtree to be explored at the current recursion
level.
Inactive paths that were considered previously are shown in gray.
Blue edges and red-boxed nodes denote the edges and leaves
explored and chosen by \cref{alg:sampler-opt}, using the
outputs of $\Flip$ from \cref{fig:optimal-sampler-trace}.
For refined trees (\cref{fig:optimal-sampler-ddg}, right column), red
edges and red-labeled nodes denote the outcome of refining the red-boxed
node in the original tree on the left, which will be explored at the next
recursion.
\Cref{alg:sampler-opt} only ever constructs the blue edges and red-boxed
nodes in \cref{fig:optimal-sampler-ddg}: the rest are shown for
illustration.

The following results formally justify the correctness and entropy-optimality
of \cref{alg:sampler-opt}.

\begin{theorem}[name=,restate=BitPatternsSimple]
\label{proposition:bit-patterns-simple}
Let $z,x,y \in [0,1]$ satisfy $z = x + y$.
Suppose $\ell \ge 0$ is any index such that $z_\ell = 1$ and
$z_j = 0$ for all $j > \ell$, where
$z = (z_0.z_1z_2\ldots)_2$,
$x = (x_0.x_1x_2\ldots)_2$ and
$y = (y_0.y_1y_2\ldots)_2$
are concise binary expansions.
The binary expansions of $x$ and $y$ match exactly one of the following patterns:
\begin{align}
\footnotesize
\begin{NiceArray}{c@{\;\;}ccclr}
~
&~
&~
&\begin{bmatrix}[c|c] 0 & 1 \\ 1 & 0 \end{bmatrix}
&\begin{bmatrix}[c] 0 \\ 0 \end{bmatrix}^\infty
&\mbox{\normalfont (Pattern \ref*{eq:bit-patterns-simple}.1)}
\\[7pt]
\begin{array}{c@{}} ~ \\ + \end{array}
&
\begin{bmatrix}[ccccc]
x_{\ell} & \dots & x_{\ell'} & \ldots & ~ \\
y_{\ell} & \dots & y_{\ell'} & \ldots & ~
\end{bmatrix}
&=\phantom{x}
&\begin{bmatrix}[c|c] 0 & 1 \\ 0 & 1 \end{bmatrix}
&\begin{bmatrix}[c|c] 0 & 1 \\ 1 & 0 \end{bmatrix}^{*}
 \begin{bmatrix} 1 \\ 1 \end{bmatrix}
 \begin{bmatrix} 0 \\ 0 \end{bmatrix}^\infty
&\mbox{\normalfont (Pattern \ref*{eq:bit-patterns-simple}.2)}
\\[7pt]
~
&~
&~
&\begin{bmatrix}[c|c] 0 & 1 \\ 0 & 1 \end{bmatrix}
&\begin{bmatrix}[c|c] 0 & 1 \\ 1 & 0 \end{bmatrix}^{*}
&\mbox{\normalfont (Pattern \ref*{eq:bit-patterns-simple}.3)}
\\[10pt]
\cline{1-2} \cline{4-5}
\\[-10pt]
\begin{array}{c@{}} = \end{array}
&
\;\;\;
\begin{matrix}[ccccc]
z_{\ell} & \dots  & z_{\ell'} & \dots & ~
\end{matrix}
&=\phantom{x}
&1
&0\quad\dots\quad0\quad\dots
& %\qedhere
\CodeAfter
  \SubMatrix\{{1-4}{3-5}.
\end{NiceArray}
\label{eq:bit-patterns-simple}
\end{align}
Here, each pattern is written in the style of regular expressions:
$[R ~|~ R']$ denotes either $R$ or $R'$,
$[R]^*$ denotes zero or more occurrences of $R$, and
$[R]^\infty$ denotes the infinite occurrences of $R$.
\end{theorem}

\begin{corollary}[name=,restate=DdgStructureOne]
\label{lemma:ddg-structure-1}
When refining the deepest node in an entropy-optimal DDG tree with label $b$ at a
level $\ell$, there are three mutually exclusive and
collectively exhaustive possibilities (where $x,y\in\set{0,1}$):
\begin{center}
\scriptsize
\hfill\begin{tikzpicture}[remember picture]
\normalfont
\tikzstyle{branch}=[shape=coordinate]
\tikzset{sibling distance=5pt}
\tikzset{level 1/.style={level distance=10pt}}
\tikzset{level 2/.style={level distance=7pt}}
\tikzset{level 3/.style={level distance=7pt}}
\tikzset{level 4/.style={level distance=11pt}}
\tikzset{every tree node/.style={anchor=north}}

\tikzstyle{branch}=[shape=coordinate]
\tikzstyle{excE}=[color=blue, line width=1pt,-latex]    % executed edge
\tikzstyle{refE}=[color=red,  line width=0.6pt,densely dashed] % refined  edge
\tikzstyle{curN}=[color=black,draw=black,inner sep=1pt] % current  node
\tikzstyle{excN}=[color=black,draw=red,  inner sep=1pt, thick] % executed node
\tikzstyle{refN}=[color=red,  draw=black,inner sep=1pt] % refined node

\captionsetup[subfigure]{font=footnotesize,justification=centering}
\node[name=m1]{
\begin{tikzpicture}[remember picture]
\Tree
  [.\node[branch,name=root1,label={[font=\footnotesize]above:\begin{tabular}{@{}c@{}}Original Tree\\\phantom{(Leaf)}\end{tabular}}]{};
    $\dots$
    \node[name=leaf1,excN]{$b$};
  ]
\draw[-latex,transform canvas={xshift=5pt}] (root1.east -| leaf1.east) --node[pos=0.5,anchor=west]{$\ell$} (leaf1.south east);
\end{tikzpicture}
};

\node[name=m2,right=4 of m1.north,anchor=north]{ % right=2.5 of ...
\begin{tikzpicture}[remember picture]
\Tree
  [.\node[branch,name=root2,label={[font=\footnotesize]above:\begin{tabular}{@{}c@{}}Refined Tree\\ (Leaf)\end{tabular}}]{};
    $\dots$
      \node[name=leaf2,refN]{$bx$};
  ]
\draw[-latex,transform canvas={xshift=5pt}] (root2.east -| leaf2.east) --node[pos=0.5,anchor=west]{$\ell$} (leaf2.south -| leaf2.east);
\end{tikzpicture}
};

\node[name=m3,right=3 of m2.north,anchor=north]{
\begin{tikzpicture}[remember picture]
\Tree
  [.\node[branch,name=root3,label={[font=\footnotesize]above:\begin{tabular}{@{}c@{}}Refined Tree\\(Finite Subtree)\end{tabular}}]{};
    $\dots$
    [.\node[branch,name=leaf3,yshift=-0pt]{}; %.\node[leaf,name=leaf,thick,draw=red]{$b$};
      \edge[refE]; [
        \edge[refE]; [
          .{$\dots$}
          \edge[refE]; \node[refN]{$b1$};
          \edge[refE]; \node[refN,name=wonyeol1]{$b0$}; ]
        \edge[refE]; \node[refN]{$by$};
        ]
      \edge[refE]; \node[refN]{$bx$};
      ]
    ]
\draw[-latex,transform canvas={xshift=5pt}] (root3.east -| leaf3.east) --node[pos=0.5,anchor=west]{$\ell$} (leaf3.south -| leaf3.east);
\end{tikzpicture}
};

\node[name=m4,right=3 of m3.north,anchor=north]{
\begin{tikzpicture}[remember picture]
\Tree
  [.\node[branch,name=root4,label={[font=\footnotesize]above:\begin{tabular}{@{}c@{}}Refined Tree\\(Infinite Subtree)\end{tabular}}]{};
    $\dots$
    [.\node[branch,name=leaf4,yshift=-0pt]{};
      \edge[refE]; [
        \edge[refE]; {$\dots$}
        \edge[refE]; \node[refN]{$by$};
      ]
      \edge[refE]; \node[refN,name=wonyeol2]{$bx$};
    ]
  ]
\draw[-latex,transform canvas={xshift=5pt}] (root4.east -| leaf4.east) --node[pos=0.5,anchor=west]{$\ell$} (leaf4.south -| leaf4.east);
\end{tikzpicture}
};

\draw[->,transform canvas={yshift=-16pt},shorten <=43pt,shorten >=40pt] (root1) --node[midway,below=2pt,font=\footnotesize]{\hspace{1.5pt}Refine} (root2);

\node[name=m4cap,below=1.6 of root4]{ \footnotesize (By Pattern \ref{eq:bit-patterns-simple}.3) };
\node[name=m3cap,at=(m4cap -| root3)]{ \footnotesize (By Pattern \ref{eq:bit-patterns-simple}.2) };
\node[name=m2cap,at=(m3cap -| root2)]{ \footnotesize (By Pattern \ref{eq:bit-patterns-simple}.1) };

\end{tikzpicture}\qedhere
\end{center}
\end{corollary}

\begin{theorem}[name=,restate=BitPatterns]
\label{proposition:bit-patterns}
Let $z, x, y \in [0,1]$ satisfy $z = x + y$.
Suppose $0 \leq \ell < \ell'$ are two indexes such that $z_\ell = 1, z_{\ell'}=1$
and $z_j = 0$ for $\ell+1 \le j \le \ell'-1$,
where $z_j$, $x_j$ and $y_j$ are defined as in \cref{proposition:bit-patterns-simple}.
The binary expansions of $x$ and $y$ between locations
$\ell$ and $\ell'$ match exactly one of three possible patterns:
\begin{align}
\footnotesize
\begin{NiceArray}{@{}c@{\;\;}cccccr@{}}
~
&~
&~
&\begin{bmatrix}[c|c] 0 & 1 \\ 1 & 0 \end{bmatrix}
&\begin{bmatrix}[c] 0 \\ 0 \end{bmatrix}^{\smash{\ell'-\ell - 1}}
&\begin{bmatrix}[c|c|c] 0 & 0 & 1 \\ 0 & 1 & 0\end{bmatrix}\!\!
&\mbox{\normalfont (Pattern \ref*{eq:bit-patterns}.1)}
\\
\begin{array}{c@{}} ~ \\ + \end{array}
&
\begin{bmatrix}[ccc]
x_{\ell} & \dots & x_{\ell'} \\
y_{\ell} & \dots & y_{\ell'}
\end{bmatrix}
&=\phantom{x}
&\begin{bmatrix}[c|c] 0 & 1 \\ 0 & 1\end{bmatrix}
&\begin{bmatrix}[c|c] 0 & 1 \\ 1 & 0\end{bmatrix}^{\ell'-\ell - 1}
&\begin{bmatrix}[c] 1 \\ 1 \end{bmatrix}\!\!
&\mbox{\normalfont (Pattern \ref*{eq:bit-patterns}.2)}
\\%[-2pt]
~
&~
&~
&\begin{bmatrix}[c|c] 0 & 1 \\ 0 & 1 \end{bmatrix}
&\begin{bmatrix}[c|c] 0 & 1 \\ 1 & 0 \end{bmatrix}^{k_1}
 \begin{bmatrix}[c] 1 \\ 1 \end{bmatrix}
 \begin{bmatrix}[c] 0 \\ 0 \end{bmatrix}^{k_2}
&\begin{bmatrix}[c|c|c] 0 & 0 & 1 \\ 0 & 1 & 0 \end{bmatrix}\!\!
&\begin{aligned}
  \mbox{\normalfont where } k_1 + k_2 = \ell' - \ell - 2 & \\
  \mbox{\normalfont (Pattern \ref*{eq:bit-patterns}.3)}  &
 \end{aligned}
\\[10pt]
\cline{1-2} \cline{4-6}
\\[-10pt]
\begin{array}{c@{}} = \end{array}
&
\;\;\;
\begin{matrix}[ccc]
z_{\ell} & \dots  & z_{\ell'}
\end{matrix}
&=\phantom{x}
&1
&0\quad\dots\quad0
&1
&\qedhere
\CodeAfter
  \SubMatrix\{{1-4}{3-6}.
\end{NiceArray}
\label{eq:bit-patterns}
\end{align}
\end{theorem}

\begin{corollary}[name=,restate=DdgStructureTwo]
\label{lemma:ddg-structure-2}
Consider the process of refining a node labeled $b$ at level $\ell$ in an entropy-optimal DDG tree,
such that there also exists a node labeled $b$ at some level $\ell' > \ell$.
There exists an entropy-optimal refinement scheme such that
\begin{center}
\begin{minipage}[m]{.2\linewidth}
\vspace{-6pt}
\begin{tikzpicture}[baseline=(root.base)]\normalfont
\tikzset{edge from parent/.style={draw,thick}}
\tikzstyle{branch}=[shape=coordinate]
\tikzstyle{refined}=[densely dashed]
\Tree
  [.\node[branch,name=root,label={[font=\footnotesize]above:Optimal DDG Tree}]{};
    % \edge[draw=black] node[pos=0.5,fill=white]{$\dots$};
    \node[name=leaf1,draw,inner sep=1.5pt]{$b$};
    [.\node[name=dots]{\dots};
      [.{\dots}
        \node[name=leaf2,draw,inner sep=1.5pt]{$b$};
        $\dots$
        ]
      \node[name=leaf,draw,inner sep=1.5pt,thick,draw=red]{$b$};
    ]
  ]
\draw[-latex,transform canvas={xshift=5pt}] (root.east -| leaf.east) --node[pos=0.45,anchor=east]{$\ell$} (leaf.south east);
\draw[-latex,transform canvas={xshift=-7pt}] (root.east -| leaf1.west) --node[pos=0.4,anchor=west,inner xsep=1.5pt]{$\bar{\ell}$} (leaf1.south west);
\draw[-latex,transform canvas={xshift=-25pt}] (root.east -| leaf2.west) --node[pos=0.9,anchor=west]{$\ell'$} (leaf2.south west);

\node[name=isosceles,
  shape=isosceles triangle,
  isosceles triangle stretches,
  minimum height=0.6cm,
  minimum width=0.5cm,
  shape border rotate=90,
  draw, refined,
  anchor=north,
  at=(leaf.south)]{};
\node[name=isoscelesdotted,
  shape=isosceles triangle,
  isosceles triangle stretches,
  minimum height=1.02cm,
  minimum width=0.45cm,
  shape border rotate=90,
  draw, refined,
  anchor=north,
  at=(leaf1.south)]{};
\draw[-latex,transform canvas={xshift=10pt}] (root.east -| isosceles.south east) --node[pos=0.88,anchor=west]{$\ell''$} (isosceles.south east);
\end{tikzpicture}
\end{minipage}\hfill
\begin{minipage}[m]{.8\linewidth}
\vspace{5pt}
\begin{itemize}
\item
  If the binary expansions satisfy $[p(b0)]_\ell [p(b1)]_\ell \in \set{01,10}$,
  then the node $b$ at level $\ell$ is refined into a leaf node.
\item
  If $[p(b0)]_\ell [p(b1)]_\ell \in \set{00,11}$,
  then the node $b$ at level $\ell$ is refined into a subtree
  that terminates with a pair of nodes labeled $(b0, b1)$ at some
  level $\ell'' \in [\ell+1, \ell']$.
  All levels of this subtree above $\ell''$ have precisely one node.
\item If $[p(b0)]_\ell [p(b1)]_\ell \in \set{11}$, then the corresponding nodes
  in the DDG tree labeled $b0$ and $b1$ correspond to leaves of the
  subtree obtained by refining a previous node labeled
  $b$ at some previous level $\bar{\ell} < \ell$.
  Therefore, these bits at location $\ell$ can be ignored when refining the
  node $b$ at level $\ell$.
  \qedhere
\end{itemize}
\vspace{5pt}
\end{minipage}
\end{center}
\end{corollary}

\begin{theorem}[name=,restate=SamplerOpt]
\label{theorem:sampler-opt}
Let $p: \set{0,1}^* \to [0,1]$ be a binary-coded probability distribution.
For each $n \in \mathbb{N}$, \cref{alg:sampler-opt} generates a string
$B_1\ldots{B_n} \sim p_n$ (stored as a prefix of $b$) and is entropy-optimal for $p_n$.
\end{theorem}

\Cref{theorem:sampler-opt} states the
entropy-optimality---and, a fortiori, the soundness---of \cref{alg:sampler-opt}.
This result rests on two number theoretic properties
(\cref{proposition:bit-patterns-simple,proposition:bit-patterns})
for binary expansions of real numbers.
\Cref{lemma:ddg-structure-1,lemma:ddg-structure-2} demonstrate
what these theorems imply about the structure of refined entropy-optimal
DDG trees explored by \cref{alg:sampler-opt},
which are used to prove \cref{theorem:sampler-opt}.
\Cref{proposition:bit-patterns-simple,lemma:ddg-structure-1} justify the
correctness of \cref{alg:sampler-opt} when the current node is the deepest
leaf in the tree with label $b$.
Pattern~\ref{eq:bit-patterns-simple}.1 corresponds to the early exit in
\crefrange{algline:UnivSampler-If1-Start}{algline:UnivSampler-Print1-NoIter},
which do not require any new flips.
Patterns \ref{eq:bit-patterns-simple}.2 and \ref{eq:bit-patterns-simple}.3
correspond to the while-loop, showing there will always be precisely one
leaf node at each iteration.
Therefore, each iteration exits with probability 1/2, and the while-loop
terminates almost surely.
\Cref{proposition:bit-patterns,lemma:ddg-structure-2} justify the correctness
when the current node is not the deepest leaf labeled $b$.
Pattern~\ref{eq:bit-patterns}.1 corresponds to the early exit in
\crefrange{algline:UnivSampler-If1-Start}{algline:UnivSampler-Print1-NoIter}.
Patterns~\ref{eq:bit-patterns}.2 and \ref{eq:bit-patterns}.3
correspond to the while-loop, whose number of iterations is bounded by the
number of levels in the tree until reaching the next leaf labeled $b$.

\begin{remark}
\label{remark:refine-improve-knuth-yao}
\Cref{alg:sampler-opt} can be viewed as an optimized version
of the original \citeauthor{knuth1976} method
that achieves optimal space-time complexity for
refining entropy-optimal DDG trees.
\begin{itemize}
\item \citet[page 384]{knuth1976} describe a nondeterministic procedure
for refining entropy-optimal DDG trees.
The \citeauthor{knuth1976} method explicitly constructs a full DDG tree
at each refinement step.
When refining a leaf node $x$ with label $b$ at level $\ell$, the method
\begin{enumerate}[wide, leftmargin=*, label=(\roman*),nosep]
\item performs a preprocessing step that refines \textit{all} the leaf
  nodes labeled $b$ above $\ell$; then
\item expands \textit{all} the possible execution paths starting
  from the new subtree rooted at $x$.
\end{enumerate}
Each such refinement step takes $O(\ell')$ time and creates $O(\ell')$ leaf
nodes, where $\ell'$ is the height of a resulting DDG tree.
Therefore, the \citeauthor{knuth1976} method requires $O(k^2)$ time and
$O(k^2)$ space to construct a refined entropy-optimal DDG tree
of height $k$.

\item \cref{alg:sampler-opt} is a more efficient method that does
not explicitly construct full DDG trees. It
\begin{enumerate}[wide, leftmargin=*, label=(\roman*),nosep]
\item
  performs \textit{no} preprocessing
    (i.e, avoids refining leaf nodes higher up in the tree); and
\item
  lazily explores only a \textit{single} path down the subtree rooted at
  $x$.
\end{enumerate}
\Cref{alg:sampler-opt} requires $O(k)$ time and $O(1)$ space
to explore a single path in a refined entropy-optimal DDG tree of height $k$;
achieving optimal space-time complexity.
\Cref{lemma:ddg-structure-1,lemma:ddg-structure-2} enable this optimality,
by identifying a class of refined entropy-optimal DDG trees that have
certain ``nice'' properties which
\cref{alg:sampler-opt} exploits for efficient and lazy exploration.
\end{itemize}

The \citeauthor{knuth1976} algorithm is nondeterministic: it
can construct \textit{any} refined entropy-optimal DDG tree
\citep[page 385]{knuth1976}.
In contrast, \cref{alg:sampler-opt} is deterministic:
it can explore only \textit{some} of these trees, because not every refined
entropy-optimal DDG tree satisfies \cref{lemma:ddg-structure-2}
(cf.~\cref{fig:ky-refine}).
\end{remark}

%!TEX root=./paper.tex

\begin{figure}[!t]
\begin{center}
\begin{tikzpicture}
\normalfont
\tikzstyle{branch}=[shape=coordinate]
\tikzset{level distance=10pt, sibling distance=2pt}
\tikzset{every tree node/.style={anchor=north}}

\tikzstyle{branch}=[shape=coordinate]
\tikzstyle{excE}=[color=blue, line width=1pt,-latex]    % executed edge
\tikzstyle{refE}=[color=red,  line width=0.6pt,densely dashed] % refined  edge
\tikzstyle{curN}=[color=black,draw=black,inner sep=1pt] % current  node
\tikzstyle{excN}=[color=black,draw=red,  inner sep=1pt, thick] % executed node
\tikzstyle{refN}=[color=red,  draw=black,inner sep=1pt] % refined node

\captionsetup[subfigure]{font=footnotesize,justification=centering}
\node[name=m1]{
\begin{tikzpicture}
\Tree
  [.\node[branch,name=root,label={[font=\footnotesize]above:\begin{tabular}{@{}c@{}}Original Tree\\\phantom{(Leaf)}\end{tabular}}]{};
    [
      [
        \node[curN]{$1$};
        \node[excN]{$0$}; % [leaf]
        ]
      \node[excN]{$0$}; % [leaf]
      ]
    \node[excN]{$0$}; % [leaf]
  ]
\end{tikzpicture}
};

\node[name=m2,right=3 of m1.north,anchor=north]{
\begin{tikzpicture}
\Tree
  [.\node[branch,name=root,label={[font=\footnotesize]above:\begin{tabular}{@{}c@{}}Refined Tree\\{(Possibility 1)}\end{tabular}}]{};
    [
      [
        \node[curN]{$1$};
        \edge[]; \node[refN]{$00$}; % [refE]
        ]
      \edge[]; \node[refN]{$00$}; % [refE]
      ]
    \edge[]; [ % [refE]
      \edge[refE]; [
        \edge[refE]; [
          \edge[refE];\node[refN]{$00$};
          \edge[refE];\node[refN]{$01$};
        ]
        \edge[refE]; \node[refN]{$01$};
      ]
      \edge[refE]; \node[refN]{$01$};
      %% === old ===
      %% \edge[refE]; \node[refN]{$01$};
      %% \edge[refE]; [
      %%   \edge[refE]; \node[refN]{$01$};
      %%   \edge[refE]; [
      %%     \edge[refE];\node[refN]{$00$};
      %%     \edge[refE];\node[refN]{$01$};
      %%   ]
      %% ]
      %% ===========
    ]
  ]
\end{tikzpicture}
};

\node[name=m3,right=4 of m2.north,anchor=north]{
\begin{tikzpicture}
\Tree
  [.\node[branch,name=root,label={[font=\footnotesize]above:\begin{tabular}{@{}c@{}}Refined Tree\\{(Possibility 2)}\end{tabular}}]{};
    [
      [
        \node[curN]{$1$};
        \edge[]; \node[refN]{$00$}; % [refE]
        ]
      \edge[]; [ % [refE]
        \edge[refE]; [
          \edge[refE]; \node[refN]{$00$};
          \edge[refE]; \node[refN]{$01$};
        ]
        \edge[refE]; \node[refN]{$01$};
      ]
    ]
    \edge[]; [ % [refE]
      \edge[refE]; \node[refN]{$00$}; % 01
      \edge[refE]; \node[refN]{$01$}; % 00
    ]
  ]
\end{tikzpicture}
};

\node[name=m4,right=4 of m3.north,anchor=north]{
\begin{tikzpicture}
\Tree
  [.\node[branch,name=root,label={[font=\footnotesize]above:\begin{tabular}{@{}c@{}}Refined Tree\\(Possibility 3)\end{tabular}}]{};
    [
      [
        \node[curN]{$1$};
        \edge[]; [ % [refE]
          \edge[refE]; \node[refN]{$00$};
          \edge[refE]; \node[refN]{$01$};
        ]
        ]
      \edge[]; [ % [refE]
        \edge[refE]; \node[refN]{$00$};
        \edge[refE]; \node[refN]{$01$};
      ]
      ]
    \edge[]; [ % [refE]
      \edge[refE]; \node[refN]{$00$};
      \edge[refE]; \node[refN]{$01$};
    ]
  ]
\end{tikzpicture}
};

\end{tikzpicture}\qedhere
\end{center}
\captionsetup{skip=0pt,belowskip=0pt}
\caption{
  Three possible DDG tree refinements for outcome $0$ in a binary-coded probability
  distribution with
  $p(0) = (0.111)_2, p(1) = (0.001)_2$, and $p(00) = p(01) = (0.0111)_2$.
  The nondeterministic method of \citet{knuth1976} can explore any of these
  possibilities, with $O(k^2)$ space and time complexity for building a
  depth-$k$ refined tree.
  The deterministic method in \cref{alg:sampler-opt}, based on
  \cref{lemma:ddg-structure-1,lemma:ddg-structure-2},
  deterministically refines a path in a tree
  whose structure follows that of Possibility 3, with $O(k)$
  time and $O(1)$ space complexity.}
\label{fig:ky-refine}
\end{figure}

\section{Exact Random Variate Generators for Numerical Cumulative Distribution Functions}
\label{sec:floating}
%!TEX root=./paper.tex

This section describes an implementation of \cref{alg:sampler-opt}
for exact random variate generation using finite-precision computation.
The main algorithm, presented in \cref{sec:floating-sample}, rests on two
connections.
The first (in \cref{sec:floating-binary}) is between
binary strings and a novel unifying abstraction
for finite-precision binary number formats
(\cref{def:binary-number-format}).
The second (in \cref{sec:floating-word-ram}) is between
binary-coded probability distributions (\cref{definition:bcpd}) and
numerical implementations of cumulative distribution functions (\cref{def:cda}).

\subsection{Finite-Precision Binary Number Formats}
\label{sec:floating-binary}

\begin{definition}
\label{def:extended-reals}
The set of \textit{extended reals} is defined by $\realext \defas \real \cup \set{-\infty, +\infty, \bot}$,
where $\bot$ denotes some ``special'' value (e.g., NaN in floating-point formats).
A strict linear order $<_{\realext}$ over $\realext$ is given by
$x <_{\realext} x'$ and $-\infty <_{\realext} x <_{\realext} +\infty <_{\realext} \bot$
for all $x, x' \in \real$ with $x <_\real x'$; i.e., $\bot$ is the largest value.
A weak linear order $\leq_{\realext}$ over $\realext$ is defined as usual:
$x \leq_{\realext} x'$ if and only if $x=x'$ or $x <_{\realext} x'$.
\end{definition}

\begin{definition}
\label{def:binary-number-format}
A \textit{binary number format} $\bfmt \defas (n, \gamma_\bfmt, \phi_\bfmt)$ is a 3-tuple:
\begin{itemize}[noitemsep]
  \item $n \ge 1 $ is an integer indicating the number of binary digits in each bit string $b \in \bool^n$;
  \item $\gamma_\bfmt: \bool^n \to \realext_\bfmt$ is a mapping from $n$-bit strings
    onto a subset $\realext_\bfmt \subset \realext$ of computable reals;
  \item $\phi_\bfmt: \bool^n \to \bool^n$ is a bijection such that % ... is a bijection where
  $b <_{\rm dict} b'$ implies
  $\gamma_\bfmt(\phi_\bfmt(b)) \le_{\realext} \gamma_\bfmt(\phi_\bfmt(b'))$,
  where $<_{\rm dict}$ denotes the dictionary (i.e., lexicographic) ordering on $\bool^n$.
  \hfill \qedhere
\end{itemize}
\end{definition}
With slight abuse of notation, the set $\realext_\bfmt$ of reals
in \cref{def:binary-number-format}
is sometimes also denoted by $\bfmt$.

\begin{remark}
\label{remark:binary-number-format-ordering}
For a binary number format $\bfmt$, the bijection $\phi_\bfmt$ defines an ordering
$<_\bfmt$ on $\bool^n$.
In particular, equipping the domain of
$\phi_\bfmt$ with the dictionary ordering % (lexicographic)
$<_{\rm dict}$ gives a strict linear order $<_\bfmt$ over $\bool^n$
and a weak linear order over ${\realext}_\bfmt \subset \realext$:
\begin{align}
\phi_\bfmt(0^n)
&
<_\bfmt \phi_\bfmt(0^{n-1}1)
<_\bfmt \dots
<_\bfmt \phi_\bfmt(1^{n-1}0)
<_\bfmt \phi_\bfmt(1^n),
\label{eq:snobbing-1}
\\
\gamma_\bfmt(\phi_\bfmt(0^n))
&
\leq_{\realext} \gamma_\bfmt(\phi_\bfmt(0^{n-1}1))
\leq_{\realext} \dots
\leq_{\realext} \gamma_\bfmt(\phi_\bfmt(1^{n-1}0))
\leq_{\realext} \gamma_\bfmt(\phi_\bfmt(1^n)).
\label{eq:snobbing-2}
\end{align}
The predecessor and successor of a non-extremal value
$b \in \bool^n$ under ordering \cref{eq:snobbing-1} are denoted by
$\mathrm{pred}_\bfmt(b)$ and $\mathrm{succ}_\bfmt(b)$.
Similarly, $\mathrm{pred}_\bfmt(x)$ and $\mathrm{succ}_\bfmt(x)$ for
$x \in {\realext}_\bfmt$ are used for~\cref{eq:snobbing-2}.
\end{remark}

\begin{example}[Integer Formats:
Unsigned $\mathbb{U}_n$,
Sign-Magnitude $\mathbb{M}_n$, and
Two's-Complement $\mathbb{T}_n$]
\label{example:format-unsigned-int}
\label{example:format-sign-magnitude}
\label{example:format-twos-complement}
\mbox{}\\
\begin{adjustbox}{max width=.95\linewidth}
\centering
\newcommand\TT{\rule{0pt}{2.5ex}}       % Top strut
\newcommand\BB{\rule[-1.2ex]{0pt}{0pt}} % Bottom strut
\begin{tabular}[b]{@{}|lll|@{}}
\hline
$\mathbb{U}_n \defas (n, \gamma_{\mathbb{U}_n}, \phi_{\mathbb{U}_n})$
&
$\gamma_{\mathbb{U}_n}(b_{n-1}\ldots{b_0}) \defas \sum_{j=0}^{n-1}2^jb_j$
&
$\phi_{\mathbb{U}_n}(b_{n-1}\ldots{b_0}) \defas b_{n-1}\ldots{b_0}$
\\\hline
$\mathbb{M}_n \defas (n+1, \gamma_{\mathbb{M}_n}, \phi_{\mathbb{M}_n})$
&
$\gamma_{\mathbb{M}_n}(sb_{n-1}\ldots{b_0}) \defas (-1)^{s}\times\sum_{j=0}^{n-1}2^jb_j$
&
$\begin{aligned}[t]
  \phi_{\mathbb{M}_n}(0b_1\ldots{b_{n}}) \defas 1\bar{b}_1\ldots\bar{b}_{n}\\[-2pt]
  \phi_{\mathbb{M}_n}(1b_1\ldots{b_{n}}) \defas 0{b}_1\ldots{b}_{n}
\end{aligned}$
\TT \\\hline
$\mathbb{T}_n \defas (n+1, \gamma_{\mathbb{T}_n}, \phi_{\mathbb{T}_n})$
&
$\gamma_{\mathbb{T}_n}(sb_{n-1}\ldots{b_0}) \defas -s2^{n} + \sum_{j=0}^{n-1}2^jb_j$
&
$\begin{aligned}[t]
  \phi_{\mathbb{T}_n}(0b_1\ldots{b_{n}}) \defas 1{b}_1\ldots{b}_{n}\\[-2pt]
  \phi_{\mathbb{T}_n}(1b_1\ldots{b_{n}}) \defas 0{b}_1\ldots{b}_{n}
\end{aligned}$
\\\hline
\end{tabular}
\end{adjustbox}
\end{example}

\begin{example}[Fixed-Point Formats]
\label{example:format-fixed-point}
The unsigned fixed-point format $\mathbb{U}^m_n \defas (n, \gamma_{\mathbb{U}^m_n}, \phi_{\mathbb{U}^m_n})$
(parameterized by offset $m \in \mathbb{Z}$)
has $\gamma_{\mathbb{U}^m_n} (b_{n-1}\ldots{b_0}) \defas 2^{-m}\gamma_{\mathbb{U}_n}(b_{n-1}\ldots{b_0})$
and $\phi_{\mathbb{U}^m_n} \defas \mathrm{id}$.
The sign-magnitude $\mathbb{M}^m_n$
and two's-complement $\mathbb{T}^m_n$ fixed-point formats are defined analogously.
\end{example}

\begin{example}[Floating-Point Formats~\citep{ieee754-2019}]
\label{example:format-floating-point}
The IEEE-754 floating-point format
$\floatEm \defas (1+E+m, \gamma_{\floatEm}, \allowbreak \phi_{\floatEm})$ is comprised
of $n$-bit strings, where $E \geq 1$ is the number of exponent bits, $m \geq 1$
the number of mantissa bits, and the leading bit is a sign bit.
Letting $b_E \defas 2^{E-1}-1$ be the ``exponent~bias'', % denote
\begin{align}
\gamma_{\floatEm}(s0^Ef_1\ldots{f_m}) &\defas (-1)^s(0.f_1\ldots{f_m})_2 \times 2^{1-b_E},
&
\gamma_{\floatEm}(s1^Ef_1\ldots{f_m}) &\defas \bot,
\\
\gamma_{\floatEm}(se_{E}\ldots{e_1}f_1\ldots{f_m}) &\defas (-1)^s (1.f_1\ldots{f_m})_2 \times 2^{({e_{E}\ldots{e_1}})_2-b_E},
&
\gamma_{\floatEm}(s1^E0^m) &\defas (-1)^s  \infty.
\end{align}
There are two bit-string representations for $0 \in \realext$, $00^{E+m}$ and $10^{E+m}$,
which are referred to as positive zero and negative zero, respectively.
The bijection $\phi_{\floatEm}$ is similar to that of the sign-magnitude
format $\mathbb{M}_n$, with an offset to ensure that all strings
mapping to $\bot$ are maximal\footnote{%
The IEEE-754 floating-point standard treats NaN ($\bot$) bit patterns as unordered,
whereas \cref{eq:kinetical} treats them all identically and as a maximal
element to obtain a well-defined \WCDF. Alternative orderings of these
bit patterns are possible.}:
\begin{align}
\phi_{\floatEm}(b_0b_1\ldots{b_{n-1}}) &\defas \begin{cases}
  \phi_{\mathbb{M}_{E+m}}\big((b_0b_1\ldots b_{n-1})_2 + (2^m-1)\big)
      & \mbox{if } b_0\ldots{b_{n-1}} \le_{\rm dict} 11^E0^m \\
  b_0b_1\ldots b_{n-1} & \mbox{otherwise}.\hfill\llap{\qedsymbol}
  \end{cases}\label{eq:kinetical}
\end{align}\noqed
\end{example}

\begin{example}
The following diagram is an example format $\bfmt$,
corresponding to $\mathbb{F}^{1}_{1}$ in \cref{example:format-floating-point}.

\noindent\begin{minipage}[b]{.1\linewidth}
\hfill
\end{minipage}%
\begin{minipage}[t]{.8\linewidth}
\centering
%!TEX root=./paper.tex

\begin{tikzpicture}[scale=.8]
\footnotesize
\def\ydist{1.1}
\def\tick{.075}

% AXES
\draw[thick,-] (1,0) -- node[name=A,pos=1.05,anchor=west]{$(\set{0,1}^3, <_{\rm dict})$} (8,0);
\foreach \i/\j in {1/000, 2/001, 3/010, 4/011, 5/100, 6/101, 7/110, 8/111} {
  \draw[thick]
    (\i, -\tick)
    --
    node[label={[name=level1-bottom-\i,inner sep=0pt]below:\j}]{}
    (\i, \tick);
}

\draw[thick,-] (1,-\ydist) -- (8,-\ydist);
\foreach \i/\j in {1/110, 2/101, 3/100, 4/000, 5/001, 6/010, 7/011, 8/111} {
  \draw[thick]
    ($(\i, -\tick) + (0, -\ydist)$)
    --
    node[label={[name=level2-bottom-\i,inner sep=0pt]below:\j}]{}
    node[pos=1.25,name=level2-top-\i]{}
    ($(\i, \tick) + (0, -\ydist)$);
}

\foreach \i in {1, ..., 8} {
  \draw[-latex] (level1-bottom-\i.south) -- (level2-top-\i.center);
}

\node[name=B,at={(A |- (0,-\ydist))},anchor=center]{$(\set{0,1}^3, <_{\rm \bfmt})$};
\node[name=C,at={(B |- (0,-2*\ydist))},anchor=center]{$(\realext, <_{\rm \realext})$};

\draw[-latex] (A) --node[pos=0.5,right]{$\phi_{\bfmt}$} (B);
\draw[-latex] (B) --node[pos=0.5,right]{$\gamma_{\bfmt}$} (C);

\draw[thick,-,draw] (2,-2*\ydist) -- (6,-2*\ydist);

\foreach \i/\j in {2/$-\infty$, 3/-1, 4/0, 5/1, 6/$+\infty$} {
  \draw[thick]
    ($(\i, -\tick) + (0, -2*\ydist)$)
    --
    node[label={[inner sep=0pt]below:\j}]{}
    node[pos=1.25,name=level3-top-\i]{}
    ($(\i, \tick) + (0, -2*\ydist)$);
}

\draw[draw=none]
  ($(7, -\tick) + (0, -2*\ydist)$)
  --
  node[fill=black,circle,inner sep=1pt,label={below:$\bot$}]{}
  node[pos=1.25,name=level3-top-7]{}
  ($(7, \tick) + (0, -2*\ydist)$);

\foreach \i/\j in {1/2, 2/3, 3/4, 4/4, 5/5, 6/6, 7/7, 8/7} {
  \draw[-latex] (level2-bottom-\i.south) -- (level3-top-\j.center);
}

\end{tikzpicture}
\end{minipage}
\begin{minipage}[t]{.1\linewidth}
\hfill\qedsymbol
\end{minipage}
\noqed
\end{example}

\begin{proposition}[name=,restate=FloatingPointMonotonic]
\label{remark:floating-point-monotonic}
The ordering $<_{\floatEm}$ induced by $\phi_{\floatEm}$~\cref{eq:kinetical}
guarantees that $\gamma_{\floatEm}: {(\bool^n, <_{\floatEm})} \to ({\realext}, <_{\realext})$
is monotonic. % (ignoring strings mapping to $\bot$).
That is, for any distinct $b, b' \in \bool^{1+E+m}$ such that
$\gamma_{\floatEm}(\set{b,b'}) \not\in \set{\set{0}, \set{\bot}}$,
the following are equivalent:
$\phi_{\floatEm}^{-1}(b) <_{\rm dict} \phi_{\floatEm}^{-1}(b')
\iff
b <_{\floatEm} b'
\iff
\gamma_{\floatEm}(b) <_{\realext} \gamma_{\floatEm}(b').$
\end{proposition}

\begin{example}[Posit Format~\citep{Posit2022}]
\label{example:format-posit}
The posit format
$\positn \defas (n, \gamma_{\positn}, \phi_{\positn})$ is comprised
of $n$-bit strings ($n \,{\ge}\, 3$).
The first bit $s$ is the sign field.
The next $k+1$ bits form a variable-length regime field ($1 \le k \le n-2$)
where $b_1=\cdots=b_{k}$, $b_{k+1} = \bar{b}_1$.
The next two bits $e_1e_0$ form an exponent field.
The last bits $f_{1}\ldots f_{m}$ form a fraction field.
The real mapping has
$\gamma_{\positn}(00^{n-1}) \defas 0$,
$\gamma_{\positn}(10^{n-1}) \defas -\infty$,~and
\begin{align}
\gamma_{\positn}(sb_1\ldots b_k b_{k+1}e_1e_0f_1\ldots{f_m}) &\defas
  \begin{aligned}[t]
    &\big((1{-}3s) {+} (0.f_1\ldots{f_m})_2\big) 2^{(1-2s)
      \cdot (4(-k(1-b_1) + (k-1)b_1) + (e_1e_0)_2 + s)} \!.\! \\
  \end{aligned}
\end{align}
If the $n$-bit field is not wide enough to represent some exponent
or fraction bits, these bits are zero.
The mapping $\phi_{\positn}\defas\phi_{\mathbb{T}_n}$ is identical to that of the two's-complement
format from \cref{example:format-twos-complement}.
\end{example}

\subsection{Finite-Precision Cumulative Distribution Functions}
\label{sec:floating-word-ram}

\begin{definition}
\label{def:cda}
Let $\bfmt\,{=}\,(n,\gamma_\bfmt,\phi_\bfmt)$ be any binary number format
and $E, m$ the parameters of a floating-point format $\floatEm$.
A \textit{finite-precision cumulative distribution function}
$F: \bool^n \to {\floatEm} \cap [0,1]$
over $\bfmt$ is a nondecreasing mapping with
$F(\phi_\bfmt(1^n)) = 1$
and $b <_\bfmt b' \implies F(b) \le F(b')$.
A \textit{finite-precision survival function} $S: \bool^n \to \floatEm \cap [0,1]$ over
$\bfmt$ is a nonincreasing mapping with $S(\phi_\bfmt(1^n)) = 0$ and $b <_\bfmt b' \implies S(b') \le S(b)$.
\end{definition}

\Crefrange{lst:exponential-cf}{lst:exponential-sf} of \cref{lst:exponential} show
examples of a finite-precision \WCDF{} and \WSF{}, respectively.
While \cref{def:cda} assumes that $F$ returns IEEE-754 floats,
formats such as fixed-points
(\cref{example:format-fixed-point}) and posits
(\cref{example:format-posit}) are also possible.
The next remarks state properties of a finite-precision CDF.

\begin{remark}
\label{remark:CDA-sample-intractable}
Every finite-precision \WCDF{} $F$ defines a discrete distribution
$P_F: \bool^n \to [0,1]$, where
$P_F(b) \defas F(b) - F(\mathrm{pred}_{\bfmt}(b))$,
with the convention that $F(\mathrm{pred}_{\bfmt}(\phi_\bfmt(0^n))) \defas 0$.
Recall that directly constructing an entropy-optimal DDG tree for $P_F$ is
infeasible if it has $\Theta(2^n)$ leaves.
\end{remark}

\begin{remark}
\label{proposition:lift-cda}
Every finite-precision \WCDF{} $F : \set{0,1}^n \to \floatEm \cap [0,1]$ over a binary number format $\bfmt$
can be lifted to a CDF $\hat{F}: \realext \to [0,1]$ over $\realext$,
where $\hat{F}(x) \defas F(\roundfl{\bfmt,\rtd}(x))$ and
$\roundfl{\bfmt,\rtd}(x) \defas \max_{<_{\bfmt}} \set{b \in \set{0,1}^n \mid \gamma_{\bfmt}(b) \le_{\realext} x}$.
\cref{remark:binary-number-format-ordering} confirms
$\hat{F}$ is monotonic and right-continuous.
\end{remark}

The next result connects binary-coded probability
distributions (\cref{definition:bcpd}) with finite-precision \WCDF{s} (\cref{def:cda}),
which enables the finite-precision specialization of \cref{alg:sampler-opt}
in \cref{sec:floating-sample}.

\begin{proposition}[name=,restate=CdaPcbd]
\label{prop:cda-pcbd}
Let $F: \set{0,1}^n \to \floatEm \cap[0,1]$
be a \WCDF{} over the unsigned integer format $\mathbb{U}_n$
and $P_F$ the corresponding discrete distribution
from~\cref{remark:CDA-sample-intractable}.
The function $p_F : \set{0,1}^* \to [0,1]$ defined below is a binary-coded
probability distribution that satisfies $p_F(b) = P_F(b)$ for all $b \in \bool^n$:
\begin{align}
p_F(b) &\defas F(b1^{n-\abs{b}}) -_\real F((b0^{n-\abs{b}})^-)
&& (b \in \bool^{\le n}),
\label{eq:trefle-1}
\\
p_F(bb') &\defas p_F(b)\mathbf{1}[b'=0\ldots{0}]
&& (b \in \bool^n; b'\in\bool^+),
\label{eq:trefle-2}
\end{align}
where $x^- \defas \mathrm{pred}_{\rm dict}(x)$ for any $x \in \bool^n \setminus \set{0^n}$,
with the convention that $F((0^n)^-) \defas 0$.
\end{proposition}

\begin{proposition}[name=,restate=CdaPcbdUint]
\label{prop:cda-pcbd-unsigned-int}
Let $F: \set{0,1}^n \to \floatEm \cap[0,1]$
be a \WCDF{} over a number format $\bfmt = (n,\gamma_\bfmt, \phi_\bfmt)$.
Then $\widetilde{F} \defas \allowbreak F \circ \phi_\bfmt$ is a \WCDF{} over the
unsigned integer format
$\mathbb{U}_n = (n, (\cdot)_2, \mathrm{id})$ from \cref{example:format-unsigned-int}.
\end{proposition}

\begin{corollary}
\label{corollary:sample-F-transformation}
Let $F:\bool^n\to\floatEm \cap [0,1]$ be a \WCDF{} over a number format
$\bfmt = (n, \gamma_\bfmt, \phi_\bfmt)$.
Then a random variate $X \sim F$ can be
generated by first drawing $Z \sim F \circ \phi_\bfmt$ and setting $X \gets
\phi_\bfmt(Z)$.
\end{corollary}

\subsection{Finite-Precision Random Variate Generation Algorithms}
\label{sec:floating-sample}

%!TEX root=./paper.tex

\begin{listing}[t]
\centering
\setlength{\intextsep}{0pt}
\setlength{\textfloatsep}{0pt}
\algrenewcommand\algorithmicindent{1.0em}%
\begin{minipage}[t]{.48\linewidth}
\begin{algorithm}[H]
\captionsetup{hypcap=false}
\caption{Entropy-Optimal Generation}
\label{alg:sampler-opt-impl}
\begin{algorithmic}[1]
\Require{%
  CDF $F: \bool^n \to \floatEm \cap [0,1]$ \\
  over number format $\bfmt = (n, \gamma_\bfmt, \phi_\bfmt)$ \\
  \color{gray}{String $b \in \bool^{\le n}$; \#Flips $\ell \ge 0$;} \\
  \color{gray}{Floats $\fL, \fR \in \floatEm \cap [0,1]$}
  }
\Ensure{Exact random variate $X \sim F$}
\Function{\SampleOptImpl}{$F$, $b{=}\varepsilon$, $\ell{=}0$, $\fL{=}0$, $\fR{=}1$} \label{algline:sampler-opt-impl}
  \If{$\abs{b} = n$} \Comment{Base Case}         \label{algline:sampler-opt-impl-equiv-begin}
    \State \Return $\phi_\bfmt(b)$
    \Comment{String in Format $\bfmt$}
    \label{algline:sampler-opt-impl-return}
  \EndIf
  \Statex
  \State $b' \gets b01^{n - \abs{b} - 1}$
  \State $\fM \gets F(\phi_\bfmt(b'))$
  \label{algline:sampler-opt-compute-mid}
  \label{algline:sampler-opt-impl-F}
  \If{$\fM = \fR$} \Comment{Leaf} \label{algline:sampler-opt-impl-leaf-opt-start}
    \State \Return $\SampleOptImpl(F, b0, \ell, \fL, \fM)$
      \Comment{0}
  \EndIf
  \If{$\fM = \fL$} \Comment{Leaf}
    \State \Return $\SampleOptImpl(F, b1, \ell, \fM, \fR)$
      \Comment{1}
  \EndIf                              \label{algline:sampler-opt-impl-leaf-opt-end}
                                      \label{algline:sampler-opt-impl-equiv-end}
  % Recursive Case.
  \Statex
  % Only do refinement if \ell > 0.
  \State $\beta_0 \gets \hyperref[alg:exactsubtract1-main]{\ExactSubtractOne}(\fM, \fL)$
    % \Comment{$\fM-_{\real} \fL$}
    \label{algline:sampler-opt-impl-ExactSubtractA}
  \State $\beta_1 \gets \hyperref[alg:exactsubtract1-main]{\ExactSubtractOne}(\fR, \fM)$
    % \Comment{$\fR-_{\real} \fM$}
    \label{algline:sampler-opt-impl-ExactSubtractB}
  \If{$\ell > 0$} \label{algline:sampler-opt-impl-leaf-start}
    \State $a_0 \gets \hyperref[alg:getbit-main]{\GetBit}(\beta_0, {\ell})$
      \Comment{$[\fM-_{\real} \fL]_{\ell}$}
    \State $a_1 \gets \hyperref[alg:getbit-main]{\GetBit}(\beta_1, {\ell})$
      \Comment{$[\fR-_{\real} \fM]_{\ell}$}
    \If{$a_0 = 1 \wedge a_1 = 0$} \Comment{Leaf} \label{algline:sampler-opt-impl-Pattern1}
        \State \Return \SampleOptImpl($F$, $b0$, $\ell$, $\fL$, $\fM$)
          \Comment{0}
    \EndIf
    \If{$a_0 = 0 \wedge a_1 = 1$} \Comment{Leaf} \label{algline:sampler-opt-impl-Pattern2}
      % \Comment{Refine 1}
      \State \Return \SampleOptImpl($F$, $b1$, $\ell$, $\fM$, $\fR$)
        \Comment{1}
        \label{algline:sampler-opt-impl-leaf-end}
    \EndIf
  \EndIf
  \Statex
  \While{$\textbf{true}$} \Comment{Refine Subtree} \label{algline:sampler-opt-impl-subtree-start}
    \State $x \gets \Flip()$; $\ell \gets \ell + 1$
    \State $a_0 \gets \hyperref[alg:getbit-main]{\GetBit}(\beta_0, {\ell})$
      \Comment{$[\fM-_{\real} \fL]_{\ell}$}
    \State $a_1 \gets \hyperref[alg:getbit-main]{\GetBit}(\beta_1, {\ell})$
      \Comment{$[\fR-_{\real} \fM]_{\ell}$}
    \If{$x = 0 \wedge a_0 = 1$} \Comment{Leaf}
      \State \Return \SampleOptImpl($F$, $b0$, $\ell$, $\fL$, $\fM$)
        \Comment{0}
    \EndIf
    \If{$x = 1 \wedge a_1 = 1$} \Comment{Leaf}
      \State \Return \SampleOptImpl($F$, $b1$, $\ell$, $\fM$, $\fR$)%
        \Comment{1}
    \EndIf
  \EndWhile \label{algline:sampler-opt-impl-subtree-end}
\EndFunction
\end{algorithmic}
\end{algorithm}
\end{minipage}\hfill
\begin{minipage}[t]{.495\linewidth}
\begin{algorithm}[H]
\renewcommand{\hl}[1]{#1}
\captionsetup{hypcap=false}
\caption{Preprocessing for \hyperref[alg:getbit-main]{\GetBit}}
\label{alg:exactsubtract1-main}
\begin{algorithmic}[1]
\Require{\hl{$x, x' \in \floatEm \cap [0,1]$, $0 < x -_\real x' < 1$}}
\Ensure{$(n_1, n_2, n_\mathrm{hi}, n_\mathrm{lo}, b_1, b_2, g_\mathrm{hi}, g_\mathrm{lo})$}
\Function{\ExactSubtractOne{}}{$x$, $x'$}
  \State $(s\,e_E \ldots e_1\, f_1\dots f_m)_{\floatEm} \gets x$%
    \label{algline:exactsubtract1-main-ef-st}
  \State $(s'\,e'_E \ldots e'_1\, f'_1\dots f'_m)_{\floatEm} \gets x'$
  \State $e \gets (e_E \ldots e_1)_2$;
         $e' \gets (e'_E \ldots e'_1)_2$
  \State $\hat{e} \gets e - (2^{E-1}-1) + \mathbf{1}[e=0]$
  \State $\hat{e}' \gets e' - (2^{E-1}-1) + \mathbf{1}[e'=0]$
  \State $f \gets (1\,f_1 \ldots f_m)_2 - (\mathbf{1}[e = 0] \ll m)$
  \State $f' \gets (1\,f'_1 \ldots f'_m)_2 - (\mathbf{1}[e' = 0] \ll m)$%
    \label{algline:exactsubtract1-main-ef-ed}
  \State $f'_\mathrm{hi} \gets f' \gg \min\set{\hat{e} - \hat{e}', E+m}$%
    \label{algline:exactsubtract1-main-fpdecomp-st}
  \State $f'_\mathrm{lo} \gets f'\, \&\, ((1 \ll \min\set{\hat{e}-\hat{e}', m+1}) - 1)$%
    \label{algline:exactsubtract1-main-fpdecomp-ed}
  \State \hl{$n_1 \gets -\hat{e}-1 + \mathbf{1}[x=1]$}%
    \label{algline:exactsubtract1-main-final-st}
  \State $n_2 \gets \max\set{(\hat{e}-\hat{e}') - (m+1), 0}$
  \State \hl{$n_\mathrm{hi} \gets m+1 - \mathbf{1}[x=1]$}
  \State $n_\mathrm{lo} \gets \min\set{\hat{e}-\hat{e}', m+1}$
  \State \hl{$b_1 \gets 0$};
         \hl{$b_2 \gets \mathbf{1}[f'_\mathrm{lo} > 0]$}
  \State \hl{$g_\mathrm{hi} \gets f - f'_\mathrm{hi} - b_2$};
         $g_\mathrm{lo} \gets (b_2 \ll n_\mathrm{lo}) - f'_\mathrm{lo}$%
         \label{algline:exactsubtract1-main-final-ed}
  \State \Return $(n_1, n_2, n_\mathrm{hi}, n_\mathrm{lo}, b_1, b_2, g_\mathrm{hi}, g_\mathrm{lo})$
         \label{algline:exactsubtract1-main-return}
\EndFunction
\end{algorithmic}
\end{algorithm}
\vspace{-.05cm}
\begin{algorithm}[H]
\captionsetup{hypcap=false}
\renewcommand{\hl}[1]{#1}
\caption{Extract Binary Digit}
\label{alg:getbit-main}
\begin{algorithmic}[1]
\Require{%
  $\beta {\defas} (n_1, n_2, n_\mathrm{hi}, n_\mathrm{lo}, b_1, b_2, g_\mathrm{hi}, g_\mathrm{lo}),\ell{\ge}1$;
  where
  $\begin{array}[t]{@{}l}
  n_1, n_2, n_\mathrm{hi}, n_\mathrm{lo} \ge 0,\;
  \text{\hl{$b_1, b_2 \in \set{0,1}$}},
  \\
  0 \leq g_\mathrm{hi} < 2^{n_\mathrm{hi}},\; 0 \leq g_\mathrm{lo} < 2^{n_\mathrm{lo}}
  \end{array}$
  \\
  are from
  $\hyperref[alg:exactsubtract1-main]{\ExactSubtractOne}(x,x')$
  }
\Ensure{$\ell$-th bit of $x -_\real x'$ in binary expansion}
\Function{\GetBit}{$\beta,\ell$}
  \If{$\ell \leq n_1$}
    \Return \hl{$b_1$}
  \EndIf
  \If{$\ell \leq n_1 + n_\mathrm{hi}$}
    \Return $g_{\mathrm{hi}, \,\ell - n_1}$
  \EndIf
  \If{$\ell \leq n_1 + n_\mathrm{hi} + n_2$}
    \Return \hl{$b_2$}
  \EndIf
  \If{$\ell \leq n_1 + n_\mathrm{hi} + n_2 + n_\mathrm{lo}$}
    \State \Return $g_{\mathrm{lo}, \,\ell - (n_1 + n_\mathrm{hi} + n_2)}$
  \EndIf
  \State \Return 0
\EndFunction
\end{algorithmic}
\end{algorithm}
\end{minipage}
\vspace{-.4cm}
\end{listing}

\Cref{sec:floating-binary,sec:floating-word-ram} provide all the necessary
ingredients for soundly implementing \cref{alg:sampler-opt}
using finite-precision computation, as shown in \cref{alg:sampler-opt-impl}.
% %
The only required parameter is the \WCDF{} $F: \set{0,1}^n \to \floatEm \cap [0,1]$
over a binary number format $\bfmt$.
% %
The remaining parameters have default values and are used
only by recursive calls, where:
  $b$ stores the bit string generated so far;
  $\ell$ stores the number of calls to $\Flip$, i.e., the current
  level in the underlying DDG tree;
  $\fL, \fR$ store the subtrahend and minuend in the probability
  $p_F(b) = F (\phi_\bfmt(b1^{n-\abs{b}})) -_{\real} F (\phi_\bfmt((b0^{n-\abs{b}})^{-}))$~\cref{eq:trefle-1}.

We now describe \cref{alg:sampler-opt-impl} in detail.
\Crefrange{algline:sampler-opt-impl-equiv-begin}{algline:sampler-opt-impl-return}
show the base case, where the $n$-bit string $b$ (in the unsigned integer
format) has been generated from $F \circ \phi_{\bfmt}$, and then projected to the target
format $\bfmt$ using the map $\phi_{\bfmt}$, based on \cref{prop:cda-pcbd-unsigned-int,corollary:sample-F-transformation}.
\Cref{algline:sampler-opt-compute-mid} %% \Crefrange{algline:sampler-opt-compute-mid-start}{algline:sampler-opt-compute-mid-end}
computes the cumulative probability $f_2$ of the ``midpoint'' string $b'$, which
splits the interval defined by the the current string $b$ in half (based on \cref{eq:trefle-1}),
i.e., $f_2 -_{\real} f_0$ (resp.~$f_1-_{\real}f_2$) is the probability that the next bit is 0 (resp.~1).
\Crefrange{algline:sampler-opt-impl-leaf-opt-start}{algline:sampler-opt-impl-leaf-opt-end}
are optimizations for when one of these probabilities is zero, so the
refined subtree must be a leaf.
\Crefrange{algline:sampler-opt-impl-leaf-start}{algline:sampler-opt-impl-leaf-end}
occur when the refined subtree is a leaf (mirroring
\crefrange{algline:UnivSampler-If1-Start}{algline:UnivSampler-Print1-NoIter}
of \cref{alg:sampler-opt}).
\Crefrange{algline:sampler-opt-impl-subtree-start}{algline:sampler-opt-impl-subtree-end}
occur when the refined subtree is a not leaf (mirroring
\crefrange{algline:UnivSampler-Loop}{algline:UnivSampler-Print1}
of \cref{alg:sampler-opt}).

In \cref{algline:sampler-opt-impl-ExactSubtractA,algline:sampler-opt-impl-ExactSubtractB},
a main implementation challenge is extracting the binary expansions
of $\fM -_\real \fL$ and $\fR -_\real \fM$ in such a way that
avoids expensive arbitrary-precision arithmetic on the one
hand and rounding errors from a direct floating-point subtraction on the other hand.
Whereas the 2Sum/Fast2Sum~\citep{moller1965} methods can be used to compute
the exact round-off error from a floating-point subtraction, they are not
applicable here, as the goal is to extract the individual bits of the
difference.
\Cref{alg:exactsubtract1-main,alg:getbit-main} provide an efficient
solution using fast integer arithmetic.
For $x, x ' \in \floatEm \cap [0,1]$,
\cref{alg:exactsubtract1-main} computes a compact representation of $x -_\real x'$
as described in \cref{eq:exactsubtract1-main} of \cref{prop:exactsubtract1}.
\Crefrange{algline:exactsubtract1-main-ef-st}{algline:exactsubtract1-main-ef-ed}
extract the exponent and significand of $x$ and $x'$ (i.e., $\hat{e}$, $f$, $\hat{e}'$, $f'$),
and \crefrange{algline:exactsubtract1-main-fpdecomp-st}{algline:exactsubtract1-main-fpdecomp-ed}
decompose $f'$ into two parts to align the binary expansions of $x$ and $x'$.
\Crefrange{algline:exactsubtract1-main-final-st}{algline:exactsubtract1-main-return}
then compute a tuple of integers and booleans that encode $x -_\real x'$.
For this tuple and $\ell \geq 1$, \cref{alg:getbit-main}
outputs the $\ell$-th bit of $x -_\real x'$
based on the encoding scheme shown in \cref{eq:exactsubtract1-main}.
\Cref{prop:exactsubtract1} establishes the correctness of \Cref{alg:exactsubtract1-main,alg:getbit-main}, and the guarantee
that all intermediate values fit in a single machine word.

\bgroup
\begin{theorem}[name=,restate=ExactSubtractOneProp]
  \label{prop:exactsubtract1}
  Suppose $x, x' \in \floatEm$ satisfy $0 < x -_\real x' < 1$, and consider any integer $\ell \geq 1$.
  Let $\beta = (n_1, n_2, n_\mathrm{hi}, n_\mathrm{lo}, b_1, b_2, g_\mathrm{hi}, g_\mathrm{lo})$
  be the output of $\hyperref[alg:exactsubtract1-main]{\ExactSubtractOne}(x, x')$ (\cref{alg:exactsubtract1-main}),
  and let $b'$ be the output of $\hyperref[alg:getbit-main]{\GetBit}(\beta, \ell)$ (\Cref{alg:getbit-main}).
  Then,
  \begin{align}    \label{eq:exactsubtract1-main}
    x -_\real x'
    &= \Big(0.\,
    \underset{n_1 \text{ bits}}{\boxed{b_1 \ldots b_1       \vphantom{b_1 g_\mathrm{hi}}       }}
    \underset{n_\mathrm{hi} \text{ bits}}{\boxed{\,\;\; g_\mathrm{hi} \vphantom{b_1 g_\mathrm{hi}} \,\;\;}}
    \underset{n_2 \text{ bits}}{\boxed{b_2 \ldots b_2       \vphantom{b_1 g_\mathrm{hi}}       }}
    \underset{n_\mathrm{lo} \text{ bits}}{\boxed{\,\;\; g_\mathrm{lo} \vphantom{b_1 g_\mathrm{hi}} \,\;\;}}
    \,
    \Big)_2
  \end{align}
  and $b'$ is the $\ell$-th digit of $x-_\real x'$ in binary expansion.
  Also, all intermediate values appearing in both algorithms
  are representable as $(1+E+m)$-bit signed integers.
\end{theorem}
\egroup

The next result establishes the entropy-optimality (and, in turn, soundness)
of \cref{alg:sampler-opt-impl}, combining
\cref{theorem:sampler-opt} (correctness of the infinite version),
\cref{prop:exactsubtract1} (correctness of bit extraction),
\cref{prop:cda-pcbd} (correspondence of CDF and binary-coded probability distribution),
and \cref{prop:cda-pcbd-unsigned-int,corollary:sample-F-transformation}
(generation over $\mathbb{U}_n$ followed by transformation through $\phi_{\bfmt}$).

\begin{theorem}
\label{theorem:sampler-impl-correct-opt}
If $F: \set{0,1}^n \to \floatEm \cap[0,1]$
is a \WCDF{} over a binary number format $\bfmt = (n, \gamma_\bfmt, \phi_\bfmt)$,
then $\hyperref[algline:sampler-opt-impl]{\SampleOptImpl}(F)$ (\cref{alg:sampler-opt-impl})
is an entropy-optimal random variate generator
that returns a string $x \in \set{0,1}^n$ in the format $\bfmt$ with cumulative probability
${F}(x)$.
\end{theorem}

The worst-case entropy cost of \cref{alg:sampler-opt-impl} is
$2^{E-1} + m - 2$ bits, which is the index location of a nonzero bit in
the binary expansion of the smallest nonzero probability in $\floatEm$.
This observation implies that \cref{alg:sampler-opt-impl} halts,
because all the probabilities are dyadic rationals.
The next result establishes a more useful upper bound for
\cref{alg:sampler-opt-impl} in terms of its expected entropy cost.

\begin{theorem}[name=,restate=CostSamplerOptImpl]
\label{corollary:cost-sampler-opt-impl}
The expected entropy cost of \cref{alg:sampler-opt-impl} is at most
$m + 2 - 2^{-2^{E-1} + 3}$ bits.
\end{theorem}

\section{Extended-Accuracy Generation by Leveraging Numerical Survival Functions}
\label{sec:survival}
%!TEX root=./paper.tex

\Cref{alg:sampler-opt-impl} from \cref{sec:floating} requires a
finite-precision \WCDF{} $F$, which computes floating-point probabilities of
intervals $[-\infty, x]$ in the ``left'' tail of the distribution.
Recall, however, that floats have ``high-precision'' near 0 as compared to 1, i.e.,
there are roughly $2^{E-1}$ more floats in $[0,0.5)$ as compared to $[0.5,1)$.
As a result, $F$ more accurately represents the left tail
$[-\infty, x]$ (probabilities near 0), as compared to
the right tail $[x, \infty]$ (probabilities near 1).
To achieve high-accuracy floating-point probabilities in the right tail,
we can combine $F$ with a finite-precision \textit{survival function} $S$
(\cref{def:cda}).

For example, the Rayleigh distribution has the theoretical
range $[0, \infty)$ and CDF $F(t) = 1-\smash{e^{-t^2/2}}$.
Typical floating-point implementations of the CDF and SF
correspond to the following ranges:
\begin{alignat}{4}
&\texttt{standard\_rayleigh\_cdf}&\,&\texttt{= lambda t: -math.expm1(-t*t/2)} \;&\leadsto& [3.50\times 10^{-162}&&, 8.65]
\label{eq:rayleigh-cdf}
\\[-5pt]
&\texttt{standard\_rayleigh\_sf}&\,&\texttt{= lambda t: math.exp(-t*t/2)} \;&\leadsto& [1.05\times10^{-8} &&, 38.60]
\label{eq:rayleigh-sf}
\end{alignat}
\Cref{eq:rayleigh-cdf,eq:rayleigh-sf} show that
the combined range $[3.5\times 10^{-162}, 38.60]$ of these complimentary
specifications is far more accurate than using only the CDF or SF.
Another illustrative example is symmetric distributions.
Consider the \WCDF{} (\texttt{gsl\_cdf\_gaussian\_P}) and \WSF{}
(\texttt{gsl\_cdf\_gaussian\_Q}) of a Gaussian from the
GSL (\cref{lst:gsl-samplers}).
For $\texttt{sigma}=1$, the theoretical range is $(-\infty, \infty)$, but the floating-point
ranges are $[-37.52, 8.29]$ and $[-8.29, 37.52]$, respectively.

In the infinite-precision Real-RAM model, the CDF $F$ and SF $S$
of a random variable $X$ can be combined
by using the following property, which holds
for every real ``cutoff-point'' $t^* \in \real$:
\begin{align}
\Pr(X \le t) &= (1-\mathbf{1}[t \ge t^*])F(t) + \mathbf{1}[t \ge t^*](1-S(t))
&&(t \in \real).
\label{eq:kosher}
\end{align}
Combining $F$ and $S$ as in \cref{eq:kosher}
must be done with caution in the finite-precision setting,
because $(1-S(t)) \notin \floatEm$ for
many values of $t \in \bfmt$. % t \in \floatEm
We address this challenge by introducing a \ECDF{} (\cref{def:edf}),
which is a combined representation for
$(F(t), 1-S(t))$ that avoids an explicit inexact floating-point subtraction.
The key idea is to use $F$ and $S$ to represent the left
and right tails, respectively, of the distribution, with a
cutoff point $b^* \in \bfmt$ that is the exact median of $F$.

\begin{definition}
\label{def:edf}
A \textit{finite-precision dual distribution function}
(\ECDF{}) over a binary number format $\bfmt = (n, \gamma_\bfmt, \phi_\bfmt)$ is a
mapping $G: \bool^n \to \set{0,1} \times \big(\floatEm \cap [0,1/2]\big)$
such that
$G^*(\phi_\bfmt(1^n)) \,{=}\, 1$ and
$b \,{<_\bfmt}\, b' \;{\Longrightarrow}\; G^*(b) \,{\le}\, G^*(b')$,
where $G^* \,{:}\, \bool^n \,{\to}\, [0,1]$ is defined by
$G^*(b) \defas (1-d)f + d(1-f)$
for $b \in \bool^n$ and $(d,f) \defas G(b)$.
\end{definition}

\begin{remark}
\label{remark:ecdf-double-values}
A finite-precision \WCDF{} $F$ returning floating-point probabilities in $\floatEm$ can represent
a distribution with at most $(2^{E-1}-1)2^m$ outcomes, while a finite-precision
\ECDF{} $G$ can represent \textit{twice} as many outcomes.
The representable probabilities are always integer multiples of $\smash{2^{-(2^{E-1}+m-2)}}$.
\end{remark}

%!TEX root=./paper.tex

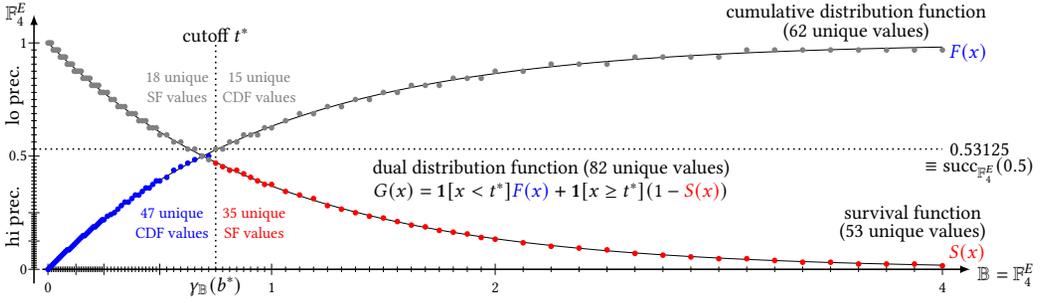
\begin{figure}[t]
\centering
\begin{adjustbox}{max width=\linewidth}
\begin{tikzpicture}

% Cutoff points.
\def\xmid{0.75}
\def\ymid{0.53125}

% Scale factors
\def\sx{4}
\def\sy{4}

% AXES
\draw[thick,-latex] (-.10,0) -- node[pos=1,yshift=-.1cm,anchor=west]{$\bfmt= \mathbb{F}^{E}_{4}$} (16.5,0);
\draw[thick,-latex] (-.25,-.25)
  --
  node[pos=.25,rotate=90,anchor=south,inner ysep=5pt]{hi prec.}
  node[pos=0.7,rotate=90,anchor=south,inner ysep=5pt]{lo prec.}
  node[pos=1,anchor=east]{$\mathbb{F}^{E}_{4}$}
  (-.25,\sy*1.125);

% F(x) continuous.
\draw[domain=0:\sx*\xmid, smooth, variable=\x, black] plot ({\x}, {\sy*(1-exp(-\x/\sx))});
\draw[domain=\sx*\xmid:\sx*4, smooth, variable=\x, black] plot ({\x}, {\sy*(1-exp(-\x/\sx))});
\node[at={(\sx*4,\sy*1)},anchor=160,draw=none,color=blue,label={[anchor=east,xshift=0.45cm,yshift=0.25cm]above:\begin{tabular}{@{}c@{}}cumulative distribution function\\[-2.5pt](62 unique values)\end{tabular}}]{$F(x)$};

% S(x) continuous.
\draw[domain=0:\sx*\xmid, smooth, variable=\x, black] plot ({\x}, {\sy*(exp(-\x/\sx))});
\draw[domain=\sx*\xmid:\sx*4, smooth, variable=\x, black] plot ({\x}, {\sy*(exp(-\x/\sx))});
\node[at={(\sx*4,\sy*0)},anchor=south west,draw=none,color=red,label={[anchor=east,xshift=0.45cm,yshift=0.25cm,]above:\begin{tabular}{@{}c@{}}survival function\\[-2.5pt](53 unique values)\end{tabular}}]{$S(x)$};

\node[at={(2.25*\sx,0.5*\sy)},anchor=north,inner sep=0pt]{
  \begin{tabular}{c}
  dual distribution function (82 unique values)\\
  $G(x) = \mathbf{1}[x < t^*]\textcolor{blue}{F(x)} + \mathbf{1}[x \ge t^*](1-\textcolor{red}{S(x)})$
  \end{tabular}
};

% CUTOFF
\draw[thick,dotted] (\sx*\xmid,0) -- node[pos=0,anchor=north,inner xsep=0pt]{$\gamma_\bfmt(b^*)$} node[pos=1,anchor=south,inner sep=0pt]{cutoff $t^*$} (\sx*\xmid,\sy*1);
\draw[thick,dotted] (-0.25,\sy*\ymid) -- node[pos=1,right,label={[label distance=0pt,inner sep=0pt,yshift=.1cm]below:$\equiv\mathrm{succ}_{\mathbb{F}^E_4}(0.5)$}]{$\ymid$} (\sx*4,\sy*\ymid);

% Y-AXIS
\foreach \y [
  count=\i,
  evaluate={
    \n=(mod(\i-1,16)==0 ? .1 : .05);
    }] in {0.0,0.0078125,0.015625,0.0234375,0.03125,0.0390625,0.046875,0.0546875,0.0625,0.0703125,0.078125,0.0859375,0.09375,0.1015625,0.109375,0.1171875,0.125,0.1328125,0.140625,0.1484375,0.15625,0.1640625,0.171875,0.1796875,0.1875,0.1953125,0.203125,0.2109375,0.21875,0.2265625,0.234375,0.2421875,0.25,0.265625,0.28125,0.296875,0.3125,0.328125,0.34375,0.359375,0.375,0.390625,0.40625,0.421875,0.4375,0.453125,0.46875,0.484375,0.5,0.53125,0.5625,0.59375,0.625,0.65625,0.6875,0.71875,0.75,0.78125,0.8125,0.84375,0.875,0.90625,0.9375,0.96875,1} {
  \draw ($(-\n,0) + (-0.25, \sy*\y)$) -- ($(\n, 0) + (-0.25, \sy*\y)$);
}
\foreach \y in {0, 0.5, 1} {\node[anchor=east, font=\footnotesize, at={(-0.25, \sy*\y)}]{\y};}

% X-AXIS
\foreach \x [
  count=\i,
  evaluate={
    \n=(mod(\i-1,16)==0 ? .1 : .05);
        }] in {0.0,0.0078125,0.015625,0.0234375,0.03125,0.0390625,0.046875,0.0546875,0.0625,0.0703125,0.078125,0.0859375,0.09375,0.1015625,0.109375,0.1171875,0.125,0.1328125,0.140625,0.1484375,0.15625,0.1640625,0.171875,0.1796875,0.1875,0.1953125,0.203125,0.2109375,0.21875,0.2265625,0.234375,0.2421875,0.25,0.265625,0.28125,0.296875,0.3125,0.328125,0.34375,0.359375,0.375,0.390625,0.40625,0.421875,0.4375,0.453125,0.46875,0.484375,0.5,0.53125,0.5625,0.59375,0.625,0.65625,0.6875,0.71875,0.75,0.78125,0.8125,0.84375,0.875,0.90625,0.9375,0.96875,1,1.0625,1.125,1.1875,1.25,1.3125,1.375,1.4375,1.5,1.5625,1.625,1.6875,1.75,1.8125,1.875,1.9375,2.0,2.125,2.25,2.375,2.5,2.625,2.75,2.875,3.0,3.125,3.25,3.375,3.5,3.625,3.75,3.875,4} {
  \draw (\sx*\x,-\n) -- (\sx*\x,\n);
}
\foreach \x in {0, 1, 2, 4} {\node[anchor=north, font=\footnotesize, at={(\sx*\x,-0.1)}]{\x};}

% Draw F(x).
\foreach \x/\y [
  evaluate={
    \n=(\x<\xmid ? "blue" : "gray");
        }] in {0.0/0.0,0.0078125/0.0078125,0.015625/0.015625,0.0234375/0.0234375,0.03125/0.03125,0.0390625/0.0390625,0.046875/0.046875,0.0546875/0.0546875,0.0625/0.0625,0.0703125/0.0703125,0.078125/0.078125,0.0859375/0.0859375,0.09375/0.0859375,0.1015625/0.09375,0.109375/0.1015625,0.1171875/0.109375,0.125/0.1171875,0.1328125/0.125,0.140625/0.1328125,0.1484375/0.140625,0.15625/0.1484375,0.1640625/0.1484375,0.171875/0.15625,0.1796875/0.1640625,0.1875/0.171875,0.1953125/0.1796875,0.203125/0.1875,0.2109375/0.1875,0.21875/0.1953125,0.2265625/0.203125,0.234375/0.2109375,0.2421875/0.21875,0.25/0.21875,0.265625/0.234375,0.28125/0.2421875,0.296875/0.25,0.3125/0.265625,0.328125/0.28125,0.34375/0.296875,0.359375/0.296875,0.375/0.3125,0.390625/0.328125,0.40625/0.328125,0.421875/0.34375,0.4375/0.359375,0.453125/0.359375,0.46875/0.375,0.484375/0.390625,0.5/0.390625,0.53125/0.40625,0.5625/0.4375,0.59375/0.453125,0.625/0.46875,0.65625/0.484375,0.6875/0.5,0.71875/0.5,0.75/0.53125,0.78125/0.53125,0.8125/0.5625,0.84375/0.5625,0.875/0.59375,0.90625/0.59375,0.9375/0.59375,0.96875/0.625,1/0.625,1.0625/0.65625,1.125/0.6875,1.1875/0.6875,1.25/0.71875,1.3125/0.71875,1.375/0.75,1.4375/0.75,1.5/0.78125,1.5625/0.78125,1.625/0.8125,1.6875/0.8125,1.75/0.8125,1.8125/0.84375,1.875/0.84375,1.9375/0.84375,2.0/0.875,2.125/0.875,2.25/0.90625,2.375/0.90625,2.5/0.90625,2.625/0.9375,2.75/0.9375,2.875/0.9375,3.0/0.9375,3.125/0.96875,3.25/0.96875,3.375/0.96875,3.5/0.96875,3.625/0.96875,3.75/0.96875,3.875/0.96875,4/0.96875} {
  \node[circle,inner sep=1pt,fill=\n,at={(\sx*\x,\sy*\y)}]{};
}
\node[color=gray,at={(\sx*\xmid,\sy*0.8)}, right,font=\footnotesize]{\begin{tabular}{@{}c@{}}15 unique\\CDF values\end{tabular}};
\node[color=gray,at={(\sx*\xmid,\sy*0.8)}, left,font=\footnotesize]{\begin{tabular}{@{}c@{}}18 unique\\SF values\end{tabular}};

\node[color=blue,at={(\sx*\xmid,\sy*0.2)}, left,font=\footnotesize]{\begin{tabular}{@{}c@{}}47 unique\\CDF values\end{tabular}};
\node[color=red,at={(\sx*\xmid,\sy*0.2)}, right,font=\footnotesize]{\begin{tabular}{@{}c@{}}35 unique\\SF values\end{tabular}};

% Draw S(x)
\foreach \x/\y [
  evaluate={
    \n=(\xmid<=\x ? "red" : "gray");
        }] in {0.0/1.0,0.0078125/1.0,0.015625/1.0,0.0234375/0.96875,0.03125/0.96875,0.0390625/0.96875,0.046875/0.96875,0.0546875/0.9375,0.0625/0.9375,0.0703125/0.9375,0.078125/0.9375,0.0859375/0.90625,0.09375/0.90625,0.1015625/0.90625,0.109375/0.90625,0.1171875/0.875,0.125/0.875,0.1328125/0.875,0.140625/0.875,0.1484375/0.875,0.15625/0.84375,0.1640625/0.84375,0.171875/0.84375,0.1796875/0.84375,0.1875/0.84375,0.1953125/0.8125,0.203125/0.8125,0.2109375/0.8125,0.21875/0.8125,0.2265625/0.8125,0.234375/0.78125,0.2421875/0.78125,0.25/0.78125,0.265625/0.78125,0.28125/0.75,0.296875/0.75,0.3125/0.71875,0.328125/0.71875,0.34375/0.71875,0.359375/0.6875,0.375/0.6875,0.390625/0.6875,0.40625/0.65625,0.421875/0.65625,0.4375/0.65625,0.453125/0.625,0.46875/0.625,0.484375/0.625,0.5/0.59375,0.53125/0.59375,0.5625/0.5625,0.59375/0.5625,0.625/0.53125,0.65625/0.53125,0.6875/0.5,0.71875/0.484375,0.75/0.46875,0.78125/0.453125,0.8125/0.4375,0.84375/0.4375,0.875/0.421875,0.90625/0.40625,0.9375/0.390625,0.96875/0.375,1/0.375,1.0625/0.34375,1.125/0.328125,1.1875/0.3125,1.25/0.28125,1.3125/0.265625,1.375/0.25,1.4375/0.234375,1.5/0.2265625,1.5625/0.2109375,1.625/0.1953125,1.6875/0.1875,1.75/0.171875,1.8125/0.1640625,1.875/0.15625,1.9375/0.140625,2.0/0.1328125,2.125/0.1171875,2.25/0.1015625,2.375/0.09375,2.5/0.0859375,2.625/0.0703125,2.75/0.0625,2.875/0.0546875,3.0/0.046875,3.125/0.046875,3.25/0.0390625,3.375/0.03125,3.5/0.03125,3.625/0.0234375,3.75/0.0234375,3.875/0.0234375,4/0.015625} {
  \node[circle,inner sep=1pt,fill=\n,at={(\sx*\x,\sy*\y)}]{};
}

\end{tikzpicture}
\end{adjustbox}
\captionsetup{skip=5pt}
\caption{A dual distribution function $G$ combines a finite-precision
cumulative distribution $F$ and survival function $S$ to represent
probabilities in the left tail (blue; below the median cutoff) and
right tail (red; above the median cutoff), respectively.
This combination ensures that all explicitly represented floats lie in the
high-precision range $[0,1/2]$, which in turn supports more unique values
for random variate generation.}
\label{fig:survival}
\end{figure}

\begin{theorem}[name=,restate=EcdaDefault]
\label{theorem:ecda-default}
Let $F$ be a \WCDF{} and $S$ a \WSF{} over a binary number format $\bfmt$, such that
$S(b^*)<1/2$ for some cutoff
$b^* \defas \hyperref[alg:quantile]\Quantile(F, \mathrm{succ}_{\floatEm}(0.5)) \in \set{0,1}^n$.
A sound \ECDF{} $G$ over $\bfmt$ satisfying \cref{def:edf} is
\bgroup
\setlength{\abovedisplayskip}{0pt}
\setlength{\belowdisplayskip}{2pt}
\begin{align}
G(b) \defas (0,F(b)) \mbox{ if } b <_\bfmt b^*,
&&
G(b) \defas (1, S(b)) &\mbox{ if } b \ge_\bfmt b^*
&&
(b \in \bool^n).
\label{eq:erinite-1}
\qquad \qedsymbol
\end{align}
\egroup
\noqed
\end{theorem}

\Cref{fig:survival} shows a \ECDF{} for the exponential distribution
(\cref{lst:exponential}).
Colored dots show the selected
floating-point probabilities in $[0,0.5)$.
Solid lines show the underlying real functions.

\Cref{alg:sampler-opt-impl} must be modified to soundly generate from a
\ECDF{} $G$, which returns a pair $c = (d,f)$ denoting $f$ if $d=0$ or
$1-f$ if $d=1$, as described in
\cref{remark:ecdf-samplers-changes,alg:sampler-opt-ext-impl}
of \cref{appx:survival}.
\Cref{alg:quantile} gives a fast implementation of
\hyperref[alg:quantile]{\Quantile} for any \WCDF{}, which is used
to obtain the cutoff $b^*$ in \cref{theorem:ecda-default};
and \cref{alg:quantile-ext}
shows the quantile of a \ECDF.
As compared to heuristic methods such as \texttt{exp\_qf}
(\cref{lst:exponential})
or
\texttt{gsl\_cdf\_gaussian\_Pinv}
(\cref{lst:gsl-samplers}),
our quantile computations are \textit{exact} for the implemented
random variate generator (cf.~\cref{fig:overview}).

\section{Evaluation}
\label{sec:evaluation}
%!TEX root=./paper.tex

%!TEX root=./paper.tex

\begin{table}[t]
\newcommand{\agslName}{GSL}
\newcommand{\naiveName}{CBS}
\newcommand{\optName}{OPT}
\setlength{\intextsep}{0pt}
\setlength{\textfloatsep}{0pt}
\sisetup{
  round-mode = places,
  round-precision = 2,
  exponent-product = \times,
}

\newcommand{\sci}[1]{\num[scientific-notation=true]{#1}}
\newcommand{\dd}{$\smash{{}^\dagger}$}

\captionsetup{skip=0pt,belowskip=0pt,aboveskip=0pt}
\caption{
  Comparison of
  optimal generation (OPT, \cref{alg:sampler-opt-impl}) with two baselines:
  the GNU scientific library (GSL~\citep{galassi2009}) and
  conditional bit sampling (CBS~\citep{Sobolewski1972}, \cref{alg:sampler-naive-impl}).
  \WCDF{} implementations used for OPT and CBS are from the GSL
  (cf.~\cref{lst:gsl-samplers}), which both generate exact random
  variates from the CDF.
  In terms of \cref{fig:real-world}: the bits/variate column reports the
  input entropy consumption rate (lower=better), and the variates/sec
  column reports the output random variate generation rate (higher=better).
  \mbox{\dd=``discrete'' distribution.}
  }
\label{table:efficiency}
\begin{adjustbox}{max width=\linewidth}
\begin{tabular}{|llrr||llrr|}
\hline
\textbf{Distribution} & \textbf{Method} & \textbf{Bits/Variate} & {\textbf{Variates/Sec}} &
\textbf{Distribution} & \textbf{Method} & \textbf{Bits/Variate} & {\textbf{Variates/Sec}}
\\
\hline\hline
Beta(5, 5)           & \agslName  & \num{262.800640} & \sci{5.014039e+05} & Gumbel2(1, 5)           & \agslName  & \num{64.000000}  & \sci{1.371366e+06} \\
                     & \naiveName & \num{52.1016}    & \sci{2.800854e+04} &                         & \naiveName & \num{49.2613}    & \sci{4.583750e+04} \\
                     & \optName   & \num{24.9796}    & \sci{5.418850e+04} &                         & \optName   & \num{24.9910}    & \sci{1.717092e+05} \\\hline
Binomial(.2, 100)\dd & \agslName  & \num{224.793600} & \sci{4.978592e+05} & Hypergeom(5, 20, 7)\dd  & \agslName  & \num{447.992960} & \sci{3.049059e+05} \\
                     & \naiveName & \num{15.7614}    & \sci{3.146395e+04} &                         & \naiveName & \num{6.2527}     & \sci{1.088815e+05} \\
                     & \optName   & \num{5.1083}     & \sci{3.617906e+04} &                         & \optName   & \num{3.0068}     & \sci{1.424786e+05} \\\hline
Cauchy(7)            & \agslName  & \num{64.000000}  & \sci{1.362212e+06} & Laplace(2)              & \agslName  & \num{64.000000}  & \sci{1.464987e+06} \\
                     & \naiveName & \num{51.4511}    & \sci{4.835754e+04} &                         & \naiveName & \num{47.8329}    & \sci{5.036312e+04} \\
                     & \optName   & \num{25.0024}    & \sci{2.208285e+05} &                         & \optName   & \num{25.0027}    & \sci{2.871748e+05} \\\hline
ChiSquare(13)        & \agslName  & \num{64.000000}  & \sci{1.244555e+06} & Logistic(.5)            & \agslName  & \num{64.000000}  & \sci{1.387347e+06} \\
                     & \naiveName & \num{47.4285}    & \sci{2.645146e+04} &                         & \naiveName & \num{48.8009}    & \sci{4.689904e+04} \\
                     & \optName   & \num{24.9947}    & \sci{5.190168e+04} &                         & \optName   & \num{24.9661}    & \sci{2.040483e+05} \\\hline
Exponential(15)      & \agslName  & \num{64.000000}  & \sci{1.386194e+06} & Lognormal(1, 1)         & \agslName  & \num{163.015680} & \sci{7.110858e+05} \\
                     & \naiveName & \num{48.5612}    & \sci{4.608677e+04} &                         & \naiveName & \num{49.2715}    & \sci{4.100512e+04} \\
                     & \optName   & \num{24.9786}    & \sci{2.328018e+05} &                         & \optName   & \num{24.9828}    & \sci{1.870522e+05} \\\hline
ExpPow(1, .5)        & \agslName  & \num{197.025920} & \sci{5.974786e+05} & NegBinomial(.71, 18)\dd & \agslName  & \num{665.830400} & \sci{2.165346e+05} \\
                     & \naiveName & \num{47.3099}    & \sci{3.574199e+04} &                         & \naiveName & \num{12.5432}    & \sci{4.008337e+04} \\
                     & \optName   & \num{25.0075}    & \sci{8.674983e+04} &                         & \optName   & \num{4.6883}     & \sci{4.601234e+04} \\\hline
Fdist(5, 2)          & \agslName  & \num{268.950400} & \sci{4.704111e+05} & Pareto(3,2)             & \agslName  & \num{64.000000}  & \sci{1.409245e+06} \\
                     & \naiveName & \num{51.4517}    & \sci{2.630188e+04} &                         & \naiveName & \num{45.9211}    & \sci{5.347050e+04} \\
                     & \optName   & \num{25.0017}    & \sci{6.294296e+04} &                         & \optName   & \num{24.9907}    & \sci{2.299062e+05} \\\hline
Flat(-7, 3)          & \agslName  & \num{64.000000}  & \sci{1.447807e+06} & Pascal(1, 5)\dd         & \agslName  & \num{195.592960} & \sci{5.000000e+05} \\
                     & \naiveName & \num{43.5235}    & \sci{5.657709e+04} &                         & \naiveName & \num{0.0000}     & \sci{3.125283e+04} \\
                     & \optName   & \num{24.9832}    & \sci{4.834655e+05} &                         & \optName   & \num{0.0000}     & \sci{2.092400e+05} \\\hline
Gamma(.5, 1)         & \agslName  & \num{198.261760} & \sci{6.235191e+05} & Poisson(71)\dd          & \agslName  & \num{697.221120} & \sci{1.896705e+05} \\
                     & \naiveName & \num{57.0001}    & \sci{1.403867e+04} &                         & \naiveName & \num{18.3205}    & \sci{2.067889e+04} \\
                     & \optName   & \num{25.0094}    & \sci{1.799603e+04} &                         & \optName   & \num{6.1858}     & \sci{2.306055e+04} \\\hline
Gaussian(15)         & \agslName  & \num{162.732800} & \sci{7.553441e+05} & Rayleigh(11)            & \agslName  & \num{64.000000}  & \sci{1.442377e+06} \\
                     & \naiveName & \num{46.4141}    & \sci{4.953609e+04} &                         & \naiveName & \num{48.5165}    & \sci{5.076529e+04} \\
                     & \optName   & \num{25.0026}    & \sci{2.326880e+05} &                         & \optName   & \num{24.9860}    & \sci{2.171317e+05} \\\hline
Geometric(.4)\dd     & \agslName  & \num{64.000000}  & \sci{1.376273e+06} & Tdist(5)                & \agslName  & \num{279.768960} & \sci{4.386542e+05} \\
                     & \naiveName & \num{6.0647}     & \sci{2.033306e+05} &                         & \naiveName & \num{49.5592}    & \sci{2.646889e+04} \\
                     & \optName   & \num{3.7760}     & \sci{3.289257e+05} &                         & \optName   & \num{25.0230}    & \sci{4.902658e+04} \\\hline
Gumbel1(1,1)         & \agslName  & \num{64.000000}  & \sci{1.411632e+06} & Weibull(2, 3)           & \agslName  & \num{64.000000}  & \sci{1.386578e+06} \\
                     & \naiveName & \num{50.2938}    & \sci{4.796324e+04} &                         & \naiveName & \num{55.3537}    & \sci{4.107248e+04} \\
                     & \optName   & \num{24.9979}    & \sci{2.356268e+05} &                         & \optName   & \num{24.9653}    & \sci{1.475231e+05} \\\hline
                 \hline
\end{tabular}
\end{adjustbox}
\vspace{-.5cm}
\end{table}

We implemented a C library (\url{https://github.com/probsys/librvg}) with
all the algorithms described in this article.
\Cref{lst:exponential,lst:gsl-samplers}
shows examples of the programming interface,
using the macros \texttt{GENERATE\_FROM\_(CDF|SF|DDF)}.
Our evaluation~\citep{artifact25} investigates the following research
questions.
\begin{enumerate}[label=\textbf{(Q\arabic*)},wide=0pt,leftmargin=*]
\item\label{question:efficiency}
How does the entropy-optimal method (\cref{alg:sampler-opt-impl})
compare to exact conditional bit sampling~\citep[\S{II.B}]{Sobolewski1972}
(\cref{alg:sampler-naive-impl})
and the inexact GSL generators~\citep{galassi2009}, in
  terms of \begin{enumerate*}[label=(\roman*)]
    \item input bits per output variate; and
    \item output variates per wall-clock second?
  \end{enumerate*} (\cref{sec:evaluation-efficiency})

\item\label{question:range}
How do the ranges of random variate generators specified
  by a finite-precision \WCDF{}, \WSF{}, and \ECDF{} compare to one another, and to those of
  GSL generators?
  (\cref{sec:evaluation-range})

\item\label{question:ratios}
How large is the entropy and runtime overhead
  when using the extended-accuracy variants of
  conditional bit sampling and entropy-optimal generators
  (\cref{alg:sampler-naive-ext-impl,alg:sampler-opt-ext-impl})
  that use a \ECDF{}
  described in \cref{sec:survival}, instead of the original generators
  (\cref{alg:sampler-naive-impl,alg:sampler-opt-impl}) described
  in \cref{sec:floating} that use a \WCDF?
  (\cref{sec:evaluation-ratios})

\end{enumerate}

\subsection{\texorpdfstring%
  {\labelcref{question:efficiency} Input Entropy Rate and Output Generation Rate}%
  {(Q1) Input Entropy Rate and Output Generation Rate}%
  }
\label{sec:evaluation-efficiency}

%% discrete (6)
%% - Binomial, Geometric, Hypergeom, NegBinomial, Pascal, Poisson.
%% continuous (18)
%% - Beta, Cauchy, ChiSquare, Exponential, ExpPow, Fdist, Flat, Gamma, Gaussian,
%% - Gumbel1, Gumbel2, Laplace, Logistic, Lognormal, Pareto, Rayleigh, Tdist, Weibull.
\Cref{table:efficiency} shows measurements for 18 representative
``continuous'' distributions and 6 representative ``discrete''
distributions.
The entropy source used in this experiment is a GSL pseudorandom number
generator (PRNG) that calls \texttt{/dev/urandom} to obtain 8 random bytes stored
in a 64-bit word.

\paragraph{Input Entropy Rate}
In terms of input bits/variate (lower is better), the
conditional bit sampling (CBS) baseline is
1x--3.1x more expensive than \cref{alg:sampler-opt-impl} (OPT).
An interesting finding is that the optimal generator draws around 25
bits on average for the 18 ``continuous'' distributions, which is two bits
higher than the 23-bit mantissa in IEEE-754 single-precision format used to
represent the output of the CDF (cf.~\cref{corollary:cost-sampler-opt-impl}).
This finding suggests that the GSL CDF implementations are close to the maximum entropy distributions
identified by \cref{corollary:cost-sampler-opt-impl}.
The GSL generators are 2.6x--142x more expensive in terms
of bits/variate as compared to OPT.
For Pascal(1,5), a deterministic distribution, the GSL draws
195.59 bits/variate, whereas OPT uses zero.
These large differences in entropy cost highlight fundamental
inefficiencies of Real-RAM algorithms in the GSL.
Even though a single
infinitely precise uniform random variable in $[0,1]$ contains the same
amount of entropy as countably many such variates
(\cref{remark:random-variable-domain}), this cost equivalence
does not hold in finite-precision implementations, where each
``floating-point'' uniform requires many random bits (e.g., 32, 53, 64; \cref{remark:uniform-real-world}).
The GSL generators that require 64 bits/variate (e.g., Cauchy, Geometric,
Pareto, Weibull) always use exactly one floating-point uniform.
The more expensive GSL generators use rejection sampling or special
relationships between random variables (e.g., Beta is a ratio of
Gammas; NegBinomial uses a Gamma and Poisson)
further driving up the entropy cost.

\paragraph{Output Generation Rate}
In terms of output variates/sec (higher is better),
OPT is 1.1x--8.5x faster than CBS.
Both methods evaluate the \WCDF{} $F$ the same number of times per output
variate.
The runtime of CBS is driven by the high overhead of
computing conditional probabilities for the chain rule,
which requires expensive arbitrary-precision arithmetic.
A main cost of OPT is extracting bits from exact
differences of floats using
\cref{alg:exactsubtract1-main,alg:getbit-main}.
The GSL generators deliver the fastest output generation rate (2.14x--34.6x
higher, median 6.4x) as they do not compute $F$,
but offer no formal guarantees
(\cref{sec:overview-software}) and lower accuracy (\cref{sec:evaluation-range})
as compared to OPT.
While entropy cost gives a theoretically precise runtime measure through
the DDG tree formalism, wall-clock is dictated by many implementation
details (e.g., caching, parallelism, vectorization, PRNG cost, CDF
evaluation cost, programming language, etc.) that could be
further optimized in our prototype.

%!TEX root=./paper.tex

\begin{table}[t]
\newcommand{\sci}[1]{\num[round-mode = places,round-precision = 2,scientific-notation=true]{#1}}
\newcommand{\nosci}[1]{\num[round-mode = places,round-precision = 2,scientific-notation=fixed,fixed-exponent=0]{#1}}
\def\factor{3}
\newcommand{\interval}[3]{%
  \tikz[draw=#3,fill=#3]{%
    \node[name=a,rectangle, fill, inner sep=1.5pt, outer sep=0pt,at={(\factor*#1,0)}]{};%
    \node[name=b,rectangle, fill, inner sep=1.5pt, outer sep=0pt,at={(\factor*#2,0)}]{};%
    \draw[line width=1pt] (a.center) -- (b.center);%
    \draw[draw=none] (0,0) -- (\factor,0);}}

\newcommand{\agslName}{GSL}
\newcommand{\cdfName}{CDF}
\newcommand{\sfName}{SF}
\newcommand{\ecdfName}{DDF}

\newcommand{\agslColor}{black!90!white}
\newcommand{\cdfColor}{red!90!white}
\newcommand{\sfColor}{blue!90!white}
\newcommand{\ecdfColor}{ForestGreen}
\sisetup{round-mode=places, round-precision=0}

\captionsetup{skip=0pt}
\caption{Comparison of random variate generators from the GNU Scientific Library (GSL)
and exact random variate generators for a finite-precision cumulative
distribution function (CDF), survival function (SF), or a
combination of the two (DDF) on 13 probability distributions. The
random variate range shows the minimum and maximum values of the
output of each generator. Intervals visualized on a log scale.}
\label{table:support}
\scriptsize
\begin{tabular*}{\linewidth}{|l@{\extracolsep{\fill}}lrclr|}
\hline
\textbf{Distribution} & \textbf{Method} & \multicolumn{3}{c}{\textbf{Random Variate Range}} & \textbf{Analysis Time} \\ \hline\hline
Cauchy(1)             & \agslName       & \sci{-1.367131e+09}                               & \interval{0.398}{0.602288}{\agslColor} & \sci{1.367130e+09}   & \qty{40.558809}{\second} \\
\;$(-\infty,\infty)$  & \cdfName        & \sci{-4.543071e+44}                               & \interval{0.000}{0.578695}{\cdfColor}  & \sci{1.068071e+07}   & $<$\qty{50}{\micro\second}   \\
                      & \sfName         & \sci{-1.068071e+07}                               & \interval{0.421}{1.000000}{\sfColor}   & \sci{4.543071e+44}   & $<$\qty{50}{\micro\second}   \\
                      & \ecdfName       & \sci{-4.543071e+44}                               & \interval{0.000}{1.000000}{\ecdfColor} & \sci{4.543071e+44}   & $<$\qty{50}{\micro\second}   \\ \hline
Exponential(1)        & \agslName       & \sci{0.000000e+00}                                & \interval{0.000}{0.667343}{\agslColor} & \nosci{2.218071e+01} & \qty{36.474485}{\second} \\
\;$(0,\infty)$        & \cdfName        & \sci{7.006492e-46}                                & \interval{0.000}{0.614188}{\cdfColor}  & \nosci{1.732868e+01} & $<$\qty{50}{\micro\second}   \\
                      & \sfName         & \sci{2.980232e-08}                                & \interval{0.000}{1.000000}{\sfColor}   & \nosci{1.039721e+02} & $<$\qty{50}{\micro\second}   \\
                      & \ecdfName       & \sci{7.006492e-46}                                & \interval{0.000}{1.000000}{\ecdfColor} & \nosci{1.039721e+02} & $<$\qty{50}{\micro\second}   \\ \hline
Flat(.1, 3.14)        & \agslName       & \nosci{1.000000e-01}                              & \interval{0.000}{1.000000}{\agslColor} & \nosci{3.140000e+00} & \qty{19.215090}{\second} \\
\;$(.1, 3.14)$        & \cdfName        & \nosci{1.000000e-01}                              & \interval{0.000}{1.000000}{\cdfColor}  & \nosci{3.140000e+00} & $<$\qty{50}{\micro\second}   \\
                      & \sfName         & \nosci{1.000001e-01}                              & \interval{0.000}{1.000000}{\sfColor}   & \nosci{3.140000e+00} & $<$\qty{50}{\micro\second}   \\
                      & \ecdfName       & \nosci{1.000000e-01}                              & \interval{0.000}{1.000000}{\ecdfColor} & \nosci{3.140000e+00} & $<$\qty{50}{\micro\second}   \\ \hline
Gumbel1(1,1)          & \agslName       & \nosci{-3.099223e+00}                             & \interval{0.065}{0.750005}{\agslColor} & \nosci{2.218071e+01} & \qty{67.066984}{\second} \\
\;$(-\infty, \infty)$ & \cdfName        & \nosci{-4.644122e+00}                             & \interval{0.000}{0.710058}{\cdfColor}  & \nosci{1.732868e+01} & $<$\qty{50}{\micro\second}   \\
                      & \sfName         & \nosci{-2.852363e+00}                             & \interval{0.079}{1.000000}{\sfColor}   & \nosci{1.039721e+02} & $<$\qty{50}{\micro\second}   \\
                      & \ecdfName       & \nosci{-4.644122e+00}                             & \interval{0.000}{1.000000}{\ecdfColor} & \nosci{1.039721e+02} & $<$\qty{50}{\micro\second}   \\ \hline
Gumbel2(1, 1)         & \agslName       & \sci{4.508400e-02}                                & \interval{0.014}{0.246969}{\agslColor} & \sci{4.294967e+09}   & 1\qty{08.680311}{\second} \\
\;$(0,\infty)$        & \cdfName        & \sci{9.617967e-03}                                & \interval{0.000}{0.202298}{\cdfColor}  & \sci{3.355443e+07}   & $<$\qty{50}{\micro\second}   \\
                      & \sfName         & \sci{5.770780e-02}                                & \interval{0.016}{1.000000}{\sfColor}   & \sci{1.427248e+45}   & $<$\qty{50}{\micro\second}   \\
                      & \ecdfName       & \sci{9.617967e-03}                                & \interval{0.000}{1.000000}{\ecdfColor} & \sci{1.427248e+45}   & $<$\qty{50}{\micro\second}   \\ \hline
Laplace(1)            & \agslName       & \nosci{-2.148756e+01}                             & \interval{0.169}{0.830730}{\agslColor} & \nosci{2.148756e+01} & \qty{47.817449}{\second} \\
\;$(-\infty, \infty)$ & \cdfName        & \nosci{-1.032789e+02}                             & \interval{0.000}{0.803135}{\cdfColor}  & \nosci{1.663553e+01} & $<$\qty{50}{\micro\second}   \\
                      & \sfName         & \nosci{-1.663553e+01}                             & \interval{0.197}{1.000000}{\sfColor}   & \nosci{1.032789e+02} & $<$\qty{50}{\micro\second}   \\
                      & \ecdfName       & \nosci{-1.032789e+02}                             & \interval{0.000}{1.000000}{\ecdfColor} & \nosci{1.032789e+02} & $<$\qty{50}{\micro\second}   \\\hline
Logistic(1)           & \agslName       & \nosci{-2.218071e+01}                             & \interval{0.166}{0.833672}{\agslColor} & \nosci{2.218071e+01} & \qty{39.313025}{\second} \\
\;$(-\infty, \infty)$ & \cdfName        & \nosci{-1.039721e+02}                             & \interval{0.000}{0.807094}{\cdfColor}  & \nosci{1.732868e+01} & $<$\qty{50}{\micro\second}   \\
                      & \sfName         & \nosci{-1.732868e+01}                             & \interval{0.193}{1.000000}{\sfColor}   & \nosci{1.039721e+02} & $<$\qty{50}{\micro\second}   \\
                      & \ecdfName       & \nosci{-1.039721e+02}                             & \interval{0.000}{1.000000}{\ecdfColor} & \nosci{1.039721e+02} & $<$\qty{50}{\micro\second}   \\\hline
Pareto(3, 2)          & \agslName       & \sci{2.000000e+00}                                & \interval{0.000}{0.213333}{\agslColor} & \sci{3.250997e+03}   & \qty{61.274993}{\second} \\
\;$(2,\infty)$        & \cdfName        & \sci{2.000000e+00}                                & \interval{0.000}{0.166667}{\cdfColor}  & \sci{6.450796e+02}   & $<$\qty{50}{\micro\second}  \\
                      & \sfName         & \sci{2.000000e+00}                                & \interval{0.000}{1.000000}{\sfColor}   & \sci{2.251800e+15}   & $<$\qty{50}{\micro\second}  \\
                      & \ecdfName       & \sci{2.000000e+00}                                & \interval{0.000}{1.000000}{\ecdfColor} & \sci{2.251800e+15}   & $<$\qty{50}{\micro\second}  \\\hline
Rayleigh(1)           & \agslName       & \sci{2.200000e-05}                                & \interval{0.753}{0.985777}{\agslColor} & \nosci{6.660437e+00} & \qty{35.377007}{\second} \\
\;$(0,\infty)$        & \cdfName        & \sci{3.743392e-23}                                & \interval{0.000}{0.983504}{\cdfColor}  & \nosci{5.887050e+00} & $<$\qty{50}{\micro\second}   \\
                      & \sfName         & \sci{2.441406e-04}                                & \interval{0.798}{1.000000}{\sfColor}   & \nosci{1.442027e+01} & $<$\qty{50}{\micro\second}   \\
                      & \ecdfName       & \sci{3.743392e-23}                                & \interval{0.000}{1.000000}{\ecdfColor} & \nosci{1.442027e+01} & $<$\qty{50}{\micro\second}   \\\hline
Weibull(1, 1)         & \agslName       & \sci{0.000000e+00}                                & \interval{0.000}{0.667343}{\agslColor} & \nosci{2.218071e+01} & \qty{92.220060}{\second} \\
\;$(0,\infty)$        & \cdfName        & \sci{7.006492e-46}                                & \interval{0.000}{0.614188}{\cdfColor}  & \nosci{1.732868e+01} & $<$\qty{50}{\micro\second}   \\
                      & \sfName         & \sci{2.980232e-08}                                & \interval{0.812}{1.000000}{\sfColor}   & \nosci{1.039721e+02} & $<$\qty{50}{\micro\second}   \\
                      & \ecdfName       & \sci{7.006492e-46}                                & \interval{0.000}{1.000000}{\ecdfColor} & \nosci{1.039721e+02} & $<$\qty{50}{\micro\second}   \\\hline
Gamma(.5, 1)          & \agslName       & ---                                               & {\color{\agslColor}unknown}            & ---                  & $\infty$ \\
\;$(0,\infty)$        & \cdfName        & \sci{3.855593e-91}                                & \interval{0.000}{0.991146}{\cdfColor}  & \nosci{1.536017e+01} & $<$\qty{50}{\micro\second}   \\
                      & \sfName         & \sci{6.975737e-16}                                & \interval{0.814}{1.000000}{\sfColor}   & \nosci{1.010868e+02} & $<$\qty{50}{\micro\second}   \\
                      & \ecdfName       & \sci{3.855593e-91}                                & \interval{0.000}{1.000000}{\ecdfColor} & \nosci{1.010868e+02} & $<$\qty{50}{\micro\second}   \\\hline
Gaussian(0, 1)        & \agslName       & ---                                               & {\color{\agslColor}unknown}            & ---                  & $\infty$ \\
\;$(-\infty, \infty)$ & \cdfName        & \nosci{-1.417019e+01}                             & \interval{0.000}{0.818748}{\cdfColor}  & \nosci{5.419983e+00} & $<$\qty{50}{\micro\second}   \\
                      & \sfName         & \nosci{-5.419983e+00}                             & \interval{0.181}{1.000000}{\sfColor}   & \nosci{1.417019e+01} & $<$\qty{50}{\micro\second}   \\
                      & \ecdfName       & \nosci{-1.417019e+01}                             & \interval{0.000}{1.000000}{\ecdfColor} & \nosci{1.417019e+01} & $<$\qty{50}{\micro\second}   \\\hline
Tdist(1)              & \agslName       & ---                                               & {\color{\agslColor}unknown}            & ---                  & $\infty$ \\
\;$(-\infty, \infty)$ & \cdfName        & \sci{-4.543071e+44}                               & \interval{0.000}{0.578695}{\cdfColor}  & \sci{1.068071e+07}   & $<$\qty{50}{\micro\second} \\
                      & \sfName         & \sci{-1.068071e+07}                               & \interval{0.421}{1.000000}{\sfColor}   & \sci{4.543071e+44}   & $<$\qty{50}{\micro\second} \\
                      & \ecdfName       & \sci{-4.543071e+44}                               & \interval{0.000}{1.000000}{\ecdfColor} & \sci{4.543071e+44}   & $<$\qty{50}{\micro\second} \\\hline\hline
\end{tabular*}
\vspace{-.5cm}
\end{table}

\subsection{\texorpdfstring%
  {\labelcref{question:range} Output Range of Random Variate Generators}%
  {(Q2) Output Range of Random Variate Generators}%
  }
\label{sec:evaluation-range}

\Cref{table:support} shows a comparison of the min--max output
range for GSL generators and those
specified formally by a finite-precision \WCDF{}/SF (\cref{sec:floating}) and
\ECDF{} (\cref{sec:survival}), for 13 distributions.
The theoretical ranges of these distributions are shown in the first column.
We identify several takeaways:

\begin{itemize}[wide=0pt]
\item
The output range of a GSL generator is often
close to that of a finite-precision \WCDF{} or SF, but always inferior to
the range of the extended-accuracy \ECDF{} (which is up to $10^{35}$x wider).

\item For symmetric distributions (Cauchy, Laplace, Logistic,
Gaussian, Tdist), the  finite-precision \ECDF{} fixes the asymmetry in
\WCDF{} and SF, by ensuring identical ranges below and above the median.

\item For distributions over nonnegative numbers (Exponential, Gumbel2, Pareto, Rayleigh, Weibull, Gamma),
the \ECDF{} extends the output range by many orders of magnitude,
combining the \WCDF{} to represent values near 0 and \WSF{} to represent values away from zero.

\item The output ranges of our generators can be quickly obtained
(using, e.g., \cref{alg:quantile}) in microseconds.
In contrast, finding the range of a GSL generator requires dozens of
seconds in certain cases that can be enumerated (i.e., the algorithm draws
a single 32-bit floating-point uniform) and cannot be done in cases
that draw two or more floating-point uniforms (e.g., Gamma, Gaussian, and Tdist).
It is impractical to enumerate the CDF of a GSL generator in
all cases.
\end{itemize}

\subsection{\texorpdfstring%
  {\labelcref{question:ratios} Runtime Overhead of Extended-Accuracy Generators}%
  {(Q3) Runtime Overhead of Extended-Accuracy Generators}%
  }
\label{sec:evaluation-ratios}

\Cref{fig:ratios} shows the overhead of using the extended-accuracy
algorithms described in \cref{sec:survival} in terms of bits/variate
and variates/sec, for both conditional bit sampling (\cref{alg:sampler-naive-impl,alg:sampler-naive-ext-impl})
and optimal generation (\cref{alg:sampler-opt-impl,alg:sampler-opt-ext-impl}).
The bits/variate ratios are slightly above one in most cases because a
\ECDF{} can represent twice as many outcomes compared to a \WCDF{}
(\cref{remark:ecdf-double-values}), and in turn higher-entropy
distributions.
The variates/sec ratios for conditional bit sampling are 0.51x--0.85x
(average=0.63x).
This high overhead arises from the larger number of machine words needed to
compute ratios of probabilities using arbitrary-precision arithmetic.
The variates/sec ratios using the optimal generators are
0.40x--1.28x (average=1.00x).
The only substantial slowdown (0.40x) is on the degenerate Pascal(1,5)
distribution.
The \ECDF{} algorithms have a (statistically significant) higher output
generation rate than the \WCDF{} algorithms in 10/24 cases.
In summary, the extended-accuracy generators incur minimal overhead
compared to their lower-accuracy counterparts.

%!TEX root=./paper.tex

\begin{figure}[t]
\includegraphics[width=\linewidth]{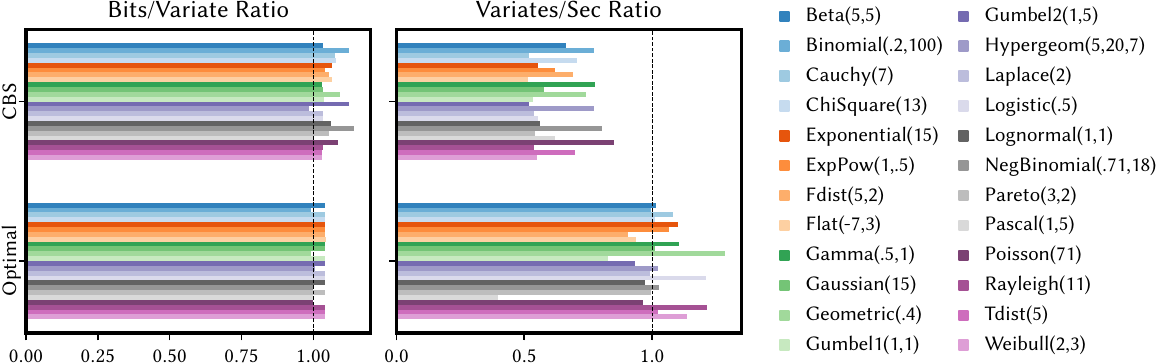}
\captionsetup{belowskip=-10pt}
\caption{%
Ratio of bits/variate (lower=better) and variates/sec (higher=better)
using the \ECDF{} specification (\cref{sec:survival}) versus the \WCDF{}
specification (\cref{sec:floating}) of the 24 target probability
distributions from \cref{table:efficiency}.
CBS compares
  \cref{alg:sampler-naive-ext-impl}/\cref{alg:sampler-naive-impl}.
Optimal compares
  \cref{alg:sampler-opt-ext-impl}/\cref{alg:sampler-opt-impl}.
}
\label{fig:ratios}
\end{figure}

\section{Related Work}
\label{sec:related}
%!TEX root=./paper.tex

Random variate generation has been traditionally grounded in the idealized
Real-RAM model of computation~\citep{devroye1986,occil2023a,occil2023b,occil2024}:
\cref{sec:overview-software} demonstrates several challenges with this
approach.
Our approach deviates from this tradition in two ways---\begin{enumerate*}[label=(\roman*)]
\item it is based on a realistic model of the finite-precision computer on
which the algorithms execute; and
\item random variate generators are automatically synthesized from numerical
programs specifying the CDF or SF
\end{enumerate*}---enabling \crefrange{contribution:binary}{contribution:evaluation}.

\paragraph{DDG Trees}

Our approach builds on the discrete distribution generating (DDG) tree
formalism of \citet{knuth1976}.
\Cref{contribution:binary} improves on their original method for lazy DDG
exploration~\citep[pp 384--385]{knuth1976} by giving a deterministic algorithm
that is space-time optimal (\cref{remark:refine-improve-knuth-yao}); whereas
\crefrange{contribution:floating}{contribution:evaluation} go beyond the
work of \citet{knuth1976}.
Many existing DDG algorithms for discrete distributions require enumerating
the target
probabilities~\citep{draper2025,saad2020popl,saad2020aistats,roy2013,Karmakar2018},
which is intractable for the class of finite-precision probability distributions
that we consider.

\paragraph{Finite Precision}

Several works have introduced finite-precision generators for
specific distributions (e.g.,
Laplace~\citep{Mironov2012,gazeau2016},
exponential~\citep{thomas2008},
uniform~\citep{Goualard2020,Goualard2022,Lumbroso2013,campbell2014},
Gaussian~\citep{walter2019,giles2023}).
This work introduces general methods that are not specific to any
particular distribution.
\Citet{derflinger2010} describe an approximate generation method given a
numerical implementation of a probability density function, although
its theoretical guarantees only hold under the Real-RAM model~\citep[Remark 9]{derflinger2010}.
\Citet{uyematsu2003} give an implementation of the (entropy-suboptimal)
\Citet{han1997} algorithm using integer arithmetic, which requires
very high precision.
Fore example, given $n$-bit floating-point probabilities in $\floatEm$,
the \citet{uyematsu2003} method requires
$2^{E}+2m-1 \gg n$ bits of precision for the integer arithmetic
to be exact (e.g., 2151 bits for 64-bit floats).
In contrast, our work builds on the (entropy-optimal) \citet{knuth1976}
method and requires integer arithmetic with exactly $n \defas 1 + E +m$
bits of precision, matching the precision level used
to specify the numerical CDF implementation (e.g., 64 bits for 64-bit floats).

\paragraph{Arbitrary Precision}

\Citet{devroye2015} present universal
generation algorithms that require arbitrary-precision arithmetic (e.g.,
MPFR~\citep{Fousse2007} or GMP~\citep{Granlund2023}).
Specialized arbitrary-precision generators for the discrete Gaussian
distribution have been widely
studied~\citep{canone2020,karney2016,Du2022,ducas2012}, given its
prominence in cryptography and differential privacy.
In contrast to this line of work, our method uses finite- instead of
arbitrary-precision, retaining high performance and predictability of
runtime and memory.
A promising direction is to develop finite-precision CDF, SF, or DDF
specifications that meet the accuracy requirements for these
applications, which could be implemented using numerical libraries for
approximating real functions with error
guarantees~\citep{ziv2001,briggs2024,lim2021,lim2022,daramyloirat2006,sibidanov2022}.

\paragraph{Probabilistic Programming}
Several probabilistic programming languages and solvers
use the CDF to form discrete approximations of
continuous probability distributions~\citep{belle2015,pedro2019,garg2024,saad2021sppl}.
The denotational semantics of these systems adopt the infinite-precision
Real-RAM model, which does not comport with their actual
operational semantics on a finite-precision computer.
The resulting systems offer no correctness or exactness guarantees for
the machine implementation.
\Citet{bagnall2023} develop formally verified generators for discrete
probabilistic programs with loops and conditioning using
arbitrary-precision rational arithmetic, and use it to implement exact
generators for the discrete Laplace and Gaussian distributions.
A promising idea along this direction is to build a probabilistic
programming language that instead uses the exact finite-precision random
variate generators described in \crefrange{sec:binary}{sec:survival},
as the basic probabilistic primitives with formal guarantees.

\section{Conclusion}
\label{sec:conclusion}
%!TEX root=./paper.tex

As the role of probability in computer science continues to
grow~\citep{mitzenmacher2017,barthe2020}, there is a growing need for
random variate generation methods with well-characterized behavior.
We hope this work lays a foundation for a new class of random variate
generators that are equipped with theoretical guarantees while delivering
improvements in automation, accuracy, and entropy consumption.

\section*{Data-Availability Statement}
A C library with all the algorithms described in this article
is available at \url{https://github.com/probsys/librvg}.
A reproduction package for the evaluation in \cref{sec:evaluation}
is available on Zenodo~\citep{artifact25}.

\begin{acks}
The authors thank the referees and Martin Rinard for helpful feedback.
Feras Saad and Wonyeol Lee are corresponding authors.
This material is based upon work supported by the
\grantsponsor{NSF}{National Science Foundation}{https://doi.org/10.13039/100000001}
under Grant No.~\grantnum{NSF}{2311983}.
Any opinions, findings, and conclusions or recommendations
in this material are those of the authors and do not necessarily
reflect the views of the NSF.
\end{acks}

\bibliography{paper}

\AtEndDocument{%
%!TEX root=./paper.tex
\clearpage
\appendix

% Make algorithms numbered in appendix.
% https://tex.stackexchange.com/questions/118632/change-caption-number-of-an-algorithm
\renewcommand{\thealgorithm}{\thesection\arabic{algorithm}}

\clearpage

\section{Conditional Bit Sampling: A Baseline for Exact Random Variate Generation}
\label{appx:naive-baseline}
%!TEX root=./paper.tex

This appendix describes conditional bit sampling~\citep[\S{II.B}]{Sobolewski1972}
a baseline method for generating a random string
from any binary-coded probability distribution (\cref{definition:bcpd}).
Whereas the original presentation of this method in \citet{Sobolewski1972}
used the Real-RAM model, the presentation of conditional bit sampling in
this appendix uses the DDG formalism and finite-precision computation,
which gives new insights on its behavior.
\Cref{appx:naive-baseline-binary} discusses the general case.
\Cref{appx:naive-baseline-impl} shows how to implement this baseline
given a finite-precision \WCDF.

\subsection{Generation from a Binary-Coded Probability Distribution}
\label{appx:naive-baseline-binary}

Conditional bit sampling a joint distribution of $(B_1, \ldots, B_n)$
uses the chain rule of probability, i.e., it generates
$B_1$, then $B_2 \mid B_1$, then $B_3 \mid B_1, B_2$, and so on.

\begin{proposition}[name=,restate=SamplerNaive]
\label{theorem:sampler-naive}
Let $p: \bool^* \to [0,1]$ be a binary-coded probability
distribution. For each $n \in \mathbb{N}$. The following process
generates a random string $B_1\ldots B_n \sim p_n$:
\begin{align}
B_1 &\sim \mathrm{Bernoulli}(p(1));
\qquad B_j \sim \mathrm{Bernoulli}\left[\frac{p(B_1\ldots{B_{j-1}}1)}{p(B_1\ldots{B_{j-1}})}\right]  && (j = 2, \dots, n).
\rlap{\;\;\;\qedhere}
\label{eq:preoccupative}
\end{align}
\end{proposition}

\begin{proof}
By induction.
The base case is obvious.
Assume the proposition holds for any integer $n-1$.
Put $(b_1,\dots,b_{n-1})$ so that $p(b_1,\dots,b_{n-1}) > 0$.
Then
\begin{align}
&\Pr\left(\cap_{i=1}^{n-1}\set{B_i=b_i}, B_n=b_n\right)
\\
&= \Pr\left(B_n=b_n \mid \cap_{i=1}^{n-1}\set{B_i=b_i}\right) \Pr\left(\cap_{i=1}^n\set{B_i=b_i}\right)
\\
&=
  \left[\frac{p(b_1\ldots{b_{n-1}}1)}{p(b_1\ldots{b_{n-1}})}\right]^{b_n}
  \left[1-\frac{p(b_1\ldots{b_{n-1}}1)}{p(b_1\ldots{b_{n-1}})}\right]^{1-b_n}
  \cdot
  p(b_1\ldots{b_{n-1}})
\\
&=
  \left[\frac{p(b_1\ldots{b_{n-1}}1)}{p(b_1\ldots{b_{n-1}})} \right]^{b_n}
  \left[\frac{p(b_1\ldots{b_{n-1}}) - p(b_1\ldots{b_{n-1}}1)}{p(b_1\ldots{b_{n-1}})} \right]^{1-b_n}
  \cdot
  p(b_1\ldots{b_{n-1}})
\\
&=
  \left[p(b_1\ldots{b_{n-1}}1)\right]^{b_n}
  \left[p(b_1\ldots{b_{n-1}}0)\right]^{1-b_n}
\\
&= p_n(b_1\ldots{b_n}).
\end{align}
\end{proof}

\paragraph{Optimal Bernoulli Generation}
In the Real-RAM model, a random variable $X \sim \mathrm{Bernoulli}(p)$
can be defined using the inverse-transform method:
$X(\omega) \defas \mathbf{1}[\omega \in (0, p]].$
To arrive at a random variate generator for $\mathrm{Bernoulli}(p)$
based on \cref{definition:random-variate-generator},
consider generating $K \sim \mathrm{Geometric}(1/2)$ and then setting
$X \gets p_K$, where $p=(0.p_1p_2\dots)_2 \in (0,1)$.
The proof of correctness is direct:
\begin{align}
\Pr(X = 1)
  &= \textstyle\sum_{k=1}^{\infty}\Pr(X = 1 {\mid} K=k) \Pr(K=k)
  = \sum_{k=1}^{\infty}\mathbf{1}[p_k=1] 1/2^k
  = \sum_{k=1}^{\infty}p_k/2^k
  = p, \label{eq:superfit-1} \\
\Pr(X = 0)
  &= \textstyle\sum_{k=1}^{\infty}\Pr(X = 0 {\mid} K=k) \Pr(K=k)
  = \sum_{k=1}^{\infty}(1-\mathbf{1}[p_k=1]) 1/2^k
  = 1 - p. \label{eq:superfit-2}
\end{align}
The expected entropy cost of this method is two bits for generating $K$.
\Citet[Appendix B]{Lumbroso2013} proves this method is
entropy-optimal, but their proof assumes implicitly that $p$ is not a
dyadic rational.
If $p=k/2^m$ is a dyadic rational for odd $k$, this method is suboptimal
because the bits $(p_{m+1}, p_{m+2}, \dots)$ are zero, and so $K$ need
not be generated beyond $m$.
In particular, it suffices to generate
$K \gets \min(m, K')$ where $K' \sim \mathrm{Geometric}(1/2)$,
and then set $X \gets p_{K}$ if $K < m$ and $X \gets \Flip$ otherwise.
The expected entropy cost is then
\begin{align}
\underbrace{\textstyle\sum_{i=1}^{m-1}i2^{-i} + (m-1)2^{-(m-1)}}_{K'}
  + \underbrace{\vphantom{\textstyle\sum_{i=1}^{m-1}} 2^{1-m}}_{\Flip}
= 2 - 2^{2-m} + 2^{1-m}
= 2 - 2^{1-m}.
\label{eq:underfed}
\end{align}

\paragraph{Composing Optimal Bernoulli Generators}
\Cref{alg:sampler-naive} presents a random variate generator for a
binary-coded probability distribution $p$ based on
\cref{theorem:sampler-naive,eq:superfit-1,eq:superfit-2}.
The $\Flip$ primitive returns the next unbiased random bit from
the i.i.d.~bit stream (\cref{fig:real-world}, bottom).
Each recursive step of \cref{alg:sampler-naive} is entropy-optimal
for the $\mathrm{Bernoulli}(p(b0)/p(b))$ distribution.
However, the following example shows that a \textit{sequence} of $n$
generations~\cref{eq:preoccupative} is not entropy-optimal for $p_n$.

%!TEX root=./paper.tex

\begin{figure}[t]
\centering
\tikzset{level distance=10pt, sibling distance=2pt}
\tikzset{every tree node/.style={anchor=north}}
\tikzset{every leaf node/.style={draw,inner sep=1.5pt}}
\tikzset{branch/.style={shape=coordinate}}
\begin{subfigure}[b]{.4\linewidth}
\centering
\begin{adjustbox}{max width=\linewidth}
\begin{tikzpicture}
\centering
\Tree
[.\node[branch,name=L1]{};
  \edge[dash dot];
  [
    \edge[dash dot]; \node[name=R1]{\phantom{01}};
    \edge[dash dot];
    [.\node[branch,name=L21]{};
      [ [ [ \node[name=R21]{\phantom{01}}; 01 ] 00 ] 00 ] 01 ]
  ]
  \edge[dash dot];
  [.\node[branch,name=L22]{};
    [ [.\node[shape=coordinate,name=L23]{}; [ [ [ \node[name=R22]{\phantom{01}}; 01 ] 01 ] 10 ] 10 ] 11 ] 10 ]
]
\draw[-latex,red,dashed,out=150,in=170] (R1.center) to (L1);
\draw[-latex,red,bend left=45] (R21.center) to (L21.south);
\draw[-latex,red] (R22.center) to[bend left=45] (L23.west);

\end{tikzpicture}
\end{adjustbox}
\caption{Conditional Bit Sampling (\cref{alg:sampler-naive})}
\label{fig:DDG-sequential-composition-suboptimal}
\end{subfigure}
\qquad\qquad
\begin{subfigure}[b]{.4\linewidth}
\centering
\begin{adjustbox}{max width=.9\linewidth}
\begin{tikzpicture}
\Tree
  [.\node[name=root,branch]{};
    [ [ [ \node[name=back]{\phantom{01}}; 01 ] [ 10 11 ] ] [ 00 01 ] ]
  [ [ 11 10 ] 10 ]
]
\draw[red,-latex,out=100,in=150] (back.center) to (root.south);
\end{tikzpicture}
\end{adjustbox}
\caption{Optimal Generation (\cref{alg:sampler-opt})}
\label{fig:DDG-sequential-composition-optimal}
\end{subfigure}
\caption{DDG trees of random variate generators for the discrete distribution $p_2$ in \cref{eq:deperditely}.
In \subref{fig:DDG-sequential-composition-suboptimal}, dashed lines show the DDG tree for $B_1$; solid
lines show DDG trees for $B_2{\mid}B_1{=}0$ (left subtree) and $B_2{\mid}B_1{=}1$ (right subtree).
}
\label{fig:DDG-sequential-composition}
\end{figure}
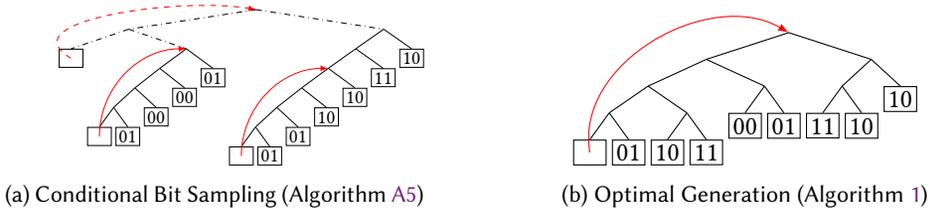

\begin{example}
\label{example:bernoulli-compose-suboptimal}
Consider a binary-coded probability distribution $p$ such that
\begin{align}
p(0)   &= 1/3,
&&p(1)  = 2/3,
&&p(00) = 2/15,
&&p(01) = 3/15,
&&p(10) = 7/15,
&&p(11) = 3/15.
\label{eq:deperditely}
\end{align}
Following \cref{eq:preoccupative},
generating $B_1B_2 \sim p_2$ via $B_1 \sim \mathrm{Bernoulli}(2/3)$
and $B_2 \mid B_1 \sim \mathrm{Bernoulli}(3/(5+5B_1))$
consumes 4 bits on average (\cref{fig:DDG-sequential-composition-suboptimal}).
However, an entropy-optimal generator for $p_2$ constructed from the
binary expansions of $(p(00),p(01),p(10),p(11))$ % the probabilities
consumes 3.2 bits on average (\cref{fig:DDG-sequential-composition-optimal}):
\begin{minipage}{.42\linewidth}
\begin{align*}
  \begin{Bmatrix} p(00) \\ p(01) \\ p(10) \\ p(11) \end{Bmatrix}
  = \begin{Bmatrix} 2/15 \\ 3/15 \\ 7/15 \\ 3/15 \end{Bmatrix}
  = 0.\left\lbrace\overline{\begin{matrix}
    0 & 0 & 1 & 0 \\
    0 & 0 & 1 & 1 \\
    0 & 1 & 1 & 1 \\
    0 & 0 & 1 & 1
  \end{matrix}}\right.
\end{align*}
\end{minipage}\hfill
\begin{minipage}{.58\linewidth}
\begin{align}
\begin{aligned}[b]
& N = \begin{aligned}[t]
  (0\cdot1\cdot{1/2^1}) & + (1\cdot2\cdot{1/2^2}) + (4\cdot3\cdot{1/2^3}) \\
                      & + (3\cdot4\cdot{1/2^4}) + (1 \cdot (4+N)\cdot{1/2^4})
  \end{aligned}\hspace{-1cm}\\
& \!\!\! \implies N = 3.2\mbox{ bits} \qquad\rlap{\qedsymbol}
\label{eq:rebringer}
\end{aligned}
\end{align}
\end{minipage}\noqed
\end{example}

%!TEX root=./paper.tex

\begin{figure}[!t]

\newcommand{\hhl}[2]{{\sethlcolor{#1}\hl{#2}}}
\newcommand{\hlp}[1]{\hhl{pink}{#1}}
\newcommand{\hlg}[1]{\hhl{Goldenrod}{#1}}
\newcommand{\hlb}[1]{\hhl{cyan!50!white}{#1}}
\newcommand{\gr}[1]{{\textcolor{gray}{#1}}}
\newcommand{\tm}[1]{\tikz[overlay, remember picture, baseline=(#1.south)] \node[name=#1,rectangle,draw] {};}

\setlength{\tabcolsep}{4pt}

\begin{minipage}[t]{\linewidth}
\centering
\captionsetup{belowskip=0pt,aboveskip=0pt}
\subcaption{Example binary-coded probability distribution $p$ unrolled to the first four bits.}
\label{fig:optimal-sampler-naive-dist}
\begin{adjustbox}{max width=\linewidth}
\begin{tikzpicture}
\node[]{$
\begin{NiceMatrix}[l]
 0000 \mapsto \frac{6}{137}
&0001 \mapsto \frac{12}{137}
&0010 \mapsto \frac{13}{137}
&0011 \mapsto \frac{9}{137}
% \\
&0100 \mapsto \frac{10}{137}
&0101 \mapsto \frac{12}{137}
&0110 \mapsto \frac{6}{137}
&0111 \mapsto \frac{1}{137}
\\[2pt]
1000 \mapsto \frac{1}{137}
&1001 \mapsto \frac{2}{137}
&1010 \mapsto \frac{13}{137}
&1011 \mapsto \frac{8}{137}
% \\
&1100 \mapsto \frac{14}{137}
&1101 \mapsto \frac{13}{137}
&1110 \mapsto \frac{7}{137}
&1111 \mapsto \frac{10}{137}
\CodeAfter
  \SubMatrix\lbrace{1-1}{2-8}\rbrace
\end{NiceMatrix}$};
\end{tikzpicture}
\end{adjustbox}
\end{minipage}

\setlength{\intextsep}{2pt}

\begin{minipage}[t]{.45\linewidth}
\begin{algorithm}[H]
\captionsetup{hypcap=false}
\caption{Conditional Bit Sampling}
\label{alg:sampler-naive}
\algrenewcommand\algorithmicindent{1.0em}%
\begin{algorithmic}[1]
\Function{\SampleNaive}{$p$, $b=\varepsilon$}
  \If{$p(b0) = p(b)$} \Comment{Leaf}
    \State \Return $\SampleNaive(p, b0)$
    \Comment{0}%
    \label{algline:BadSampler-Print0-NoIter}
  \EndIf
  \If{$p(b1) = p(b)$} \Comment{Leaf}
    \State \Return \SampleNaive($p$, $b1$)
    \Comment{1}%
    \label{algline:BadSampler-Print1-NoIter}
  \EndIf
  \For{$i=1,2,\dots$} \label{algline:BadSampler-Loop} \Comment{Refine Subtree}
    \State $x \gets \Flip()$
    \If{$x = 0 \wedge [p({b0})/p(b)]_{i} = 1$} \Comment{Leaf}
      \State \Return $\SampleNaive(p, b0)$
      \Comment{0}%
      \label{algline:BadSampler-Print0}
    \EndIf
    \If{$x = 1 \wedge [p({b1})/p(b)]_{i} = 1$} \Comment{Leaf}
      \State \Return $\SampleNaive(p, b1)$
      \Comment{1}%
      \label{algline:BadSampler-Print1}
    \EndIf
  \EndFor
\EndFunction
\end{algorithmic}
\end{algorithm}
\end{minipage}\hfill
\begin{subtable}[t]{.5\linewidth}
\centering
\captionsetup{skip=0pt}
\subcaption{Example trace of \cref{alg:sampler-naive} on $p$ from \cref{fig:optimal-sampler-naive-dist}.}
\label{fig:optimal-sampler-naive-trace}
\renewcommand{\tm}[1]{}
\begin{adjustbox}{max width=\linewidth}
\begin{tabular}{|l||p{1cm}@{=\,}p{.5cm}@{\,=\,\gr{0.}\,}lllllll|c@{}}
\cline{4-10}
\multicolumn{1}{c}{Recur.} & \multicolumn{2}{c|}{~} & \multicolumn{7}{c|}{\Flip\,$x$} & ~
\\
\multicolumn{1}{c}{Level}  & \multicolumn{2}{c|}{Probabilities}
                                        & 1      & 0         & \hlb{1} & \hlb{1}   & \hlb{0}   & 1         & \hlb{1} & Output $b$ \\ \hline\hline
\multirow{2}{*}{0} & $p(0)$                   & $\frac{69}{137}$ & \bf1      & \bf0 & \bf0       & \gr{0}     & \gr{0}     & \gr{0} & \gr{0}     & ~ \Tstrut \\
~                  & $p(1)$                   & $\frac{68}{137}$ & \bf0      & \bf1 & \hlg{\bf1} & \gr{1}     & \gr{1}     & \gr{1} & \gr{1}     & \hlp{1} \Tstrut\Bstrut \\ \cline{1-10}
\multirow{2}{*}{1} & $\frac{p(10)}{p(1)}$     & $\frac{6}{17}$   & \tm{3-s0} & ~    &            & \bf{0}     & \gr{1}     & \gr{0} & \gr{1}     & ~ \Tstrut \\
~                  & $\frac{p(11)}{p(1)}$     & $\frac{11}{17}$  & \tm{4-s0} & ~    &            & \hlg{\bf1} & \gr{0}     & \gr{1} & \gr{0}     & \hlp{1} \Tstrut\Bstrut \\ \cline{1-10}
\multirow{2}{*}{2} & $\frac{p(110)}{p(11)}$   & $\frac{27}{44}$  & \tm{5-s0} & ~    & ~          &            & \hlg{\bf1} & \gr{0} & \gr{0}     & \hlp{0} \Tstrut \\
~                  & $\frac{p(111)}{p(11)}$   & $\frac{17}{44}$  & \tm{6-s0} & ~    & ~          &            & \bf0       & \gr{1} & \gr{1}     & ~ \Tstrut\Bstrut \\ \cline{1-10}
\multirow{2}{*}{3} & $\frac{p(1100)}{p(110)}$ & $\frac{14}{27}$  & \tm{7-s0} & ~    & ~          & ~          &            & \bf1   & \bf0       & ~ \Tstrut \\
~                  & $\frac{p(1101)}{p(110)}$ & $\frac{13}{27}$  & \tm{8-s0} & ~    & ~          & ~          &            & \bf0   & \hlg{\bf1} & \hlp{1} \Tstrut\Bstrut \\\hline\hline
\end{tabular}
\end{adjustbox}
\end{subtable}
\addtocounter{figure}{-1}
\captionsetup{skip=4pt}
\caption{Conditional bit sampling algorithm for any binary-coded probability
distribution $p: \set{0,1}^* \to [0,1]$ using the chain rule. Refer to
the caption of \cref{fig:optimal-sampler} for details.}
\label{fig:optimal-sampler-naive}
\end{figure}

The next proposition establishes bounds on the entropy gap between the
conditional bit sampling and optimal generators in \cref{alg:sampler-naive,alg:sampler-opt}.
It shows that this gap could be zero, or very large.

\begin{proposition}[name=,restate=OptNaiveBound]
\label{proposition:opt-naive-bound}
For a binary-coded probability distribution $p$, let
$C^{\rm opt}_n(p)$ and
$C^{\rm cbs}_n(p)$
 denote the entropy costs of
\cref{alg:sampler-opt,alg:sampler-naive}, respectively, up until
generating an $n$-bit string.
There exist binary-coded distributions $p$ and $p'$
and an arbitrarily small constant $\varepsilon > 0$ such that
\begin{align}
\mathbb{E}[C^{\rm cbs}_n(p)] - \mathbb{E}[C^{\rm opt}_n(p)] = 0;
&&
\mathbb{E}[C^{\rm cbs}_n(p')] - \mathbb{E}[C^{\rm opt}_n(p')] = 2n - 2 - \varepsilon.
\qquad\rlap{\qedsymbol}
\end{align}\noqed
\end{proposition}

\begin{proof}
For the first equality, suppose $p$ is such that each
distribution $p_j$ is uniform distribution over $\bool^j$
($j=1,\dots,n$), so their entropies satisfy $H(p_j) = j$.
Since conditional bit sampling using $\hyperref[alg:sampler-bernoulli-impl]{\Bernoulli}$ makes a coin flip
with dyadic weight $1/2$ at each step, it consumes 1 bit per step on
average for a total of $n$ bits, matching the entropy-optimal generator.

For the second equality, consider a
distribution $p'$ such that $p'_n$ has $2^n-1$ outcomes of
probability $\varepsilon \ll 2^{-n}$ and one outcome has
probability $\gamma \defas 1- (2^n-1)\varepsilon$.
The cost $C^{\rm opt}_n$ of the optimal generator satisfies
\begin{align}
H(p'_n) \le \mathbb{E}[C^{\rm opt}_n(p'_n)] \le H(p'_n) + 2
\implies
\mathbb{E}[C^{\rm opt}_n(p'_n)] \le \varepsilon\log_2(1/\varepsilon) + \gamma\log_2(\gamma^{-1}) + 2 \approx 2.
\end{align}
For conditional bit sampling, the average cost
\begin{align}
\mathbb{E}[C^{\rm cbs}_n] = T_n(n) = 2n - 2\varepsilon\left(n2^{n-1} + 1 - 2^{n}\right) \approx 2n
\end{align}
is the value of the following recurrence at $T_n(n)$:
\begin{align}
\begin{aligned}
T_n(1) = 2,
\qquad
T_n(k) &= 2 + (k-1)c_n(k) + T_n(k-1)(1-c_n(k)) && (2 \le k \le n), \\
\mbox{where } c_n(k) &\defas \frac{2^{k-1}\varepsilon}{(2^k-1)\varepsilon + \gamma}.
\label{eq:Saxony}
\end{aligned}
\end{align}
To justify \cref{eq:Saxony},
in the base case $(k=1)$ the distribution $p_2$ over $\bool^1$ is non-uniform
which requires 2 bits.
For the inductive case, there are $2^k$ outcomes in total,
where $(2^k-1)$ have probability $\varepsilon$
and one outcome has probability $\gamma$.
Conditional bit sampling uses two bits to flip a coin with weight $c_n(k)$,
which chooses between the $2^{k-1}$
equal-probability outcomes (after which $k-1$ flips are needed) and
the remaining $2^{k-1}$ outcomes (after which $T_n(k-1)$ flips are needed).
\end{proof}

\subsection{Finite-Precision Implementation}
\label{appx:naive-baseline-impl}

The \hyperref[alg:sampler-bernoulli-impl]{\Bernoulli} function in \Cref{alg:sampler-bernoulli-impl}
defines an entropy-optimal random variate generator
for $\mathrm{Bernoulli}(i/k)$ bit with rational weights,
based on~\cref{eq:superfit-1,eq:superfit-2}.
The structure of this algorithm mirrors that of
\cref{alg:BinaryExpansion}, % (\cref{appx:binary}),
which extracts the \textit{concise} binary expansion $(0.p_1p_2\dots)_2$
of a rational probability $p=i/k \in (0,1)$ based on the identity
\begin{align}
i/k =
  \begin{cases}
  1/2 + 1/2 \cdot (2i-k)/k    &\mbox{if } k \le 2i \\
  0 + 1/2 \cdot (2i/k)        &\mbox{if } 2i < k.
  \end{cases}
  \label{eq:bernoulli-rational-identity}
\end{align}

The \hyperref[alg:sampler-naive-impl]{\SampleNaiveImpl}
function in \cref{alg:sampler-naive-impl}
implements \cref{alg:sampler-naive}
by successively calling \hyperref[alg:sampler-bernoulli-impl]{\Bernoulli}
on ratios of rational probabilities.
At each recursive step, the string $b \in \bool^{\le n}$ has been
determined so far.
Following \cref{eq:trefle-1}, the floats $\fL, \fR$
store the subtrahend and minuend in the probability
$p_F(b) = F (\phi_\bfmt(b1^{n-\abs{b}})) -_{\real} F (\phi_\bfmt((b0^{n-\abs{b}})^{-}))$;
and the outcomes $b0$ and $b1$
have probabilities $p_F(b0) = \fM - \fL$ and $p_F(b1) = \fR - \fM$,
respectively, where $\fM \defas F(\phi_\bfmt(b'))$
is the cumulative probability of the ``midpoint''
string $b' \in \set{0,1}^n$ that lies between $b0$ and $b1$.
\hyperref[algline:sampler-naive-impl-kn]{\ExactRatio} is any algorithm that returns
integers $(i, k)$ such that $i/k = (\fR - \fM) / (\fR - \fL)$.
While these integers cannot be computed directly using floating-point arithmetic
and are not guaranteed to fit in a single $1+E+m$ bit machine word,
the next results establish tight bounds on the finite buffer sizes
needed to store $i, k$, showing that the method does not
require arbitrary-precision arithmetic.

\begin{theorem}[name=,restate=ExactRatioWorst]
\label{prop:exact-ratio-worst}
Suppose $E \ge 1$ and $m \ge 3$. Given $\fL,\fR,\fM \in \floatEm \cap [0,1]$ with $\fL \neq \fR$,
let $i(\fL,\fR,\fM)$ and $k(\fL,\fR,\fM)$ be coprime integers such that
$i(\fL,\fR,\fM)/k(\fL,\fR,\fM) = (\fR-\fM)/(\fR-\fL)$.
Then
\begin{align}
\max_{0 \le \fL < \fM < \fR \le 1} \set*{ \ceil*{1 + \log_2 i(\fL,\fR,\fM)} + \ceil*{1 + \log_2 k(\fL,\fR,\fM)} }
  = 2\cdot (2^{E-1} + m -2) + 1.
\rlap{\;\qedhere}
\end{align}
\end{theorem}

\begin{proof}
Let $\delta \defas 2^{E-1} -2 + m \ge 2$.
Every $x \in [0,1]\cap\floatEm$ is an integer multiple of
the smallest positive subnormal $2^{-\delta}$.
Then $2^\delta (\fR-\fM)$ and $2^\delta (\fR-\fL)$
are integers between $0$ and $2^\delta$.
Put $\fL=0$, $\fR = 1$, and $\fM=2^{-\delta}$.
Then $(\fR-\fM)/(\fR-\fL) = (2^\delta-1)/2^\delta = i(\fL,\fR,\fM)/k(\fL,\fR,\fM)$
form the largest pair of coprime numbers in $\set{0,\dots,2^\delta}$.
\end{proof}

\begin{corollary}
%% For 32-bit float: 2^( 8-1)+23-1 bits =  150 bits ~  18.8 (= 150/8) bytes ~  4.7 (= 150/32) words.
%% For 64-bit float: 2^(11-1)+52-1 bits = 1075 bits ~ 134.4 (=1075/8) bytes ~ 16.8 (=1075/64) words.
In \cref{alg:sampler-naive-impl}, \hyperref[algline:sampler-naive-impl-kn]{\ExactRatio} returns integers
$i, k$ each comprised of at most $2^{E-1}+m-1$ bits, i.e., at most
5 (resp.~17) machine words on IEEE-754 single (resp.~double) precision.
\end{corollary}

%!TEX root=./paper.tex

\begin{figure}[b]
\centering

\begin{minipage}[t]{.495\linewidth}
\begin{algorithm}[H]
\captionsetup{hypcap=false}
\caption{Conditional Bit Sampling}
\label{alg:sampler-naive-impl}
\begin{algorithmic}[1]
\Require{%
  CDF $F: \bool^n \to \floatEm \cap [0,1]$ \\
  over number format $\bfmt = (n, \gamma_\bfmt, \phi_\bfmt)$ \\
  \color{gray}{String $b \in \bool^{\le n}$}; \\
  \color{gray}{Floats $\fL, \fR \in \floatEm \cap [0,1]$}
  }
\Ensure{Exact random variate $X \sim \hat{F}$}
\Function{\SampleNaiveImpl}{$F$, $b{=}\varepsilon$, $\fL{=}0$, $\fR{=}1$}
  \If{$\abs{b} = n$} \Comment{Base Case}
    \State \Return $\gamma_\bfmt(\phi_\bfmt(b))$ \Comment{Number in $\realext_{\bfmt}$}
  \EndIf
  \State $b' \gets b01^{n - \abs{b} - 1}$; $\fM \gets F(\phi_\bfmt(b'))$ \label{algline:sampler-naive-impl-F}
  \If{$\fM = \fR$} \Comment{Leaf}
    \State \Return $\SampleNaiveImpl(F, b0, \fL, \fM)$
    \Comment{0}
  \EndIf
  \If{$\fM = \fL$} \Comment{Leaf}
    \State \Return $\SampleNaiveImpl(F, b1, \fM, \fR)$
    \Comment{1}
  \EndIf
  \LComment{$i/k \defas (\fR-\fM)/(\fR-\fL)$} \label{algline:sampler-naive-impl-kn}
  \State $(i, k) \gets \ExactRatio(\fL, \fM, \fR)$
  \State $z \gets \hyperref[alg:sampler-bernoulli-impl]{\Bernoulli}(i, k)$ \Comment{Refine Subtree}
  \State $(\fL', \fR') \gets (z = 0) \,\textbf{?}\, (\fL, \fM) \textbf{:} (\fM, \fR)$
  \State \Return $\SampleNaiveImpl(F, bz, \fL', \fR')$ % \hyperref[alg:sampler-naive-impl]{\SampleNaiveImpl}
    \Comment{$z$}
\EndFunction
\end{algorithmic}
\end{algorithm}
\end{minipage}

\bigskip

\begin{minipage}[t]{.495\linewidth}
\begin{algorithm}[H]
\captionsetup{hypcap=false}
\caption{Optimal Bernoulli Generation}
\label{alg:sampler-bernoulli-impl}
\begin{algorithmic}[1]
\Require{Integers $i,k$ with $0 < i < k$}
\Ensure{Exact flip $X \sim \mathrm{Bernoulli}(i/k)$}
\Function{Bernoulli}{$i, k$}
\While{\textbf{true}}
  \State $i \gets 2i$
  \If{$i = k$}
      \State \Return \Flip() \Comment{Dyadic}
      \State ~
  \ElsIf{$i > k$}
      \State $b \gets 1$
      \State $i \gets i - k$
  \Else
       \State $b \gets 0$
  \EndIf
\If{$\Flip()$}
  \State \Return $b$ \EndIf
\EndWhile
\EndFunction
\end{algorithmic}
\end{algorithm}
\end{minipage}\hfill
\begin{minipage}[t]{.495\linewidth}
\begin{algorithm}[H]
    \centering
    \captionsetup{hypcap=false}
    \caption{Computing Binary Expansion}
    \label{alg:BinaryExpansion}
    \begin{algorithmic}[1]
    \Require{Integers $i,k$ with $0 < i < k$}
    \Ensure{\mbox{Print concise binary expansion of $i/k$}}
    \Function{BinaryExpansion}{$i, k$}
    \While{\textbf{true}}
      \State $i \gets 2i$
      \If{$i = k$}
        \State \textbf{print}\, 1 \Comment{Dyadic}
        \State \textbf{print}\, 0 \, \textbf{forever}
      \ElsIf{$i > k$}
        \State $b \gets 1$
        \State $i \gets i - k$
      \Else
        \State $b \gets 0$
      \EndIf
    \State \textbf{print}\, $b$
    \EndWhile
    \EndFunction
    \State ~
    \end{algorithmic}
\end{algorithm}
\end{minipage}
\end{figure}

\clearpage

\section{Deferred Results in \S\ref{sec:overview}}
\label{appx:overview}
%!TEX root=./paper.tex

This appendix studies the properties of random ``uniform'' floating-point numbers
obtained by dividing integers.
\Citet{Goualard2020} provide a survey and case study of these algorithms
in practice.

\begin{proposition}
  \label{prop:std-uniform-density}
  Let $\floatEm$ be a floating-point format % with $\{0,1\} \subseteq \floatEm$,
  and $\roundfl{} : \real \to \floatEm \cup \set{-\infty, +\infty}$
  be its rounding function in round-to-nearest-even mode.
  Then, for any integer $\ell \in [0, 2^{E-1}-2]$, the density of the set
  \(
  %\begin{align}
    \set{\roundfl{}({i}/{2^\ell}) \mid i \in \set{0, 1, \ldots, 2^\ell - 1} } % \subseteq \floatEm \cap [0,1]
  %\end{align}
  \)
  within the set $(\floatEm \cap [0,1])$ of all floats in the unit interval is
  \begin{align}
    \label{eq:std-uniform-density-main}
    \frac{2^\ell}{2^m(2^{E-1} -1) + 1} \quad \text{if } \ell \leq m+1,
    &&
    \frac{2^m(\ell-m+1) + 1}{2^m(2^{E-1}-1) + 1} \quad \text{if } \ell > m+1.
  \end{align}
  The same result holds for the set $ \set{\roundfl{}(\roundfl{}(i)/\roundfl{}(2^\ell)) \mid i \in \set{0, 1, \ldots, 2^\ell - 1} }$.
\end{proposition}
\begin{proof}
  Let $S \defas \set*{{i}/{2^\ell} \mid i \in \set*{0, 1, \ldots, 2^\ell - 1}}$, and
  $\roundfl{}(A) \defas \set{\roundfl{}(a) \mid a \in A}$ be the rounding of a set $A \subseteq \real$.
  Then, $\roundfl{}(S) \subseteq [0,1] \cap \floatEm$ because $S \subseteq [0,1]$ and $\set{0,1} \subseteq \floatEm$.
  Hence, for the first claim, it suffices to show that $\abs*{\roundfl{}(S)} / \abs*{[0,1] \cap \floatEm}$ equals to \cref{eq:std-uniform-density-main}.

  First, we compute $\abs*{[0,1] \cap \floatEm}$:
  \begin{align}
    \abs*{[0,1] \cap \floatEm}
    & = \abs*{[0,1) \cap \floatEm} + 1
    \\
    & = (\text{\# binades in } [0,1)) \cdot (\text{\# floats in each binade}) + 1
    \\
    & = \smash{(2^{E-1} - 1) \cdot 2^m + 1}.
  \end{align}
  Here, the number of binades in $[0,1)$ is $1 + (2^{E-1}-2) = 2^{E-1} - 1$
  for three reasons:
  $[0,\omega) \cap \floatEm$ forms precisely one binade,
  where $\omega \defas 2^{-2^{E-1}+2}$ is the smallest positive normal float in $\floatEm$;
  the smallest exponent in $[\omega,1)$ is $-2^{E-1}+2$; and
  the largest exponent in $[\omega,1)$ is $-1$.

  Next, we compute $|\roundfl{}(S)|$.
  Suppose that $\ell \leq m+1$.
  Then, for each $j \in \set{0, \ldots, \ell-1}$ and $i \in \set{2^{j}, \ldots, 2^{j+1}-1}$,
  we have $i \cdot 2^{-\ell} = (i \cdot 2^{-j}) \cdot 2^{-\ell + j} \in \floatEm$
  because $1 \leq i \cdot 2^{-j} < 2$, $i$ is an $(m+1)$-bit unsigned integer (by $0 \leq i < 2^\ell \leq 2^{m+1}$),
  and $-\ell + j$ is greater than or equal to the exponent of $\omega$ (by $-\ell+j \geq -\ell \geq -2^{E-1}+2$).
  This and $0 \in \floatEm$ imply $\roundfl{}(S) = S$ and the desired result:
  \begin{align}
    |\roundfl{}(S)| = |S| = 2^\ell.
  \end{align}

  Now, suppose that $\ell > m+1$.
  Consider the partition of $S = S_1 \cup S_2$ given by
  $S_1 \defas \set{ i \cdot 2^{-\ell} \mid i \in \set{0, \ldots, 2^{m+1}-1}}$ and
  $S_2 \defas \set{ i \cdot 2^{-\ell} \mid i \in \set{2^{m+1}, \ldots, 2^\ell-1}}$.
  Then, $\roundfl{}(S_1) = S_1$ by the above argument.
  For $\roundfl{}(S_2)$, we claim that $\roundfl{}(S_2) = [2^{-\ell+m+1}, 1] \cap \floatEm$.
  This subclaim immediately implies the desired result:
  \begin{align}
    \abs*{\roundfl{}(S)} = \abs*{\roundfl{}(S_1)} + \abs*{\roundfl{}(S_2)}
    = \abs*{S_1} + \abs*{[2^{-\ell+m+1}, 1] \cap \floatEm}
    = 2^{m+1} + (\ell-m-1) \cdot 2^m + 1 %% = 2^m (\ell-m+1) + 1,
  \end{align}
  where the first equality holds by $\roundfl{}(S_1) \cap \roundfl{}(S_2) = \emptyset$
  (since $\roundfl{}(S_1) = S_1 \subseteq [0, 2^{-\ell+m+1})$).
  Hence, it suffices to show the above subclaim. We prove this claim in three steps.
  \begin{enumerate}
  \item
    We have $\roundfl{}(S_2) \subseteq [2^{-\ell+m+1}, 1] \cap \floatEm$
    because $S_2 \subseteq [2^{-\ell+m+1}, 1]$ and $\set*{2^{-\ell+m+1}, 1} \subseteq \floatEm$,
    where $2^{-\ell+m+1} \in \floatEm$ is by $-\ell+m+1 \geq -\ell \geq -2^{E-1}+2$.
  \item
    We have $\roundfl{}(S_2) \supseteq \set{1}$
    because $\roundfl{}((2^\ell-1)/2^\ell) = \roundfl{}(1-2^{-\ell}) = 1$,
    where the last equality is by $1 \in \floatEm$, $2^{-\ell} \leq \frac{1}{2} 2^{-1-m}$,
    and the round-to-nearest-even mode of $\roundfl{}(\cdot)$.
  \item
    We show $\roundfl{}(S_2) \supseteq [2^{-\ell+m+1}, 1) \cap \floatEm$.
    To prove this claim, we write $[2^{-\ell+m+1}, 1) \cap \floatEm$ as
    \begin{align}
      \set*{ (1.b_1 \ldots b_m)_{2} \cdot 2^{-\ell+m+j}
        \mid b_1, \ldots, b_m \in \set{0,1}, j \in \set{1, \ldots, \ell-m-1}}.
    \end{align}
    Here, we have $(1.b_1 \ldots b_m)_{2} \cdot 2^{-\ell+m+j} = ((1 b_1 \ldots b_m)_{2} \cdot 2^j) \cdot 2^{-\ell}$,
    where $(1 b_1 \ldots b_m)_{2} \cdot 2^j \in [2^{m+1}, 2^\ell-1]$ by $m+j \geq m+1$ and $m+1+j \leq \ell$.
    Hence, $[2^{-\ell+m+1}, 1) \cap \floatEm \subseteq S_2$,
    implying that $[2^{-\ell+m+1}, 1) \cap \floatEm \subseteq \roundfl{}(S_2)$.
  \end{enumerate}
  The subclaim is thus established, completing the proof of the first claim.

  For the second claim, the proof is exactly the same because
  \begin{enumerate*}[label=(\roman*)]
    \item $\roundfl{}(2^\ell) = 2^\ell$ by $0 \leq \ell \leq 2^{E-1}-2$; and
    \item the above proof used only those $i \in \set{0, \ldots, 2^\ell-1}$
    such that $i = (1.b_1 \ldots b_m)_2 \cdot 2^{j}$ for some $b_1, \ldots, b_m \in \set{0,1}$ and $j \in \set{0, \ldots, \ell-1}$,
    which satisfy $\roundfl{}(i) = i$.
  \end{enumerate*}
\end{proof}

\begin{remark}
  \label{remark:uniform-real-world}
  \cref{prop:std-uniform-density} describes a standard method~\citep{Goualard2020}
  to generate a floating-point random variate from $\mathrm{Uniform}([0,1])$,
  and presents the proportion of floats in $[0,1]$ covered by this method.
  We list some of the actual implementations of this method and show relevant details.
  \begin{itemize}
  \item
    In the GNU Standard C++ Library (libstdc++), the function
    \href{https://en.cppreference.com/w/cpp/numeric/random/generate_canonical}
         {\texttt{std::generate\_canonical()}}
    generates a 64-bit float with $\ell=64$ and a 32-bit float with $\ell=32$ (when invoked with
    \href{https://en.cppreference.com/w/cpp/numeric/random/mersenne_twister_engine}
         {\texttt{std::mt19937}}),%
    \footnote{
    %% ========================================================================================
    %% std::generate_canonical():
    %%         https://github.com/gcc-mirror/gcc/blob/releases/gcc-14/libstdc%2B%2B-v3/include/bits/random.tcc#L3349 [<--- double/float]
    %% std::mt19937:
    %%         https://github.com/gcc-mirror/gcc/blob/releases/gcc-14/libstdc%2B%2B-v3/include/bits/random.h#L1717
    %%         https://en.cppreference.com/w/cpp/numeric/random/mersenne_twister_engine
    %% --> {min,max}():
    %%         https://github.com/gcc-mirror/gcc/blob/releases/gcc-14/libstdc%2B%2B-v3/include/bits/random.h#L672
    %%         https://github.com/gcc-mirror/gcc/blob/releases/gcc-14/libstdc%2B%2B-v3/include/bits/random.h#L679
    %%         https://en.cppreference.com/w/cpp/numeric/random/mersenne_twister_engine/min
    %%         https://en.cppreference.com/w/cpp/numeric/random/mersenne_twister_engine/max
    %% ========================================================================================
    \url{https://github.com/gcc-mirror/gcc/blob/releases/gcc-14.2.0/libstdc++-v3/include/bits/random.tcc\#L3370}
    }
    covering $1.27$\% of 64-bit floats in $[0,1]$ and $7.87$\% of 32-bit floats in $[0,1]$.
    This function for 32-bit floats is similar to the first code on page 2 of \citet{Downey2007}.
  \item
    In the GNU Scientific Library (GSL), the function
    \href{https://www.gnu.org/software/gsl/doc/html/rng.html#c.gsl_rng_uniform}
         {\texttt{gsl\_rng\_uniform()}}
    generates a 64-bit float with $\ell=32$ (when invoked with
    \href{https://www.gnu.org/software/gsl/doc/html/rng.html#c.gsl_rng_mt19937}
         {\texttt{gsl\_rng\_mt19937}}),%
    \footnote{
    %% ========================================================================================
    %% gsl_rng_uniform():    https://github.com/ampl/gsl/blob/v2.7.0/rng/gsl_rng.h#L165
    %% gsl_rng_mt19937:      https://github.com/ampl/gsl/blob/v2.7.0/rng/mt.c#L234
    %% --> mt_type:          https://github.com/ampl/gsl/blob/v2.7.0/rng/mt.c#L214
    %% --> mt_get_double():  https://github.com/ampl/gsl/blob/v2.7.0/rng/mt.c#L125 [<--- double]
    %% ========================================================================================
    \url{https://github.com/ampl/gsl/blob/v2.7.0/rng/mt.c\#L127}
    }
    covering only $9.32 \times 10^{-8}$\% of 64-bit floats in $[0,1]$.
    This function corresponds to $\texttt{uniform()}$ in \cref{lst:exponential}.
    The GSL, however, does not provide a 32-bit version of this function.
  \item
    In Python and SciPy, the functions
    \href{https://docs.python.org/3/library/random.html#random.random}
         {\texttt{random.random()}} and
    \href{https://docs.scipy.org/doc/scipy/reference/generated/scipy.stats.uniform.html}
         {\texttt{scipy.stats.}\allowbreak\texttt{uniform.rvs()}}
    generate a 64-bit float with $\ell=53$,%
    \footnote{
    Python: \url{https://github.com/python/cpython/blob/v3.13.0/Modules/_randommodule.c\#L191}
    \\\phantom{${}^8$}%
    SciPy: \url{https://github.com/numpy/numpy/blob/v2.0.0/numpy/random/src/mt19937/mt19937.h\#L58}
    }%
    covering $0.20$\% of 64-bit floats in $[0,1]$.
    The \texttt{random} module and SciPy, however, do not provide a 32-bit version of this function.
  \item
    In NumPy and PyTorch, the functions
    \href{https://numpy.org/doc/2.0/reference/random/generated/numpy.random.Generator.random.html}
         {\texttt{numpy.Generator.random()}} and
    \href{https://pytorch.org/docs/2.3/generated/torch.rand.html}
         {\texttt{torch.rand()}}
    generate a 64-bit float with $\ell=53$ and a 32-bit float with $\ell=24$ (when invoked with
    \href{https://numpy.org/doc/2.0/reference/random/bit_generators/pcg64.html#numpy.random.PCG64}
         {\texttt{numpy.}\allowbreak\texttt{random.}\allowbreak\texttt{PCG64}} for NumPy),%
    \footnote{
    NumPy (64-bit): \url{https://github.com/numpy/numpy/blob/v2.0.0/numpy/random/_common.pxd\#L69}
    \\\phantom{${}^8$}%
    NumPy (32-bit): \url{https://github.com/numpy/numpy/blob/v2.0.0/numpy/random/src/distributions/distributions.c\#L20}
    \\\phantom{${}^8$}%
    PyTorch: \url{https://github.com/pytorch/pytorch/blob/v2.3.0/aten/src/ATen/core/TransformationHelper.h\#L85}
    }
    covering $0.20$\% of 64-bit floats in $[0,1]$ and $1.57$\% of 32-bit floats in $[0,1]$, respectively.
    \hfill \qedhere
  \end{itemize}
\end{remark}

%!TEX root=./paper.tex
\begin{table}[t]
\centering
\caption{Comparison of coverage of floating-point numbers in the unit
interval when generating ``uniform'' random variates using exact generation
from the \WCDF{} \texttt{cdf\_uniform\_round\_dn}
(\cpageref{lst:gsl-samplers} of main text) and the usual division
method~\citep{Goualard2020}.
Results are shown for the 8-bit binary number format $\float^5_2$,
i.e., 5 exponent bits and 2 mantissa bits, which contains 60 floats in the
interval $[0,1)$.
The `No.~of Observed Samples' column shows the empirical frequency with which
each float was observed in a sample of 10,000,000 random variates generated
using the exact and division method.
The division method covers only $2/32$ floats in the first 8 binades (left
table, 0 and 0.0039062500), whereas the exact method covers all floats with
the correct frequencies.
For 32-bit or 64-bit floating-point systems, the division method
covers an even smaller fraction of all the possible floats, whereas the exact
method retains 100\% coverage.
}
\label{table:float-coverage}
\begin{adjustbox}{max width=\linewidth}
\centering
\footnotesize
\begin{tabular}[t]{|rlrr|}
% \hline
\multicolumn{1}{c}{~} & ~     & \multicolumn{2}{c}{\bfseries No.~of Observed Samples} \\ \cline{3-4}
\multicolumn{1}{c}{~} & \textbf{Float} & {Exact Method} & \multicolumn{1}{c}{Division Method} \Tstrut \\ \hline
0  & 0.0000000000000000 & 164     & 39385 \\
1  & 0.0000152587890625 & 146     & 0 \\
2  & 0.0000305175781250 & 173     & 0 \\
3  & 0.0000457763671875 & 155     & 0
\\
4  & 0.0000610351562500 & 145     & 0 \\
5  & 0.0000762939453125 & 129     & 0 \\
6  & 0.0000915527343750 & 166     & 0 \\
7  & 0.0001068115234375 & 136     & 0
\\
8  & 0.000122070312500 & 293     & 0 \\
9  & 0.000152587890625 & 308     & 0 \\
10 & 0.000183105468750 & 292     & 0 \\
11 & 0.000213623046875 & 321     & 0
\\
12 & 0.00024414062500 & 638     & 0 \\
13 & 0.00030517578125 & 643     & 0 \\
14 & 0.00036621093750 & 689     & 0 \\
15 & 0.00042724609375 & 618     & 0
\\
16 & 0.0004882812500 & 1197    & 0 \\
17 & 0.0006103515625 & 1170    & 0 \\
18 & 0.0007324218750 & 1175    & 0 \\
19 & 0.0008544921875 & 1239    & 0
\\
20 & 0.000976562500 & 2412    & 0 \\
21 & 0.001220703125 & 2345    & 0 \\
22 & 0.001464843750 & 2442    & 0 \\
23 & 0.001708984375 & 2439    & 0
\\
24 & 0.00195312500 & 4870    & 0 \\
25 & 0.00244140625 & 4778    & 0 \\
26 & 0.00292968750 & 4858    & 0 \\
27 & 0.00341796875 & 4814    & 0
\\
28 & 0.0039062500 & 9716    & 39080 \\
29 & 0.0048828125 & 9657    & 0 \\
30 & 0.0058593750 & 9801    & 0 \\
31 & 0.0068359375 & 9872    & 0
\\\hline\hline
\end{tabular}\quad
\begin{tabular}[t]{|rlrr|}
\multicolumn{1}{c}{~} & ~     & \multicolumn{2}{c}{\bfseries No.~of Observed Samples} \\ \cline{3-4}
\multicolumn{1}{c}{~} & \textbf{Float} & {Exact Method} & \multicolumn{1}{c}{Division Method} \Tstrut \\ \hline
32 & 0.007812500 & 19499   & 39273 \\
33 & 0.009765625 & 19440   & 0 \\
34 & 0.011718750 & 19771   & 39027 \\
35 & 0.013671875 & 19797   & 0
\\
36 & 0.01562500 & 38848   & 39205 \\
37 & 0.01953125 & 39399   & 39012 \\
38 & 0.02343750 & 39118   & 38920 \\
39 & 0.02734375 & 38898   & 39177
\\
40 & 0.0312500 & 78148   & 78228 \\
41 & 0.0390625 & 77943   & 77833 \\
42 & 0.0468750 & 77378   & 77494 \\
43 & 0.0546875 & 78279   & 77863
\\
44 & 0.062500 & 156466  & 156536 \\
45 & 0.078125 & 156020  & 156099 \\
46 & 0.093750 & 156554  & 155971 \\
47 & 0.109375 & 156471  & 157199
\\
48 & 0.12500 & 313092  & 312150 \\
49 & 0.15625 & 313487  & 312113 \\
50 & 0.18750 & 313109  & 311700 \\
51 & 0.21875 & 311921  & 312760
\\
52 & 0.2500 & 625607  & 625439 \\
53 & 0.3125 & 625895  & 624591 \\
54 & 0.3750 & 624381  & 626462 \\
55 & 0.4375 & 624009  & 624240
\\
56 & 0.500 & 1250741 & 1249658 \\
57 & 0.625 & 1249698 & 1251008 \\
58 & 0.750 & 1249268 & 1250410 \\
59 & 0.875 & 1248962 & 1249167 \\
\hline\hline
\end{tabular}
\end{adjustbox}
\end{table}

\clearpage

\section{Deferred Results in \S\ref{sec:binary}}
\label{appx:binary}
%!TEX root=./paper.tex

\LiftRVG*
\begin{proof}
Write $\dom(X) = \set{u_1, u_2, \ldots} \subset \set{0,1}^*$ and
define $I_i \defas [(0.u_i)_2, (0.u_i \bar{1})_2)$.
By the first property
in \cref{eq:isothujone} (prefix-free),
the intervals $I_1, I_2, \dots$ are pairwise disjoint.
By the second property in \cref{eq:isothujone} (exhaustive),
the union of these disjoint intervals forms a measure one subset of $[0,1]$, since
\begin{align}
\lambda(\cup_i I_i)
=\textstyle \sum_i \lambda(I_i)
&=\textstyle \sum_i \left[(0.u_i 1^\infty)_2 - (0.u_i 0^\infty)_2\right] \\
&=\textstyle \sum_i \left[((0.u_i 0^\infty)_2 + 2^{-\abs{u_i}}) - (0.u_i 0^\infty)_2\right]
=\textstyle \sum_i 2^{-\abs{u_i}} = 1.
\end{align}
Then for each $\omega \in [0,1]$,
let $X_\gamma(\omega) \defas \gamma(X(u_j))$ if there exists $j$ such that $\omega \in I_j$,
and arbitrarily otherwise, on the measure zero set $[0,1]\setminus \cup_i I_i$.
With this construction, for any $x \in \real$ we have
\begin{align}
\textstyle\Pr(X_\gamma = x)
&= \Pr\left(\bigcup_{i=1}^{\infty}\set{I_i \mid \gamma(X(u_i))=x}\right)
\\
&= \sum_{i=1}^{\infty} \lambda(I_i)\mathbf{1}[\gamma(X(u_i)) = x]
= \sum_{i=1}^{\infty} 2^{-\abs{u_i}}\mathbf{1}[\gamma(X(u_i)) = x].
\end{align}
\end{proof}

The remainder of this section proves the correctness (\cref{theorem:sampler-correct}) and
entropy-optimality (\cref{theorem:sampler-optimal}) of \cref{alg:sampler-opt},
which together imply \cref{theorem:sampler-opt}.

\begin{proposition}
  \label{proposition:bin-exp-zeros}
  For any $z,x,y \in [0,1]$ with $z = x + y$, let
  $z = (z_0.z_1z_2\ldots)_2$,
  $x = (x_0.x_1x_2\ldots)_2$, and
  $y = (y_0.y_1y_2\ldots)_2$
  be concise binary expansions.
  Suppose $\ell \geq 0$ is any index.
  If $x_\ell y_\ell z_\ell \in \set{000, 011, 101, 110}$ and $z_j = 0$ for all $j > \ell$,
  then $x_j y_j = 00$ for all $j > \ell$.
\end{proposition}
\begin{proof}
  Toward a contradiction, assume $x_k = 1$ or $y_k = 1$ for some $k > \ell$.
  Let $x' \defas (x_\ell. x_{\ell+1} \ldots)_2$,
  $y' \defas (y_\ell. y_{\ell+1} \ldots)_2$, and
  $z' \defas (z_\ell. z_{\ell+1} \ldots)_2 = z_\ell$.
  Then, $x_\ell + y_\ell = (x_\ell.0\ldots)_2 + (y_\ell.0\ldots)_2 < x' + y' < (x_\ell.1\ldots)_2 + (y_\ell.1\ldots)_2 = x_\ell + y_\ell + 2$,
  where the second $<$ is by the conciseness of $(x_\ell. x_{\ell+1} \ldots)_2$ and $(y_\ell. y_{\ell+1} \ldots)_2$.
  On the other hand, $x+y = z$ implies $x' + y' \in \set{z', z' + (10.0)_2} = \set{z_\ell, z_\ell + 2}$,
  where the $\in$ is by the conciseness of $(x_\ell. x_{\ell+1} \ldots)_2$ and $(y_\ell. y_{\ell+1} \ldots)_2$,
  and the $=$ is by assumption.
  Since $x_\ell + y_\ell \in \set{z_\ell, z_\ell+2}$ by assumption, we get a contradiction.
\end{proof}

\BitPatternsSimple*
\begin{proof}
For each $j \ge 0$, write $z_{j} = x_j + y_j + c_j \pmod{2}$,
where $c_j \in \set{0,1}$ is the ``carry'' bit at location $j$ of the binary
addition.
The three patterns in \cref{eq:bit-patterns-simple} correspond to
three different cases of $(x_\ell x_{\ell+1} \ldots)$ and $(y_\ell y_{\ell+1} \ldots)$.

\begin{enumerate}[wide,label={\textbf{\ref{eq:bit-patterns-simple}.\arabic*}},]

\item
Assume $x_\ell y_\ell \in \set{01, 10}$.
Since $x_\ell y_\ell z_\ell \in \set{011, 101}$ and $z_j = 0$ for all $j > \ell$,
\cref{proposition:bin-exp-zeros} implies that $x_j y_j = 00$ for all $j > \ell$.

\item
Assume $x_\ell y_\ell \in \set{00, 11}$ and $x_j y_j \neq 11$ for all $j > \ell$.
We argue that $x_j y_j \neq 00$ for all $j > \ell$.
For contradiction, assume $x_k y_k = 00$ for some $k > \ell$; let $k$ be the smallest such index.
Then, $c_{k-1} = 0$ holds. %, where $c_{k-1}$ is the carry bit at location $k-1$.
If $k = \ell+1$, then $x_\ell + y_\ell + c_\ell = 0 \allowbreak \neq 1 = z_\ell \pmod{2}$, a contradiction.
If $k > \ell+1$, then $x_{k-1} y_{k-1} = 11$ must hold by
$x_{k-1} + y_{k-1} + c_{k-1} \allowbreak = z_{k-1} = 0 \pmod{2}$, $c_{k-1}=0$, and the minimality of $k$.
This contradicts to $x_j y_j \neq 11$ for all $j > \ell$.

\item
Assume $x_\ell y_\ell \in \set{00, 11}$ and $x_k y_k = 11$ for some $k > \ell$.
Let $k$ be the smallest such index.
Since $x_k y_k z_k = 110$ and $z_j = 0$ for all $j > k$, \cref{proposition:bin-exp-zeros} implies that $x_j y_j = 00$ for all $j > k$.
By the minimality of $k$, it suffices show that $x_j y_j \neq 00$ for all $j \in (\ell, k)$.
Toward a contradiction, assume not.
Then, $x_{j'} y_{j'} z_{j'} = 000$ for some $j' \in (\ell, k)$.
Since $z_j = 0$ for all $j > j'$, \cref{proposition:bin-exp-zeros} implies that $x_j y_j = 00$ for all $j > j'$.
This contradicts to $x_k y_k = 11$ and $k > j'$.
\qedhere
\end{enumerate}
\end{proof}

\begin{remark}
In \cref{proposition:bit-patterns-simple}, if $z = 1$ and $x \in (0,1)$,
then $x_1y_1 \ne 00$ (otherwise $x + y < 1$).
\end{remark}

\BitPatterns*
\begin{proof}
For each $j \ge 0$, write $z_{j} = x_j + y_j + c_j \pmod{2}$,
where $c_j \in \set{0,1}$ is the ``carry'' bit at location $j$ of the binary
addition.
The three patterns in \cref{eq:bit-patterns} correspond to
three different cases of the final bits $x_{\ell'}y_{\ell'}$ and initial bits $x_{\ell}y_{\ell}$.

\begin{enumerate}[wide, label={\textbf{\ref{eq:bit-patterns}.\arabic*}},]

\item[\textbf{\ref{eq:bit-patterns}.2}]
Assume $x_{\ell'}y_{\ell'} = 11$.
We first prove that $c_{j} = 1$ and $x_{j}y_{j} \in \set{01,10}$ for each $j \in \set{\ell+1,\dots,\ell'-1}$, by induction.
For the base case ($j = \ell'-1$), $x_{\ell'}y_{\ell'} = 11$ implies $c_{\ell'-1} = 1$.
Since $x_{\ell'-1} + y_{\ell'-1} + c_{\ell'-1} = z_{\ell'-1} = 0 \pmod{2}$, we must have $x_{\ell'-1}y_{\ell'-1} \in \set{01,10}$.
For the inductive case ($\ell+1 \leq j < \ell'-1$), the induction hypothesis that $x_{j+1}y_{j+1} \in \set{01,10}$ and $c_{j+1} = 1$ implies $c_{j}=1$.
But $x_{j} + y_{j} + c_{j} = z_{j} = 0 \pmod{2}$ again implies $x_{j}y_{j} \in \set{01,10}$.

We next prove $x_\ell y_\ell \in \set{00, 11}$.
Since $x_\ell + y_\ell + c_\ell = z_\ell = 1 \pmod{2}$, it suffices to show $c_\ell=1$.
If $\ell' = \ell+1$ then $x_{\ell'} y_{\ell'} = 11$ implies this.
If $\ell' > \ell+1$ then $c_{\ell+1} = 1$ and $x_{\ell+1} y_{\ell+1} \in \set{01, 10}$ imply this.

\item[\textbf{\ref{eq:bit-patterns}.1}]
Assume $x_{\ell'}y_{\ell'} \neq 11$ and $x_{\ell}y_{\ell} \in \set{01,10}$.
From $x_{\ell'}y_{\ell'} \in \set{00, 01, 10}$, we have $c_{\ell'-1} = 0$.
This implies $\hat{z} = \hat{x} + \hat{y}$ for
$\hat{z} \defas (z_0. z_1 \ldots z_{\ell'-1} 0 \ldots)_2$,
$\hat{x} \defas (x_0. x_1 \ldots x_{\ell'-1} 0 \ldots)_2$, and
$\hat{y} \defas (y_0. y_1 \ldots y_{\ell'-1} 0 \ldots)_2$.
Since $z_j = 0$ for all $\ell < j < \ell'$, we can apply \cref{proposition:bit-patterns-simple} to $(\hat{z}, \hat{x}, \hat{y}, \ell)$.
In the theorem, only Pattern~\ref{eq:bit-patterns-simple}.1 matches because $x_\ell y_\ell z_\ell \in \set{011, 101}$.
This yields the desired pattern on $(x_{\ell+1} y_{\ell+1}, \ldots, x_{\ell'-1} y_{\ell'-1})$ described in Pattern~\ref{eq:bit-patterns}.1.

\item[\textbf{\ref{eq:bit-patterns}.3}]
Assume $x_{\ell'}y_{\ell'} \neq 11$ and $x_{\ell}y_{\ell} \in \set{00,11}$.
As in the previous case, we have $c_{\ell'-1} = 0$ and $\hat{z} = \hat{x} + \hat{y}$,
where $\hat{z}$, $\hat{x}$, and $\hat{y}$ are defined as before.
Since $z_j = 0$ for all $\ell < j < \ell'$, we can apply \cref{proposition:bit-patterns-simple} to $(\hat{z}, \hat{x}, \hat{y}, \ell)$.
In the theorem, only Pattern~\ref{eq:bit-patterns-simple}.3 matches because $x_\ell y_\ell z_\ell \in \set{001, 111}$
and $\hat{x}$ and $\hat{y}$ have trailing zeros in their binary expansions.
This yields the desired pattern on $(x_{\ell+1} y_{\ell+1}, \ldots, x_{\ell'-1} y_{\ell'-1})$ described in Pattern~\ref{eq:bit-patterns}.3.
\qedhere
\end{enumerate}
\end{proof}

\begin{proposition}
\label{proposition:loop-bound}
In $\hyperref[alg:sampler-opt]{\SampleOpt}(p,b,\ell)$, the random
number of loop iterations has a least upper bound
$j \ge 1$ if and only if $j$ is the smallest integer
such that $p^{b0}_{\ell+j}p^{b1}_{\ell+j} = 11$.
\end{proposition}

\begin{proof}
If the algorithm reaches step $j$, then it exits at this step with
probability one if and only if $p^{b0}_{\ell+j}p^{b1}_{\ell+j} = 11$
(by case analysis); so $j$ is an upper bound on the number of loop iterations.
It is the least upper bound if and only if the algorithm reaches step $j$ with
nonzero probability, which holds if and only if there is no $j' < j$
such that $p^{b0}_{\ell+j'}p^{b1}_{\ell+j'} = 11$ (otherwise it would exit at step $j'$).
\end{proof}

\begin{theorem}
\label{theorem:sampler-bits}
Let $p$ be a binary-coded probability distribution and fix an integer
$n \ge 1$.
Let $B_n$ denote the first $n$ bits (stored in variable $b$) generated by $\hyperref[alg:sampler-opt]{\SampleOpt}(p)$ and
$C_n$ the number of coin flips (stored in variable $\ell$) made up to and including
the point when the $n$-th bit is generated.
Then $\Pr(B_n=b, C_n=\ell) = 2^{-\ell}\mathbf{1}[[p(b)]_\ell=1]$.
\end{theorem}

\begin{proof}
The proof is by induction on the number of recursive calls $n$
to \hyperref[alg:sampler-opt]{\SampleOpt}.

\begin{enumerate}[wide=0pt,parsep=0pt,leftmargin=*]

\item[\textit{Base Case}.] Suppose $n=1$.
Consider the invocation $\hyperref[alg:sampler-opt]{\SampleOpt}(p, b=\varepsilon, \ell=0)$.

\begin{enumerate}[wide=\parindent, parsep=0pt, leftmargin=*]

\item[Case: $(p^0,p^1) = (1,0)$.]
In the concise binary expansion, we have
$(p^0_0,p^1_0) = (1,0)$.
Then 0 is generated with probability 1, the loop is never
entered, and exactly 0 flips are made.
An analogous result holds for the symmetric case $(p^0,p^1) = (0,1)$.

\item[Case: $p^0 \in (0,1)$.]
Fix $\ell \ge 1$ is such that $p^0_\ell = 1$.
Since $p^0+p^1=1$, $p^0_1p^1_1 \ne 00$.
\Cref{proposition:bit-patterns-simple} establishes that
$p^0_jp^1_j \in \set{01,10}$ for all $j=1,\dots,\ell-1$.
Let $x_j$ denote the outcome of $\Flip$ at iteration $j$.
At iteration $j$, the loop continues if and only if
$p^0_jp^1_j = 01 \wedge x_j = 0$ or $p^0_jp^1_j = 10 \wedge x_j = 1$.
Therefore, \cref{algline:UnivSampler-Print0} occurs after
exactly $i$ coin flips if and only if $x_j = p^0_j$
($j=1,\dots,\ell-1$) and $x_\ell = 1$; which by independence of
$\Flip$ occurs with probability $2^{-\ell}$.
Further, if $p^0_\ell = 0$, then \cref{algline:UnivSampler-Print0} at
iteration $\ell$ is never entered, so $0$ cannot possibly be generated with
exactly $\ell$ flips.
The case of $p^1 \in (0,1)$ is similar.
\end{enumerate}

\item[\textit{Inductive Case}.]
Assume the claim holds for $n > 1$.
Consider the recursive call $\hyperref[alg:sampler-opt]{\SampleOpt}(p, b, \ell)$,
where $b \in \set{0,1}^n$ is the string generated using exactly $\ell$ coin flips.
From the inductive hypothesis, $p^b_\ell = 1$ and the probability
of the current execution path is $2^{-\ell}$.
Let $\ell' = \min_{i > \ell}\set{p^b_i = 1}$ be the index (possibly
infinite) of the next 1 in
the expansion of $p^b$.
We analyze the event that the next generated bit at this stage
of the recursion is $0$,
which means the overall generated string is $b0$
(the analysis for $b1$ is entirely symmetric to $b0$).

\begin{enumerate}[wide=\parindent, parsep=0pt, leftmargin=*]

\item[Case: $\ell' < \infty$.]
Consider the three patterns from \cref{proposition:bit-patterns}
for the binary expansions $(p^{b0}_{\ell}, p^{b0}_{\ell+1},\dots)$
and $(p^{b1}_{\ell}, p^{b0}_{\ell+1},\dots)$.
Suppose Pattern \ref{eq:bit-patterns}.1 is matched.
If $p^{b0}_{\ell}p^{b1}_{\ell} = 10$,
then \cref{algline:UnivSampler-Print0-NoIter} ensures that $b0$ is
generated with 0 additional flips, so the path probability remains
$2^{-\ell}$ as desired.
Suppose Pattern \ref{eq:bit-patterns}.2 or
\ref{eq:bit-patterns}.2 are matched, so the loop is entered.
Let $k \in \set{1,\dots,\ell'-\ell}$ be such that $p^{b0}_{\ell+k}=1$.
By an analogous argument to the base case, the probability of exiting
at \cref{algline:UnivSampler-Print0} after $k$ loop iterations
(i.e., $k$ additional flips)
is precisely $2^{-k}$, which gives $\ell + k$ flips overall
with path probability $2^{-(\ell+k)}$.

\item[Case: $\ell' = \infty$.]
% \Cref{proposition:bit-patterns-simple} gives the corresponding bit pattern.
Pattern \ref{eq:bit-patterns-simple} shows the bit configuration.
If $p^{b0}_\ell p^{b1}_\ell \in \set{10,01}$ the loop is not entered,
as in Pattern \ref{eq:bit-patterns}.1.
Otherwise if $p^{b0}_\ell p^{b1}_\ell = 00$, loop is entered as in
Pattern \ref{eq:bit-patterns}.2 and \ref{eq:bit-patterns}.3.
\end{enumerate}

Finally, we prove that every index $k \in \mathbb{N}$ with $p^{b0}_k = 1$
has a positive probability of being encountered in an execution
path of a recursive call
$\hyperref[alg:sampler-opt]{\SampleOpt}(p,b,\ell)$, for some $\ell \ge 0$.
From the inductive hypothesis, every $\ell$ such that $z_\ell =1$
is encountered with probability $2^{-\ell} > 0$.
It suffices
to prove that all the 1 bits among $p^{b0}_{\ell}, p^{b0}_{\ell+1}, \ldots, p^{b0}_{\ell'}$
are selected with positive probability.

\begin{enumerate}[wide=\parindent, parsep=0pt, leftmargin=*]

\item[Case: the loop is entered.]
By \cref{proposition:loop-bound}
it suffices to prove there are no 1 bits after a loop index $j$
such that $p^{b0}_{j}p^{b1}_j = 11$.
If $\ell' < \infty$, apply \cref{proposition:bit-patterns}
(Pattern \ref{eq:bit-patterns}.3) to conclude.
If $\ell' = \infty$, apply \cref{proposition:bit-patterns-simple}
(Pattern \ref{eq:bit-patterns-simple}) to conclude.

\item[Case: the loop is not entered.]
If $\ell' < \infty$, apply
\cref{proposition:bit-patterns} (Pattern \ref{eq:bit-patterns}.1) to
conclude that all the skipped bits are 0.
If $\ell' = \infty$, we have $p^{b0}_\ell p^{b1}_\ell = 10$ (wlog).
If $p^{b0} = p^b = (1.000\ldots)_2$, then $\ell = 0$, and the
conclusion is immediate. Otherwise $p^{b0} < 1$.
Assume for a contradiction there is a minimal $j > \ell$ such that
$p^{b0}_j = 1$ and $p^{b0}_i = 0$ for $\ell < i < j$.
As $p^b_j = 0 = 1 + p^{b1}_j + c_j \pmod{2}$ it follows that $c_{j-1}=1$.
But $p^b_{j-1} = 0 = 0 + p^{b1}_{j-1} + c_{j-1} \pmod{2}$ so $c_{j-2} = 1$.
Apply repeatedly to conclude that $c_{\ell} = 1$.
But $p^b_\ell = 1 = 1 + 0 + 1 \pmod{2} = 0$, a contradiction.
\end{enumerate}

As $b$ was arbitrary, the statement holds for all $b' \in \set{0,1}^{n+1}$.
\qedhere
\end{enumerate}
\end{proof}

\begin{corollary}
\label{theorem:sampler-correct}
For each $n \ge 0$ and $b \in \set{0,1}^n$
the probability $\hyperref[alg:sampler-opt]{\SampleOpt}(p)$ generates a string matching
$b\bool^*$ is $p(b)$.
\end{corollary}

\begin{proof}
Define $B_n$ and $C_n$ as in the statement of \cref{theorem:sampler-bits}.
Then
\begin{align}
\Pr(B_n=b)
  = \sum_{\ell=0}^{\infty} \Pr(B_n=b, C_n=\ell)
  = \sum_{\ell=0}^{\infty} 2^{-\ell}\mathbf{1}[[p(b)]_\ell = 1]
  = \sum_{\ell=0}^{\infty} 2^{-\ell}[p(b)]_\ell
  = p(b).
\end{align}
\end{proof}

\begin{corollary}
\label{theorem:sampler-optimal}
For every binary-coded distribution $p$ and integer $n \ge 1$, \hyperref[alg:sampler-opt]{\SampleOpt}
defines an entropy-optimal generator for the discrete distribution
$P_n \defas \set{b \mapsto p(b) \mid b \in \bool^n}$, i.e., its average
number of random coin flips $C_n$ is minimal among all exact sampling
algorithms for $P_n$.
\end{corollary}

\begin{proof}
% By total expectation
Define $B_n$ and $C_n$ as in the statement of \cref{theorem:sampler-bits}.
Then
\begin{align}
\mathbb{E}[C_n]
  = \sum_{\ell = 0}^{\infty} \ell \cdot \Pr(C_n = \ell)
  &= \sum_{\ell = 0}^{\infty} \ell \left(\sum_{b \in \bool^n}\Pr(C_n = \ell, B_n=b)\right)\\
  &= \sum_{\ell = 0}^{\infty} \ell \left(\sum_{b \in \bool^n}2^{-\ell}[p(b)]_\ell\right)\\
  &= \sum_{b\in\set{0,1}^n}\sum_{\ell=0}^\infty \ell \cdot 2^{-\ell}[p(b)]_\ell
\end{align}
which is precisely the Knuth-Yao lower bound.
It follows from \cite{knuth1976}
that $H(P_n) \le \mathbb{E}[C_n] \le H(P_n) + 2$, where $H$ is the binary
entropy function.
\end{proof}

\SamplerOpt*

\begin{proof}
This result is a restatement of \cref{theorem:sampler-correct,theorem:sampler-optimal}.
\end{proof}

\paragraph{Implementation of \cref{alg:sampler-opt} using Lazy Computation}

\Cref{alg:sampler-opt}, which generates a stream of random bits from a
binary-coded probability distribution, can be directly implemented in a
programming language that supports lazy computation.
\Cref{lst:sampler-opt-haskell} shows one such implementation in Haskell,
using a guarded recursive call
on \cref{lst:sampler-opt-haskell-recurse}
that occurs in the data constructor (\texttt{:}) for lists

\begin{listing}[t]
\captionsetup{skip=5pt}
\caption{Implementation of \cref{alg:sampler-opt}, which generates a stream
of random bits from a binary-coded probability distribution,
using lazy computation with a guarded recursive call
in the Haskell programming language.}
\label{lst:sampler-opt-haskell}
\begin{lstlisting}[style=CC,language=haskell,basicstyle=\ttfamily\footnotesize,frame=single]
type Bit = Int
type BinaryString = [Bit]
type BinaryCodedDist = [Bit] -> Float

-- Obtain a fair random bit from the entropy source.
randBit :: Bit

-- Extract bit from a float (@\color{codegray}\cref{alg:getbit-main,alg:exactsubtract1-main}@).
extractBit :: Float -> Int -> Bit

-- Generate the next random bit from the binary coded distribution.
-- Returns the generated bit and the updated number of calls to randBit.
generateNextBit :: (BinaryCodedDist) -> BinaryString -> Int -> (Bit, Int)
generateNextBit p b l = do
  let pb0 = p (0:b)
  let pb1 = p (1:b)
  let bit0 = extractBit pb0 l
  let bit1 = extractBit pb1 l
  case (bit0, bit1) of
    (1, 0)    -> (0, l)
    (0, 1)    -> (1, l)
    otherwise -> do
      loop l
      where loop j = do
              let x = randBit
              let j' = j + 1
              let bit0 = extractBit pb0 j'
              let bit1 = extractBit pb1 j'
              if      x == 0 && bit0 == 1 then (0, j')
              else if x == 1 && bit1 == 1 then (1, j')
              else    loop j'

-- Overall recursive function.
generate :: (BinaryCodedDist) -> BinaryString
generate p = generate_ [] 0
    where
        generate_ :: BinaryString -> Int -> BinaryString
        generate_ b l =
            let (x, l') = generateNextBit p b l
@\label{lst:sampler-opt-haskell-recurse}@            in x : (generate_ (x : b) (l+l'))
\end{lstlisting}
\end{listing}

\clearpage

\section{Deferred Results in \S\ref{sec:floating}}
\label{appx:floating}
%!TEX root=./paper.tex

\FloatingPointMonotonic*
\begin{proof}
By the definition of $<_{\floatEm}$, it suffices to show the following:
for all $b, b' \in \set{0,1}^{1+E+m}$ such that $b \neq b'$ and $\gamma_{\floatEm}(\phi_{\floatEm}(\set{y, y'})) \neq \set{0}, \set{\bot}$,
\begin{align}
  \label{eq:fp-monotone-claim}
  b <_{\rm dict} b' \iff \gamma_{\floatEm}(\phi_{\floatEm}(b)) <_{\realext} \gamma_{\floatEm}(\phi_{\floatEm}(b')).
\end{align}
Recall that $\phi_{\mathbb{M}_{E+m}}$ and $\phi_{\floatEm}$ are defined as follows:
for $b_0 \ldots b_{E+m} \in \set{0,1}^{1+E+m}$,
\begin{align}
  \phi_{\mathbb{M}_{E+m}}(b_0 \ldots b_{E+m})
  &=
  \begin{cases}
    1 \bar{b}_1 \ldots \bar{b}_{E+m}
    & \text{if } b_0 = 0
    \\
    0 b_1 \ldots b_{E+m}
    & \text{if } b_0 = 1,
  \end{cases}
  \\
  \phi_{\floatEm}(b_0 \ldots b_{E+m})
  &=
  \begin{cases}
    \phi_{\mathbb{M}_{E+m}}\big((b_0 \ldots b_{E+m})_2 + (2^m -1)\big)
    & \text{if } b_0 \ldots b_{E+m} \leq_{\rm dict} 11^E 0^m
    \\
    b_0 \ldots b_{E+m}
    & \text{if } b_0 \ldots b_{E+m} >_{\rm dict} 11^E 0^m.
  \end{cases}
\end{align}
Using this mapping, we sequentially compute
two bit strings $(b)_2 + (2^m -1), \phi_{\floatEm}(b) \in \set{0,1}^{1+E+m}$
and one extended real $\gamma_{\floatEm}(\phi_{\floatEm}(b)) \in \realext$
for each $b \in \set{0,1}^{1+E+m}$:
\newcommand{\myvdots}{$\smash{\raisebox{-.3ex}{$\vdots$}}$\vphantom{$1$}}
\begin{center}
\vspace{3pt}
\begin{tabular}[c]{@{}c@{\,}|@{\,}c@{\,}|@{\,}c@{}}
  \multicolumn{3}{c}{\smash{$b$}}
  \\[6pt]
  $b_0$ & $b_1 \ldots b_E$ & $b_{E+1} \ldots b_{E+m}$
  \\ \hline
  $0$ & $00 \ldots 00$ & $00 \ldots 000$
  \\[3pt]
  \myvdots & \myvdots & \myvdots
  \\[3pt]
  $0$ & $11 \ldots 11$ & $00 \ldots 000$
  \\ \hline
  $0$ & $11 \ldots 11$ & $00 \ldots 001$
  \\[3pt]
  \myvdots & \myvdots & \myvdots
  \\[3pt]
  $1$ & $11 \ldots 10$ & $00 \ldots 001$
  \\ \hline
  $1$ & $11 \ldots 10$ & $00 \ldots 010$
  \\[3pt]
  \myvdots & \myvdots & \myvdots
  \\[3pt]
  $1$ & $11 \ldots 11$ & $00 \ldots 000$
  \\ \hline \hline
  $1$ & $11 \ldots 11$ & $00 \ldots 001$
  \\[3pt]
  \myvdots & \myvdots & \myvdots
  \\[3pt]
  $1$ & $11 \ldots 11$ & $11 \ldots 111$
\end{tabular}
\quad
\begin{tabular}[c]{@{}c@{\,}|@{\,}c@{\,}|@{\,}c@{}}
  \multicolumn{3}{c}{\smash{$(b)_2 + (2^m-1)$}}
  \\[6pt]
  $b'_0$ & $b'_1 \ldots b'_E$ & $b'_{E+1} \ldots b'_{E+m}$
  \\ \hline
  $0$ & $00 \ldots 00$ & $11 \ldots 111$
  \\[3pt]
  \myvdots & \myvdots & \myvdots
  \\[3pt]
  $0$ & $11 \ldots 11$ & $11 \ldots 111$
  \\ \hline
  $1$ & $00 \ldots 00$ & $00 \ldots 000$
  \\[3pt]
  \myvdots & \myvdots & \myvdots
  \\[3pt]
  $1$ & $11 \ldots 11$ & $00 \ldots 000$
  \\ \hline
  $1$ & $11 \ldots 11$ & $00 \ldots 001$
  \\[3pt]
  \myvdots & \myvdots & \myvdots
  \\[3pt]
  $1$ & $11 \ldots 11$ & $11 \ldots 111$
  \\\hline \hline
  \multicolumn{3}{c}{\vphantom{$1$}}
  \\[3pt]
  \multicolumn{3}{c}{\vphantom{\myvdots}unimportant}
  \\[3pt]
  \multicolumn{3}{c}{\vphantom{$1$}}
\end{tabular}
\quad
\begin{tabular}[c]{@{}c@{\,}|@{\,}c@{\,}|@{\,}c@{}}
  \multicolumn{3}{c}{\smash{$\phi_{\floatEm}(b)$}}
  \\[6pt]
  $s$ & $e_E \ldots e_1$ & $f_1 \ldots f_m$
  \\ \hline
  $1$ & $11 \ldots 11$ & $00 \ldots 000$
  \\[3pt]
  \myvdots & \myvdots & \myvdots
  \\[3pt]
  $1$ & $00 \ldots 00$ & $00 \ldots 000$
  \\ \hline
  $0$ & $00 \ldots 00$ & $00 \ldots 000$
  \\[3pt]
  \myvdots & \myvdots & \myvdots
  \\[3pt]
  $0$ & $11 \ldots 11$ & $00 \ldots 000$
  \\ \hline
  $0$ & $11 \ldots 11$ & $00 \ldots 001$
  \\[3pt]
  \myvdots & \myvdots & \myvdots
  \\[3pt]
  $0$ & $11 \ldots 11$ & $11 \ldots 111$
  \\ \hline \hline
  $1$ & $11 \ldots 11$ & $00 \ldots 001$
  \\[3pt]
  \myvdots & \myvdots & \myvdots
  \\[3pt]
  $1$ & $11 \ldots 11$ & $11 \ldots 111$
\end{tabular}
\quad
\begin{tabular}[c]{@{}c@{}}
  \multicolumn{1}{c}{\smash{$\gamma_{\floatEm}(\phi_{\floatEm}(b))$}}
  \\[6pt]
  $r \in \realext$
  \\ \hline
  $-\infty$
  \\[3pt]
  \myvdots
  \\[3pt]
  $0$
  \\ \hline
  $0$
  \\[3pt]
  \myvdots
  \\[3pt]
  $+\infty$
  \\ \hline
  $\bot$
  \\[3pt]
  \myvdots
  \\[3pt]
  $\bot$
  \\ \hline \hline
  $\bot$
  \\[3pt]
  \myvdots
  \\[3pt]
  $\bot$
\end{tabular}
\vspace{6pt}
\end{center}
Using this calculation and the definition of $\gamma_{\floatEm}$ and $<_{\realext}$,
one can check that \cref{eq:fp-monotone-claim} indeed holds.
\end{proof}

\CdaPcbd*
\begin{proof}%[Proof of \cref{prop:cda-pcbd}]
We first argue that $p_F$ is a binary-coded distribution.
If $\abs{b} = 0$, then \cref{eq:trefle-1} gives
\begin{align}
p_F(\varepsilon) = F(1^n) -_\real F((0^n)^-) = 1 - 0 = 1.
\end{align}
If $1 \le \abs{b} \le n$, then \cref{eq:trefle-1} again gives
\begin{align}
p_F(b0) + p_F(b1)
  &=
    F(b01^m)
    - F((b00^m)^{-})
    + F(b11^m)
    - F(\underbrace{(b10^m)^{-}}_{=b01^m})
    \label{eq:decadally-1}
    \\
  &= F(b11^m) - F((b0 0^m)^{-})
  = F(b1^{m+1}) - F((b0^{m+1})^{-})
  = p_F(b),
  \label{eq:decadally-2}
\end{align}
where $m \defas n - \abs{b} - 1$.
Next consider $p_F(bb')$ where $\abs{b} = n$ and $b' \in \bool^+$.
If $b' = 0\ldots{0}$,
\cref{eq:trefle-2} gives
$p(bb'0) + p(bb'1) = p(b) + 0 = p(b) = p(bb')$.
Otherwise $b' \ne 0\ldots{0}$
so \cref{eq:trefle-2} gives $p(bb'0) + p(bb'1) = 0 + 0 = 0 = p(bb')$.

Finally, we prove that $P_F(b) = p_F(b)$ for all $b \in \bool^n$,
where $P_F$.
From \cref{eq:trefle-2}, if $b = 0^n$ then $p_F(b) = F(b) = P_F(b)$.
Otherwise
$p_F(b) = F(b) - F(\mathrm{pred}_\mathrm{dict}(b)) = P_F(b)$,
which follows from $\mathrm{pred}_\mathrm{dict} = \mathrm{pred}_{\mathbb{U}_n}$.
\end{proof}

\CdaPcbdUint*
\begin{proof}
As $\phi_{\mathbb{U}_n} = \mathrm{id}$, for the minimal element:
$\widetilde{F}(\phi_{\mathbb{U}_n}(0^n)) = \widetilde{F}(0^n) = F(\phi_S(0^n)) \ge 0$.
For the maximal element:
$\widetilde{F}(\phi_{\mathbb{U}_n}(1^n))= \widetilde{F}(1^n) = F(\phi_S(1^n)) = 1$.
Since $F$ is monotonically non-decreasing and $\phi$ monotonically
increasing, so is their composition $\widetilde{F} \equiv F \circ \phi_S$.
That is, if $x <_{\mathbb{U}_n} x'$, then \cref{eq:snobbing-1} implies that
$\phi_S(x) <_{\floatEm} \phi_S(x')$, so
$\widetilde{F}(x) = F(\phi_S(x)) \le_{\floatEm} F(\phi_S(x')) = \widetilde{F}(x')$.
\end{proof}

\ExactSubtractOneProp*

\begin{proof}%[Proof of \cref{prop:exactsubtract1}]
  \setcounter{MaxMatrixCols}{20}
  \setlength{\arraycolsep}{3pt}
  We prove the claims for $x < 1$ and $x = 1$ separately.

  First, assume $x<1$. Then, $(\hat{e}, f)$ and $(\hat{e}', f')$
  computed in \crefrange{algline:exactsubtract1-main-ef-st}{algline:exactsubtract1-main-ef-ed} of \Cref{alg:exactsubtract1-main}
  satisfy
  \begin{align}
    \label{eq:exactsubtract1-main-f}
    \hat{e}, \hat{e}' &\leq -1,
    &
    f &= (f_0 \ldots f_m)_2,
    &
    f' &= (f'_0 \ldots f'_m)_2,
    &
    x &= 2^{\hat{e}} \cdot (f / 2^m),
    &
    x' &= 2^{\hat{e}'} \cdot (f' / 2^m),
  \end{align}
  where $f_0 \defas \mathbf{1}[e > 0]$ and $f'_0 \defas \mathbf{1}[e' > 0]$.
  Here, $\leq$ is by $x, x' \in [0, 1)$,
  the first two $=$ are by the definition of $f_0$ and $f'_0$,
  and the last two $=$ are by the definition of $\floatEm$ and $x, x' \in [0,\infty) \cap \floatEm$.
  We note that $f_0 = 1$ if $x$ is a normal float, and $f_0 = 0$ if $x$ is a subnormal float;
  the same hold for $f'_0, x'$.

  Using the previous observation, we show \Cref{eq:exactsubtract1-main} for $x<1$ by case analysis on $\hat{e} - \hat{e}'$.
  \begin{enumerate}[label=Case~\arabic*.,wide,leftmargin=*]
  \item ($\hat{e}- \hat{e}' \leq m+1$).
    In this case, we have
    \begin{align}
      \nonumber
      \\*[10pt]
      \label{eq:exactsubtract1-main-case1}
      \!\!\!\!
      \begin{NiceMatrix}[l]
        ~
        & \Block[r]{}{x} & =\;\;\; & 0.
        & 0 & \ldots & 0
        & \Block[c,borders={left,right,top,bottom}]{1-6}{} f_0 & \ldots & \ldots & \ldots & \ldots & f_m % ~~~~f_m~~~~
        & 0 & \ldots & 0
        \\
        -
        & \Block[r]{}{x'} & =\;\;\; & 0.
        & 0 & \ldots & 0
        & 0 & \ldots & 0
        & \Block[c,borders={left,right,top,bottom}]{1-3}{} f'_0 & \ldots & f'_{m - (\hat{e}-\hat{e}')}
        & \Block[c,borders={left,right,top,bottom}]{1-3}{} f'_{m - (\hat{e}-\hat{e}') +1} & \ldots & f'_m
        \\
        \\[23pt] \cline{2-16}
        \\[-20pt]
        \\
        ~
        & \Block[r]{2-1}{x-x'} & \Block{2-1}{=\;\;\;} & 0.
        & 0 & \ldots & 0
        & \Block[c,draw=black]{1-6}{f-f'_\mathrm{hi}} & & & & &
        & \Block[c,draw=black]{1-3}{0} & &
        & \Block[l]{}{\;\;\text{if } f'_\mathrm{lo} = 0}
        \\
        & & & 0.
        & 0 & \ldots & 0
        & \Block[c,draw=black]{1-6}{f-f'_\mathrm{hi}-1} & & & & &
        & \Block[c,draw=black]{1-3}{2^{\hat{e}-\hat{e}'} - f'_\mathrm{lo}} & &
        & \Block[l]{}{\;\;\text{if } f'_\mathrm{lo} > 0}
        \CodeAfter
        \SubMatrix\{{6-4}{7-4}.
        \OverBrace [yshift=1pt]{1-5} {1-7} {\scriptstyle -\hat{e}-1 \text{ bits}}
        \OverBrace [yshift=1pt]{1-8} {1-13}{\substack{{\displaystyle f \in \mathbb{Z}} \\[1pt] m+1 \text{ bits}}}
        \UnderBrace[yshift=1pt]{2-5} {2-10}{\scriptstyle -\hat{e}'-1 \text{ bits}}
        \UnderBrace[yshift=1pt]{2-11}{2-13}{\substack{(m+1)-(\hat{e}-\hat{e}') \text{ bits} \\[1pt] {\displaystyle f'_\textrm{hi} \in \mathbb{Z}}}}
        \UnderBrace[yshift=1pt]{2-14}{2-16}{\substack{\hat{e}-\hat{e}' \text{ bits} \\[1pt] {\displaystyle f'_\textrm{lo} \in \mathbb{Z}}}}
        \UnderBrace[yshift=1pt]{7-5} {7-7} {\scriptstyle -\hat{e}-1 \text{ bits}}
        \UnderBrace[yshift=1pt]{7-8} {7-13}{\scriptstyle m+1 \text{ bits}}
        \UnderBrace[yshift=1pt]{7-14}{7-16}{\scriptstyle \hat{e}-\hat{e}' \text{ bits}}
      \end{NiceMatrix}
      \!\!\!\!
      \\[4pt] \nonumber
    \end{align}
    Here, the equalities on $x$, $x'$, $f$ are by \Cref{eq:exactsubtract1-main-f},
    and the equalities on $f'_\mathrm{hi}$, $f'_\mathrm{lo}$ are by
    \crefrange{algline:exactsubtract1-main-fpdecomp-st}{algline:exactsubtract1-main-fpdecomp-ed} of \cref{alg:exactsubtract1-main}
    and the following: $0 \leq \hat{e} - \hat{e}' \leq m + 1$ (by $x \geq x'$) and $\hat{e} - \hat{e}' \leq E + m$ (by $E \geq 1$).
    Further, the equality on $x-x'$ is by earlier equalities and the following:
    if $f'_\mathrm{lo} > 0$, then
    $f - f'_\mathrm{hi} \geq 1$ (since $f_0 = 1$ and $\hat{e} - \hat{e}' \geq 1$ must hold) % ; otherwise, $f'_\mathrm{lo} > 0$ cannot hold)
    and $0 \leq 2^{\hat{e} - \hat{e}'} - f'_\mathrm{lo} < 2^{\hat{e} - \hat{e}'}$.

    Based on \cref{eq:exactsubtract1-main-case1}, we can check that $\beta$ computed in
    \crefrange{algline:exactsubtract1-main-final-st}{algline:exactsubtract1-main-final-ed} of \cref{alg:exactsubtract1-main}
    satisfies the following:
    $(n_1, n_\mathrm{hi}, n_2, n_\mathrm{lo}) = (-\hat{e}-1, m+1, 0, \hat{e}-\hat{e}')$
    are the numbers of bits shown in the last line of \cref{eq:exactsubtract1-main-case1};
    $(b_1, b_2) = (0, \mathbf{1}[f'_\mathrm{lo}>0])$ are the bits of $x-x'$ in the $n_1$ and $n_2$ parts;
    and $(g_\mathrm{hi}, g_\mathrm{lo})$ are the values of $x-x'$ in the $n_\mathrm{hi}$ and $n_\mathrm{lo}$ parts
    (i.e., the boxed values in the last line of \cref{eq:exactsubtract1-main-case1}).
    Hence, the last line of \cref{eq:exactsubtract1-main-case1} implies \cref{eq:exactsubtract1-main}, as desired.

  \item ($\hat{e}- \hat{e}' > m+1$).
    In this case, we have
    \begin{align}
      \nonumber
      \\[10pt]
      \label{eq:exactsubtract1-main-case2}
      \begin{NiceMatrix}[l]
        ~
        & \Block[r]{}{x} & =\;\;\; & 0.
        & 0 & \ldots & 0
        & \Block[c,borders={left,right,top,bottom}]{1-3}{} f_0 \;\; & \ldots\;\; & f_m \;\;
        & \;\; 0 \;\; & \; \ldots \; & \;\; 0 \;\;
        & 0 & \ldots & 0
        \\
        -
        & \Block[r]{}{x'} & =\;\;\; & 0.
        & 0 & \ldots & 0
        & 0 & \ldots & 0
        & \;\; 0 \;\; & \; \ldots \; & \;\; 0 \;\;
        & \Block[c,borders={left,right,top,bottom}]{1-3}{} f'_0 \;\; & \ldots\;\; & f'_m\;\;
        \\
        \\[23pt] \cline{2-16}
        \\[-20pt]
        \\
        ~
        & \Block[r]{2-1}{x-x'} & \Block{2-1}{=\;\;\;} & 0.
        & 0 & \ldots & 0
        & \Block[c,draw=black]{1-3}{f} & &
        & \;\; 0 \;\; & \;\ldots\; & \;\; 0 \;\;
        & \Block[c,draw=black]{1-3}{0} & &
        & \Block[l]{}{\;\;\text{if } f'_\mathrm{lo} = 0}
        \\
        ~
        & & & 0.
        & 0 & \ldots & 0
        & \Block[c,draw=black]{1-3}{f - 1} & &
        & \;\; 1 \;\; & \;\ldots\; & \;\; 1 \;\;
        & \Block[c,draw=black]{1-3}{2^{m+1} - f'_\mathrm{lo}} & &
        & \Block[l]{}{\;\;\text{if } f'_\mathrm{lo} > 0}
        \CodeAfter
        \SubMatrix\{{6-4}{7-4}.
        \OverBrace [yshift=1pt]{1-5} {1-7} {\scriptstyle -\hat{e}-1 \text{ bits}}
        \OverBrace [yshift=1pt]{1-8} {1-10}{\substack{{\displaystyle f \in \mathbb{Z}} \\[1pt] m+1 \text{ bits}}}
        \UnderBrace[yshift=1pt]{2-5} {2-13}{\scriptstyle -\hat{e}'-1 \text{ bits}}
        \UnderBrace[yshift=1pt]{2-14}{2-16}{\substack{m+1 \text{ bits} \\[1pt] {\displaystyle f'_\textrm{lo} \in \mathbb{Z}}}}
        \UnderBrace[yshift=1pt]{7-5} {7-7} {\scriptstyle -\hat{e}-1 \text{ bits}}
        \UnderBrace[yshift=1pt]{7-8} {7-10}{\scriptstyle m+1 \text{ bits}}
        \UnderBrace[yshift=1pt]{7-11}{7-13}{\scriptstyle (\hat{e}-\hat{e}')-(m+1) \text{ bits}}
        \UnderBrace[yshift=1pt]{7-14}{7-16}{\scriptstyle m+1 \text{ bits}}
      \end{NiceMatrix}
      \\[4pt] \nonumber
    \end{align}
    Here, the equalities on $x$, $x'$, $f$ is by the same argument in the previous case,
    the equality on $f'_\mathrm{lo}$ is by \cref{algline:exactsubtract1-main-fpdecomp-ed} of \Cref{alg:exactsubtract1-main}
    and $\hat{e} - \hat{e}' > m + 1$,
    and the equality on $x-x'$ is by earlier equalities.
    We note that $f'_\mathrm{hi}$ computed in \cref{algline:exactsubtract1-main-fpdecomp-st} of \Cref{alg:exactsubtract1-main}
    satisfies  $f'_\mathrm{hi} = 0$, because $\hat{e} - \hat{e}' \geq m + 1$ and $E+m \geq m+1$.

    Based on \cref{eq:exactsubtract1-main-case2}, we can check that $\beta$ computed in
    \crefrange{algline:exactsubtract1-main-final-st}{algline:exactsubtract1-main-final-ed} of \Cref{alg:exactsubtract1-main}
    satisfies the following:
    $(n_1, n_\mathrm{hi}, n_2, n_\mathrm{lo}) = (-\hat{e}-1, m+1, (\hat{e}-\hat{e}')-(m+1), m+1)$
    are the numbers of bits shown in the last line of \cref{eq:exactsubtract1-main-case2};
    $(b_1, b_2) = (0, \mathbf{1}[f'_\mathrm{lo}>0])$ are the bits of $x-x'$ in the $n_1$ and $n_2$ parts;
    and $(g_\mathrm{hi}, g_\mathrm{lo})$ are the values of $x-x'$ in the $n_\mathrm{hi}$ and $n_\mathrm{lo}$ parts
    (since $f'_\mathrm{hi} = 0$).
    Hence, the last line of \cref{eq:exactsubtract1-main-case2} implies \cref{eq:exactsubtract1-main}, as desired.
  \end{enumerate}

  We now show the remaining claims for $x < 1$:
  (i) $b'$ is the $\ell$-th digit of $x-x'$ in binary expansion, and
  (ii) all intermediate values appearing in
  \hyperref[alg:exactsubtract1-main]{\ExactSubtractOne} and \hyperref[alg:getbit-main]{\GetBit}
  are representable as $(m+1)$-bit (signed or unsigned) integers.
  The claim (i) follows immediately from \cref{eq:exactsubtract1-main} and the definition of $\hyperref[alg:getbit-main]{\GetBit}(\beta,\ell)$.
  The claim (ii) holds as follows: we have
  \begin{alignat}{3}
    \label{eq:exactsubtract1-main-range-e}
    e, e', \hat{e}, \hat{e}'
    &  \;\in\; [-2^{E-1}+2, 2^E-1]
    %% && \;\subseteq\; [-2^{m-1}, 2^m - 1],
    && \;\subseteq\; [-2^{E-1}, 2^E - 1],
    \\
    \label{eq:exactsubtract1-main-range-f}
    f, f', f'_\mathrm{hi}, f'_\mathrm{lo}, g_\mathrm{hi}, g_\mathrm{lo}
    &  \;\in\; [0, 2^{m+1}-1]
    %% && \;=\; [0, 2^{m+1} - 1],
    && \;\subseteq\; [0, 2^{E+m} - 1],
    \\
    \label{eq:exactsubtract1-main-range-n}
    n_1 + n_\mathrm{hi} + n_2 + n_\mathrm{lo}
    &  \;\in\; [0, -\hat{e}'+m]
    %% && \;\subseteq\; [0, 2^{m+1} - 1],
    && \;\subseteq\; [0, 2^{E+m} - 1],
  \end{alignat}
  which implies that all the above values are representable as $(1+E+m)$-bit signed integers. Here,
  \cref{eq:exactsubtract1-main-range-f} is by $E \geq 1$ and \cref{eq:exactsubtract1-main-range-n} is by
  $-\hat{e}' + m \leq (2^{E-1}-2) + m \leq 2^E + 2^m - 2 \leq 2^{E+m}-1$.

  Lastly, we consider the remaining case: $x = 1$.
  In this case, \cref{eq:exactsubtract1-main-case1,eq:exactsubtract1-main-case2} still hold except that $-\hat{e}-1 = -1$ is now less than $0$;
  and $g_\mathrm{hi} < 2^m$ holds since $x-x' < 1$ (by assumption).
  From these,
  \begin{align}
    x -_\real x'
    &= \Big(0.\,
        \underset{n_\mathrm{hi} \text{ bits}}{\boxed{\,\;\; g_\mathrm{hi} \vphantom{b_1 g_\mathrm{hi}} \,\;\;}}
        \underset{n_2 \text{ bits}}{\boxed{b_2 \ldots b_2       \vphantom{b_1 g_\mathrm{hi}}       }}
        \underset{n_\mathrm{lo} \text{ bits}}{\boxed{\,\;\; g_\mathrm{lo} \vphantom{b_1 g_\mathrm{hi}} \,\;\;}}
        \,
        \Big)_2
  \end{align}
  Since \crefrange{algline:exactsubtract1-main-final-st}{algline:exactsubtract1-main-final-ed} of \cref{alg:exactsubtract1-main}
  compute $(n_1, n_\mathrm{hi}) = (0, m)$, the output of \cref{alg:exactsubtract1-main} corresponds to the above equation.
  This implies that all the claims still hold for $x=1$. % (by the same argument used for $x<1$).
\end{proof}

\begin{proposition}[name=,restate=EntropyExtend]
\label{proposition:entropy-extend}
If $p \defas (p_1, \dots, p_{n-1}, p_n)$ and $p' \defas (p_1, \ldots, p_{n-1}, p'_n, p'_{n+1})$
are discrete probability distributions with $p_n = p'_n + p'_{n+1}$ and $p'_n, p'_{n+1} > 0$,
then $H(p') > H(p)$.
\end{proposition}

\begin{proof}
\begin{align}
&H(p') - H(p) \notag\\
&= -p'_{n}\log_2(p'_{n}) -p'_{n+1}\log_2(p'_{n+1}) + p_n\log_2(p_n) \\
&= -p'_{n}\log_2(p'_{n}) -p'_{n+1}\log_2(p'_{n+1}) + (p'_n+p'_{n+1})\log_2(p'_{n}+p'_{n+1}) \\
&= p'_n[\log_2(p'_{n}+p'_{n+1})-\log_2(p'_n)] + p'_{n+1}[\log_2(p'_{n}+p'_{n+1})-\log_2(p'_{n+1})]
> 0.
\end{align}
\end{proof}

We next establish \cref{corollary:cost-sampler-opt-impl}, whose proof
rests on \cref{prop:cda-max-entropy,prop:max-entropy-discrete}.

\begin{theorem}[name=,restate=CdaMaxEntropy]
\label{prop:cda-max-entropy}
Let $X \subset \realext$ and $Y \subset [0,1]$ be finite sets with
$|Y| \leq |X|+1$ and $\set{0,1} \subset Y$.
Define $\mathcal{F}(X, Y) \subset \realext \to [0,1]$ to be the set of CDFs
with atoms in $X$ and cumulative probabilities in $Y$.
Letting $H(F)$ denote the Shannon entropy of $F$, we have
\begin{align}
F \in {\rm argmax}_{F' \in  \mathcal{F}(X,Y)}\, H(F') \Longleftrightarrow F(X) \cup \set{0} = Y.
\rlap{\qquad\qquad\qedhere}
\end{align}
\end{theorem}

\begin{proof}
Let $X = \set{x_1 <_{\realext} \cdots <_{\realext} x_{|X|}}$ and
$Y = \set{0 = y_1 < \cdots < y_{|Y|} = 1}$.
The claim is trivial for $\abs{Y} = 2$.
Suppose $\abs{Y} \ge 3$.
Let $F \in \mathcal{F}(X,Y)$
and $p = (p_1, \ldots, p_{\abs{X}}) \in \real^{\abs{X}}$ be the discrete distribution corresponding to $F$,
i.e., $p_{i} \defas F(x_i)-F(x_{i-1})$ with the convention that $F(x_0) \defas 0$.

$(\Longrightarrow)$
Suppose $F(X) \cup \set{0} \neq Y$.
Then, there exists $y^* \in Y \setminus (F(X) \cup \set{0})$.
Since $y^* \neq 0$ and $y^* \neq 1$ (because $F(x_{\abs{X}}) = 1$),
there exists $1 < i^* < \abs{X}$ such that $F(x_{i^*}) < y^* < F(x_{i^*+1})$.
Let $\widetilde{p} \in \real^{\abs{X}+1}$ be a new discrete distribution defined by
\begin{align}
  \widetilde{p} = \big(p_1, \ldots, p_{i^*}, y^* - F(x_{i^*}), F(x_{i^*+1}) - y^*, p_{i^*+2}, \ldots, p_{\abs{X}}\big).
\end{align}
Then, $H(p) < H(\widetilde{p})$ by \cref{proposition:entropy-extend}.
Further, $\widetilde{p}$ satisfies two properties:
all the prefix sums of $\widetilde{p}$ are in $Y$, and $\widetilde{p}_{j^*} = 0$ for some $j^*$.
The first property holds because each of the prefix sum of $\widetilde{p}$ is either $y^*$ or $F(x_i)$ for some $i$.
The second property holds as follows: if $F(x_{j})=0$ for some $j$, then $p_j=0$ with $j \neq i^*+1$;
otherwise, $\abs{F(X)} = \abs{F(X) \setminus \set{0}} < \abs{Y \setminus \set{0}} = \abs{Y}-1 \leq \abs{X}$,
so $F(x_{j}) = F(x_{j-1})$ for some $j>1$, implying that $p_j=0$ with $j \neq i^*+1$.
By the two properties, there exists $\widetilde{F} \in \mathcal{F}(X,Y)$
corresponding to $\widetilde{p}$ (with $\widetilde{p}_{j^*} = 0$ excluded),
and $H(p) < H(\widetilde{p})$ implies the desired conclusion:
\begin{align}
  H(F) = H(p) < H(\widetilde{p}) = H(\widetilde{F}) \leq \max_{F' \in \mathcal{F}(X,Y)} H(F').
\end{align}

$(\Longleftarrow)$
This direction is immediate from the previous direction and the fact that
$F(X) \cup \set{0} = F'(X) \cup \set{0}$ implies $H(F) = H(F')$ for all $F,F' \in \mathcal{F}(X,Y)$.
\end{proof}

\begin{proposition}
\label{prop:max-entropy-discrete}
Let $\bfmt$ and $\floatEm$ be binary number formats with
$\abs{\floatEm \cap [0,1]} \leq \abs{\bfmt} + 1$. % \realext_\bfmt
Let $\mathcal{F}$ be the set of CDFs
with atoms in $\bfmt$ and cumulative probabilities in $\floatEm$.
Then,
$\max_{F' \in \mathcal{F}} H(F') = m + 2 - 2^{-2^{E-1} + 3}$,
\end{proposition}

\begin{proof}
Let $X = \bfmt$ and $Y = \floatEm \cap [0,1]$.
Since they satisfy all the conditions of \Cref{prop:cda-max-entropy},
this theorem implies $\max_{F' \in \mathcal{F}} H(F') = H(F)$,
where $F \in \mathcal{F}$ is a CDF that has atoms in $\bfmt$
with $F(\bfmt) \cup \set{0} = \floatEm \cap [0,1]$.
Hence, it suffices to show $H(F) = m + 2 - 2^{-2^{E-1} + 3}$.
Let $k \defas 2^{E-1}+2$ so that $-k$ is the smallest exponent in $\floatEm$.
The atoms of $F$ have probabilities given by
\begin{align}
\mbox{subnormal binade:} &&2^{-k}/2^m \label{eq:maxent-prob-subnormal} \\
\mbox{normal binades:} && (2^{-e+1}-2^{-e})/2^m = 2^{-e}/2^m && (e=k,k-1,\dots,1). \label{eq:maxent-prob-normal}
\end{align}
As there are $2^m$ equally likely outcomes in each of these binades, the entropy of $F$ is
\begin{align}
H(F)
  &= 2^m\left(2^{-k-m}\log_2(2^{k+m})\right)
    + \sum_{e=1}^{k}\left[2^m \left( 2^{-e-m}\log_2\left(2^{e+m}\right) \right) \right]
  \\
  &= 2^m\left(2^{-k-m}(k+m))\right)
    + \sum_{e=1}^{k}\left[2^m \left( 2^{-e-m}(e+m) \right) \right]
    \label{eq:ky-exact-cost}
  \\
  &= 2^{-k}({k+m})
    + \sum_{e=1}^{k}\left[2^{-e}\left({e+m}\right) \right]
  \\
  &= 2^{-k} (k+m) + \left( (m+2) - 2^{-k} (k+m+2) \right)
  \\
  &= m + 2 - 2^{-2^{E-1} + 3}.
\end{align}
\end{proof}

%!TEX root=./paper.tex

\begin{figure}[t]
\tikzset{level distance=12.5pt, sibling distance=4pt}
\tikzset{every tree node/.style={anchor=north}}
\tikzset{every leaf node/.style={draw,circle,inner sep=1.5pt}}
\tikzstyle{branch}=[shape=coordinate]
\begin{tikzpicture}
\Tree[
  .\node[coordinate, name=cord-0]{};
  [
    % DOTS BINADE
    [.$\dots$
      % SUBNORMAL BINADE
      [.\node[coordinate, name=cord-4]{};
        [ \node[name=x-4]{}; \node{}; ]
        \edge[draw=none];
        [.{$\dots$} \edge[draw=none]; \node[draw=none,name=dots-4]{}; ]
        [ \node[]{}; \node[name=y-4]{}; ]
      ]
      % FIRST NORMAL BINADE
      [.\node[coordinate, name=cord-3]{};
        [ \node[name=x-3]{}; \node{}; ]
        \edge[draw=none];
        [.{$\dots$} \edge[draw=none]; \node[draw=none,name=dots-3]{}; ]
        [ \node[]{}; \node[name=y-3]{}; ]
      ]
    ]
    % NEXT TO LAST BINADE
    [.\node[coordinate, name=cord-2]{};
      [ \node[name=x-2]{}; \node{}; ]
      \edge[draw=none];
      [.{$\dots$} \edge[draw=none]; \node[draw=none,name=dots-2]{}; ]
      [ \node[]{}; \node[name=y-2]{}; ]
    ]
  ]
  % LAST BINADE
  [.\node[coordinate, name=cord-1]{};
    [ \node[name=x-1]{}; \node{}; ]
    \edge[draw=none];
    [.{$\dots$} \edge[draw=none]; \node[draw=none,name=dots-1]{}; ]
    [ \node[]{}; \node[name=y-1]{}; ]
  ]
]

\foreach \i in {1,...,4}{\node[at={(dots-\i)}]{$\dots$};}

\draw[decorate, decoration={mirror,brace,amplitude=5pt,raise=2.5pt}]
  (x-1.west) -- (y-1.east)
  node[pos=0.5,below,yshift=-5pt,font=\scriptsize]{\begin{tabular}{c}$2^m$ outcomes\\ binade $[2^{-1},2^0)$\end{tabular}};

\draw[decorate, decoration={mirror,brace,amplitude=5pt,raise=2.5pt}]
  (x-2.west) -- (y-2.east)
  node[pos=0.5,below,yshift=-5pt,font=\scriptsize]{\begin{tabular}{c}$2^m$ outcomes\\ binade $[2^{-2},2^{-1})$\end{tabular}};

\draw[decorate, decoration={mirror,brace,amplitude=5pt,raise=2.5pt}]
  (x-3.west) -- (y-3.east)
  node[pos=0.5,below,yshift=-5pt,font=\scriptsize]{\begin{tabular}{c}$2^m$ outcomes\\ binade $[2^{-k},2^{-k+1})$\end{tabular}};

\draw[decorate, decoration={mirror,brace,amplitude=5pt,raise=2.5pt}]
  (x-4.west) -- (y-4.east)
  node[pos=0.5,below,yshift=-5pt,font=\scriptsize]{\begin{tabular}{c}$2^m$ outcomes\\ binade $[0,2^{-k})$\end{tabular}};

\draw[latex-latex]
  ([xshift=2pt]y-1.east)
  --
  ([xshift=2pt]$(cord-1 -| y-1.east)$)
  node[pos=0.75,right,font=\scriptsize,inner xsep=2pt]{$m$}
  ;

\draw[latex-latex]
  ([xshift=-2pt]x-2.west)
  --
  ([xshift=-2pt]$(cord-2 -| x-2.west)$)
  node[pos=0.75,left,font=\scriptsize,inner xsep=2pt]{$m$}
  ;

\draw[latex-latex]
  ([xshift=-2pt]x-4.west)
  --
  ([xshift=-2pt]$(cord-4 -| x-4.west)$)
  node[pos=0.75,left,font=\scriptsize,inner xsep=2pt]{$m$}
  ;

\draw[latex-latex]
  ([xshift=-2pt]$(cord-4 -| x-4.west)$)
  --
  ([xshift=-2pt]$(cord-0 -| x-4.west)$)
  node[pos=0.75,left,font=\scriptsize,inner xsep=2pt]{$k$}
  ;

\end{tikzpicture}
\caption{DDG tree for a maximum entropy distribution with cumulative
probabilities in $\floatEm$, which attains the maximum possible expected
entropy cost described in \cref{theorem:sampler-bits}.
Here, $k \defas 2^{E-1}+2$.}
\label{fig:maxent}
\end{figure}
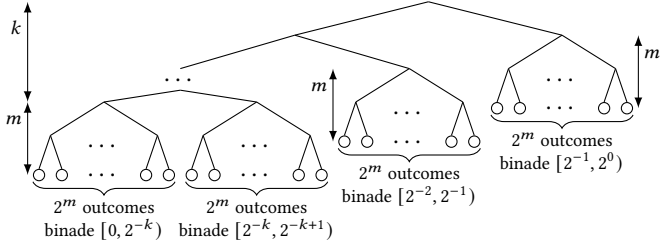

\CostSamplerOptImpl*
\begin{proof}
Let $F \in \mathcal{F}$
be any maximum entropy CDF as in the proof of \cref{prop:max-entropy-discrete}.
Following \cref{eq:maxent-prob-subnormal,eq:maxent-prob-subnormal}, all the
probabilities of outcomes in $F$ are dyadic rationals of the form
$1/2^{e+m}$, where $e \in \set{1, \ldots, k}$ and $k \defas 2^{E-1}+2$.
By \cref{theorem:knuth-yao}, any entropy-optimal DDG tree for
$F$ has precisely one leaf node for each outcome.
Each outcome with probability $1/2^{e+m}$ has a leaf at depth $e+m$
of the tree, and there are $2^m$ such leaves at this depth
(\cref{fig:maxent}).
Therefore, \cref{eq:ky-exact-cost} is the expected entropy cost
of any optimal DDG tree for $F$: the first addend is the
cost for the $2^m$ outcomes with probabilities in the subnormal binade and
the second addend is the sum of costs for
the $2^m$ outcomes with probabilities in each of the $k$ normal binades.

Finally, following an analogous argument to the proof of \cref{prop:cda-max-entropy}
and the observations that \begin{enumerate}[label=(\roman*)]
  \item all the probabilities of distributions in $\mathcal{F}$ are dyadic rationals; and
  \item \cref{eq:ky-exact-cost} characterizes the expected entropy cost,
\end{enumerate}
we conclude that this distribution $F$ attains the largest possible
average entropy cost among all distributions in $\mathcal{F}$.
\end{proof}

\clearpage

\section{Deferred Results in \S\ref{sec:survival}}
\label{appx:survival}
%!TEX root=./paper.tex

\begin{remark}
\label{remark:ecdf-samplers-changes}
We describe the essential changes in \Cref{alg:sampler-naive-ext-impl,alg:sampler-opt-ext-impl},
as compared to Algorithms~\ref{alg:sampler-naive-impl} and \ref{alg:sampler-opt-impl}, respectively.
\begin{itemize}[wide=0pt,leftmargin=*]
\item \Cref{alg:sampler-naive-impl} $\mapsto$ \Cref{alg:sampler-naive-ext-impl}.
The arguments $c_0$, $c_1$ are now
pairs with initial values $c_0 = (0,0)$ and $c_1 = (1,0)$.
\hyperref[algline:sampler-naive-ext-impl-kn]{\ExactRatio} returns $(i, k)$ such that
$i/k = (G^*(b) - G^*(b'))/(G^*(b)-G^*(b''))$.

\item \Cref{alg:sampler-opt-impl} $\mapsto$ \cref{alg:sampler-opt-ext-impl}.
The arguments $c_0$ and $c_1$ are pairs
with initial values $c_0 = (0,0)$ and $c_1 = (1,0)$.
\Crefrange{alg:exactsubtract1-main}{alg:getbit-main}
are replaced with the more general
\cref{alg:exactsubtract1,alg:exactsubtract2,alg:exactsubtractd,alg:getbit}
which exactly subtract $(d',f') \defas G(b')$ from $(d,f) \defas G(b)$ for $b >_\bfmt b'$
by handling four cases:
  \begin{enumerate}[label=Case~\arabic*.,wide,leftmargin=*]
    \item $d = d' = 0$. Apply the existing $\hyperref[alg:exactsubtract1-main]{\ExactSubtractOne}(f,f')$ in \cref{alg:exactsubtract1-main}.
    \item $d = d' = 1$. Return $\hyperref[alg:exactsubtract1-main]{\ExactSubtractOne}(f',f)$, because
    $(1 - f) - (1 - f') = f' - f$.
    \item $d = 1$ and $d' = 0$. We must extract the bits in
    $(1 - f) - f' = 1 - (f+f')$.
    This computation is implemented as \cref{alg:exactsubtract2}
    (\hyperref[alg:exactsubtract2]{\ExactSubtractTwo}),
    whose structure closely mirrors
    \cref{alg:exactsubtract1-main}
    and whose correctness is the subject of~\cref{prop:exactsubtract2}.
    \item $d = 0$ and $d' = 1$.
      This case cannot occur by \labelcref{item:Alcoranist-2} of \cref{proposition:ecda-default-properties}.
      \qedhere
  \end{enumerate}
\end{itemize}
\end{remark}

\EcdaDefault*
\begin{proof}
Assume $S(b^*) < 1/2$.
Recall that $b^*$ is defined by
$b^* \defas \min_{<_\bfmt} \set{ b \in \set{0,1}^n \mid F(b) \geq \mathrm{succ}_{\floatEm}(1/2) }$.
This assumption and definition imply that for all $b \in \set{0,1}^n$,
\begin{align}
  \label{eq:ecda-default-f-half}
  b <_\bfmt b^* & \implies F(b) \leq \mathrm{pred}_{\floatEm}(\mathrm{succ}_{\floatEm}(1/2)) = 1/2,
  \\
  \label{eq:ecda-default-s-half}
  b \geq_\bfmt b^* & \implies S(b) \leq S(b^*) < 1/2,
\end{align}
where the first $\leq$ is by $F$ being into $\floatEm$
and the second $\leq$ is by $S$ being a \WSF{} over $\bfmt$.

We now show that $G$ is a \ECDF{} over $\bfmt$.
By \cref{def:edf}, we need to prove two claims:
$G^*(\phi_\bfmt(1^n)) = 1$; and
$b <_\bfmt b'$ implies $G^*(b) \leq G^*(b')$,
where $G^* : \set{0,1}^n \to [0,1]$ is defined by
\begin{align}
  G^*(b) &\defas (1-d) f + d(1-f)
  & (b \in \set{0,1}^n; (d,f) \defas G(b)).
\end{align}
The first claim holds as follows:
since $G(b) \in \set*{(0, F(b)), (1, S(b))}$ for every $b$,
we have
\begin{align}
  G^*(\phi_\bfmt(1^n)) \in \set*{F(\phi_\bfmt(1^n)), 1-S(\phi_\bfmt(1^n))} = \set{1},
\end{align}
where the $\in$ is by the definition of $G^*$
and the $=$ is by $F(\phi_\bfmt(1^n)) = 1$ and $S(\phi_\bfmt(1^n)) = 0$
(because $F$ and $S$ are finite-precision CDF and SF over $\bfmt$, respectively).
To show the second claim, consider any $b, b' \in \set{0,1}^n$ with $b <_\bfmt b'$.
We show $G^*(b) \leq G^*(b')$ by case analysis on $(b,b')$:
\begin{align}
  \label{eq:ecda-default-inc1}
  b <_\bfmt b' <_\bfmt b^* & \implies G^*(b) = F(b) \leq F(b') = G^*(b'),
  \\
  \label{eq:ecda-default-inc2}
  b <_\bfmt b^* \leq_\bfmt b' & \implies G^*(b) = F(b) \leq 1/2 < 1-S(b') = G^*(b'),
  \\
  \label{eq:ecda-default-inc3}
  b^* \leq_\bfmt b <_\bfmt b' & \implies G^*(b) = 1-S(b) \leq 1-S(b') = G^*(b'),
\end{align}
where the first $\leq$ is by $F$ being a \WCDF{} over $\bfmt$,
the second $\leq$ is by \cref{eq:ecda-default-f-half},
the $<$ is by \cref{eq:ecda-default-s-half},
and the last $\leq$ is by $S$ being a \WSF{} over $\bfmt$.
\end{proof}

\begin{proposition}
\label{proposition:ecda-default-properties}
In the setup of \cref{theorem:ecda-default},
the \ECDF{} $G$ satisfies the following properties:
\begin{enumerate}[label=(\roman*)]
\item\label{item:Alcoranist-3}
  $G$ defines a discrete random variable $X$ over $\realext_\bfmt$,
  for which
  \begin{align}
    \Pr(X \le t) &= (1-d)f + d(1-f)
    && (t \in \realext_\bfmt; (d,f) \defas G(\roundfl{\bfmt,\rtd}(t))).
    \label{eq:Alcoranist-3}
  \end{align}
\item\label{item:Alcoranist-1}
  $\mathrm{Im}(G)
  \subset \set{(0, f) \mid 0 \le f \le 1/2}
     \cup \set{(1, f) \mid 0 \le f <   1/2}.$
\item\label{item:Alcoranist-2}
  $\pi_1(G(b)) < \pi_1(G(b'))$ implies
  $G^*(b) \,{<}\, G^*(b')$ for all $b,b' \,{\in}\, \bool^n$,
  where $\pi_1(d,f) \defas d$.
  \qedhere
\end{enumerate}
\end{proposition}
\begin{proof}
For \labelcref{item:Alcoranist-3}, define $G^\dagger : \realext \to [0,1]$ by $G^\dagger(t) \defas G^*(\roundfl{\bfmt,\rtd}(t))$ as in \cref{proposition:lift-cda}.
Then, $G^\dagger$ is a CDF over $\realext$ due to the two claims we proved about $G^*$.
Further, by the definition of $G^\dagger$ and $G^*$,
the distribution defined by $G^\dagger$ satisfies \cref{eq:Alcoranist-3} and has the support only on $\realext_\bfmt$.
The property \labelcref{item:Alcoranist-1} is immediate from \cref{eq:ecda-default-f-half,eq:ecda-default-s-half}.
For \labelcref{item:Alcoranist-2}, let $b, b' \in \set{0,1}^n$ satisfy $\pi_1(G(b)) < \pi_1(G(b'))$.
Then, $b <_\bfmt b^* \leq_\bfmt b'$ must hold, which implies $G^*(b) < G^*(b')$ by \cref{eq:ecda-default-inc2}.
\end{proof}

\begin{theorem}
  \label{prop:exactsubtract2}
  Suppose that $x, x' \in \floatEm \cap [0,\frac{1}{2}]$ satisfy $0 < x + x' < 1$,  and consider any $\ell \geq 1$.
  Let $\beta = (n_1, n_2, n_\mathrm{hi}, n_\mathrm{lo}, b_1, b_2, g_\mathrm{hi}, g_\mathrm{lo})$
  be the output of $\hyperref[alg:exactsubtract2]{\ExactSubtractTwo}(x, x')$ (\Cref{alg:exactsubtract2}),
  and let $b'$ be the output of $\hyperref[alg:getbit]{\GetBit}(\beta, \ell)$ (\Cref{alg:getbit}).
  Then,
  \begin{align}
    \label{eq:exactsubtract2}
    1 - (x + x')
    &= \Big(0.
    \underbrace{\boxed{       b_1 \ldots b_1 \vphantom{b_1 g_\mathrm{hi}}       }}_{n_1}
    \underbrace{\boxed{\,\;\; g_\mathrm{hi}  \vphantom{b_1 g_\mathrm{hi}} \,\;\;}}_{n_\mathrm{hi}}
    \underbrace{\boxed{       b_2 \ldots b_2 \vphantom{b_1 g_\mathrm{hi}}       }}_{n_2}
    \underbrace{\boxed{\,\;\; g_\mathrm{lo}  \vphantom{b_1 g_\mathrm{hi}} \,\;\;}}_{n_\mathrm{lo}}
    \Big)_2
  \end{align}
  and $b'$ is the $\ell$-th digit of $1 - (x+x')$ in binary expansion.
  Also, all intermediate values appearing in both algorithms
  are representable as $(1+E+m)$-bit (signed or unsigned) integers.
\end{theorem}
\begin{proof}
  \setcounter{MaxMatrixCols}{20}
  \setlength{\arraycolsep}{3pt}
  By \cref{algline:exactsubtract2-x-order} of \cref{alg:exactsubtract2},
  we can assume $x \geq x'$.
  We prove the claims for $x < \frac{1}{2}$ and for $x = \frac{1}{2}$ separately.
  The current proof is similar to the proof of \cref{prop:exactsubtract1}, so we focus mainly on the differences between the two proofs.

  First, assume $x < \frac{1}{2}$.
  Then, $(\hat{e}, f)$ and $(\hat{e}', f')$
  computed in \crefrange{algline:exactsubtract2-ef-st}{algline:exactsubtract2-ef-ed} of \Cref{alg:exactsubtract2}
  satisfy
  \begin{align}
    \label{eq:exactsubtract2-f}
    \hat{e}, \hat{e}' &\leq -2,
    &
    f &= (f_0 \ldots f_m)_2,
    &
    f' &= (f'_0 \ldots f'_m)_2,
    &
    x &= 2^{\hat{e}} \cdot (f / 2^m),
    &
    x' &= 2^{\hat{e}'} \cdot (f' / 2^m),
  \end{align}
  where $f_0 \defas \mathbf{1}[e > 0]$ and $f'_0 \defas \mathbf{1}[e' > 0]$.
  Here, the first inequality is by $x, x' \in [0, \frac{1}{2})$,
  and we have the remaining equalities as in the proof of \cref{prop:exactsubtract1}.

  Using the previous observation, we show \cref{eq:exactsubtract2} for $x < \frac{1}{2}$ by case analysis on $\hat{e} - \hat{e}'$.
  \begin{enumerate}[label=Case~\arabic*.,wide,leftmargin=*]
  \item ($\hat{e}- \hat{e}' \leq m+1$).
    In this case, we obtain
    \begin{align}
      \nonumber
      \\[10pt]
      \label{eq:exactsubtract2-case1}
      \!\!\!\!
      \begin{NiceMatrix}[l]
        ~
        & \Block[r]{}{x} & =\;\;\; & 0.
        & 0 & \ldots & 0 & 0
        & \Block[c,borders={left,right,top,bottom}]{1-6}{} f_0 & \ldots & \ldots & \ldots & \ldots & f_m % ~~~~f_m~~~~
        & 0 & \ldots & 0
        \\
        +
        & \Block[r]{}{x'} & =\;\;\; & 0.
        & 0 & \ldots & 0 & 0
        & 0 & \ldots & 0
        & \Block[c,borders={left,right,top,bottom}]{1-3}{} f'_0 & \ldots & f'_{m - (\hat{e}-\hat{e}')}
        & \Block[c,borders={left,right,top,bottom}]{1-3}{} f'_{m - (\hat{e}-\hat{e}') +1} & \ldots & f'_m
        \\
        \\[23pt] \cline{2-17}
        \\[-20pt]
        \\
        ~
        & x+x' & =\;\;\; & 0.
        & 0 & \ldots & 0
        & \Block[c,draw=black]{1-7}{f + f'_\mathrm{hi}} & & & & & &
        & \Block[c,draw=black]{1-3}{f'_\mathrm{lo}} & &
        \CodeAfter
        \OverBrace [yshift=1pt]{1-5} {1-8} {\scriptstyle -\hat{e}-1 \text{ bits}}
        \OverBrace [yshift=1pt]{1-9} {1-14}{\substack{{\displaystyle f \in \mathbb{Z}} \\[1pt] m+1 \text{ bits}}}
        \UnderBrace[yshift=1pt]{2-5} {2-11}{\scriptstyle -\hat{e}'-1 \text{ bits}}
        \UnderBrace[yshift=1pt]{2-12}{2-14}{\substack{(m+1)-(\hat{e}-\hat{e}') \text{ bits} \\[1pt] {\displaystyle f'_\textrm{hi} \in \mathbb{Z}}}}
        \UnderBrace[yshift=1pt]{2-15}{2-17}{\substack{\hat{e}-\hat{e}' \text{ bits} \\[1pt] {\displaystyle f'_\textrm{lo} \in \mathbb{Z}}}}
        \UnderBrace[yshift=1pt]{6-5} {6-7} {\scriptstyle -\hat{e}-2 \text{ bits}}
        \UnderBrace[yshift=1pt]{6-8} {6-14}{\scriptstyle m+2 \text{ bits}}
        \UnderBrace[yshift=1pt]{6-15}{6-17}{\scriptstyle \hat{e}-\hat{e}' \text{ bits}}
      \end{NiceMatrix}
      \!\!\!\!
      \\[4pt] \nonumber
    \end{align}
    Here, the equalities on $x$, $x'$, $f$, $f'_\mathrm{hi}$, $f'_\mathrm{lo}$ hold as in the proof of \cref{prop:exactsubtract1},
    and the equality on $x+x'$ is by $\hat{e} \leq -2$ and $f + f'_\mathrm{hi} \leq 2 \cdot (2^{m+1}-1) = 2^{m+2} -2 < 2^{m+2}$.
    From this, we obtain
    \begin{align}
      \nonumber
      \\[-17pt]
      \label{eq:exactsubtract2-case1b}
      \!\!\!\!
      \begin{NiceMatrix}[l]
        ~
        & \Block[r]{}{1} & =\;\;\; & 1.
        & 0 & \ldots & 0
        & \phantom{\qquad} 0 \phantom{} & \phantom{\quad} \ldots \phantom{\quad} & \phantom{} 0 \phantom{\qquad}
        & \phantom{\qquad} 0 \phantom{} & \phantom{\quad} \ldots \phantom{\quad} & \phantom{} 0 \phantom{\qquad}
        \\
        -
        & \Block[r]{}{x+x'} & =\;\;\; & 0.
        & 0 & \ldots & 0
        & \Block[c,draw=black]{1-3}{f + f'_\mathrm{hi}} & &
        & \Block[c,draw=black]{1-3}{f'_\mathrm{lo}} & &
        \\
        \\[-6pt] \cline{2-13}
        \\[-20pt]
        \\
        ~
        & \Block[r]{2-1}{1 - (x+x')} & \Block{2-1}{=\;\;\;} & 0.
        & 1 & \ldots & 1
        & \Block[c,draw=black]{1-3}{2^{m+2} - (f + f'_\mathrm{hi})} & &
        & \Block[c,draw=black]{1-3}{0} & &
        & \Block[l]{}{\;\;\text{if } f'_\mathrm{lo} = 0}
        \\
        ~
        & & & 0.
        & 1 & \ldots & 1
        & \Block[c,draw=black]{1-3}{2^{m+2} - (f + f'_\mathrm{hi}) - 1} & &
        & \Block[c,draw=black]{1-3}{2^{\hat{e}-\hat{e}'} - f'_\mathrm{lo}} & &
        & \Block[l]{}{\;\;\text{if } f'_\mathrm{lo} > 0}
        \CodeAfter
        \SubMatrix\{{6-4}{7-4}.
        \UnderBrace [yshift=1pt]{7-5} {7-7} {\scriptstyle -\hat{e}-2 \text{ bits}}
        \UnderBrace [yshift=1pt]{7-8} {7-10}{\scriptstyle m+2 \text{ bits}}
        \UnderBrace [yshift=1pt]{7-11}{7-13}{\scriptstyle \hat{e}-\hat{e}' \text{ bits}}
      \end{NiceMatrix}
      \\[-8pt] \nonumber
    \end{align}
    Here, the last equality is by $1 \leq f + f'_\mathrm{hi} \leq 2^{m+1}-2$;
    we have $f+f'_\mathrm{hi} \geq 1$ because $f + f'_\mathrm{hi} = 0$ implies
    $x=0=x'$ and this contradicts to the assumption $x+x' > 0$.

    Based on \cref{eq:exactsubtract2-case1b}, we can check that $\beta$ computed in
    \crefrange{algline:exactsubtract2-final-st}{algline:exactsubtract2-final-ed} of \cref{alg:exactsubtract2}
    satisfies the following:
    $(n_1, n_\mathrm{hi}, n_2, n_\mathrm{lo}) = (-\hat{e}-2, m+2, 0, \hat{e}-\hat{e}')$
    are the numbers of bits shown in the last line of \cref{eq:exactsubtract2-case1b};
    $(b_1, b_2) = (1, \mathbf{1}[f'_\mathrm{lo}>0])$ are the bits of $1-(x+x')$ in the $n_1$ and $n_2$ parts;
    and $(g_\mathrm{hi}, g_\mathrm{lo})$ are the values of $1-(x+x')$ in the $n_\mathrm{hi}$ and $n_\mathrm{lo}$ parts
    (i.e., the boxed values in the last line of \cref{eq:exactsubtract2-case1b}).
    Hence, the last line of \cref{eq:exactsubtract2-case1b} implies \cref{eq:exactsubtract2}, as desired.

  \item ($\hat{e}- \hat{e}' > m+1$).
    In this case, we obtain
    \begin{align}
      \nonumber
      \\[-17pt]
      \label{eq:exactsubtract2-case2}
      \begin{NiceMatrix}[l]
        ~
        & \Block[r]{}{1} & =\;\;\; & 1.
        & 0 & \ldots & 0 & 0
        & 0 \phantom{\quad} & \ldots \phantom{\quad} & 0 \phantom{\quad}
        & 0 \phantom{\quad} & \ldots \phantom{\quad} & 0 \phantom{\quad}
        & 0 \phantom{\quad} & \ldots \phantom{\quad} & 0 \phantom{\quad}
        \\
        -
        & \Block[r]{}{x + x'} & =\;\;\; & 0.
        & 0 & \ldots & 0 & 0
        & \Block[c,borders={left,right,top,bottom}]{1-3}{} f_0 & \ldots & f_m
        & 0 & \ldots & 0
        & \Block[c,borders={left,right,top,bottom}]{1-3}{} f'_0 & \ldots & f'_m
        \\
        \\[23pt] \cline{2-17}
        \\[-20pt]
        \\
        ~
        & \Block[r]{2-1}{1 - (x+x')} & \Block{2-1}{=\;\;\;} & 0.
        & 1 & \ldots & 1
        & \Block[c,draw=black]{1-4}{2^{m+2} - f} & & &
        & 0 & \ldots & 0
        & \Block[c,draw=black]{1-3}{0} & &
        & \Block[l]{}{\;\;\text{if } f'_\mathrm{lo} = 0}
        \\
        ~
        & & & 0.
        & 1 & \ldots & 1
        & \Block[c,draw=black]{1-4}{2^{m+2} - f - 1} & & &
        & 1 & \ldots & 1
        & \Block[c,draw=black]{1-3}{2^{m+1} - f'_\mathrm{lo}} & &
        & \Block[l]{}{\;\;\text{if } f'_\mathrm{lo} > 0}
        \CodeAfter
        \SubMatrix\{{6-4}{7-4}.
        \UnderBrace[yshift=1pt]{2-5} {2-8} {\scriptstyle -\hat{e}-1 \text{ bits}}
        \UnderBrace[yshift=1pt]{2-9} {2-11}{\substack{m+1 \text{ bits} \\[1pt] {\displaystyle f \in \mathbb{Z}}}}
        \UnderBrace[yshift=1pt]{2-12}{2-14}{\scriptstyle (\hat{e}-\hat{e}')-(m+1) \text{ bits}}
        \UnderBrace[yshift=1pt]{2-15}{2-17}{\substack{m+1 \text{ bits} \\[1pt] {\displaystyle f'_\mathrm{lo} \in \mathbb{Z}}}}
        \UnderBrace[yshift=1pt]{7-5} {7-7} {\scriptstyle -\hat{e}-2 \text{ bits}}
        \UnderBrace[yshift=1pt]{7-8} {7-11}{\scriptstyle m+2 \text{ bits}}
        \UnderBrace[yshift=1pt]{7-12}{7-14}{\scriptstyle (\hat{e}-\hat{e}')-(m+1) \text{ bits}}
        \UnderBrace[yshift=1pt]{7-15}{7-17}{\scriptstyle m+1 \text{ bits}}
      \end{NiceMatrix}
      \\[-8pt] \nonumber
    \end{align}
    Here, the equalities on $x+x'$, $f$, $f'_\mathrm{lo}$ hold as in the proof of \cref{prop:exactsubtract1}.
    Further, the equality on $1-(x+x')$ is by $1 \leq f \leq 2^{m+2}-1$,
    where we have $f \geq 1$ as in the previous case.

    Based on \cref{eq:exactsubtract2-case2}, we can check that $\beta$ computed in
    \crefrange{algline:exactsubtract2-final-st}{algline:exactsubtract2-final-ed} of \cref{alg:exactsubtract2}
    satisfies the following:
    $(n_1, n_\mathrm{hi}, n_2, n_\mathrm{lo}) = (-\hat{e}-2, m+2, (\hat{e}-\hat{e}')-(m+1), m+1)$
    are the numbers of bits shown in the last line of \cref{eq:exactsubtract2-case2};
    $(b_1, b_2) = (1, \mathbf{1}[f'_\mathrm{lo}>0])$ are the bits of $1-(x+x')$ in the $n_1$ and $n_2$ parts;
    and $(g_\mathrm{hi}, g_\mathrm{lo})$ are the values of $1-(x+x')$ in the $n_\mathrm{hi}$ and $n_\mathrm{lo}$ parts
    (since $f'_\mathrm{hi} = 0$).
    Hence, the last line of \cref{eq:exactsubtract2-case2} implies \cref{eq:exactsubtract2}, as desired.
  \end{enumerate}

  We now show the remaining claims for $x < \frac{1}{2}$:
  (i) $b'$ is the $\ell$-th digit of $1-(x+x')$ in binary expansion, and
  (ii) all intermediate values appearing in
  \hyperref[alg:exactsubtract2]{\ExactSubtractTwo} and \hyperref[alg:getbit]{\GetBit}
  are representable as $(m+2)$-bit (signed or unsigned) integers.
  The claim (i) holds as in the proof of \cref{prop:exactsubtract1}.
  To prove the claim (ii), we note that
  \crefrange{eq:exactsubtract1-main-range-e}{eq:exactsubtract1-main-range-n} hold as in the proof of \cref{prop:exactsubtract1},
  except that we have $f, f', f'_\mathrm{hi}, f'_\mathrm{lo}, g_\mathrm{hi}, g_\mathrm{lo} \;\in\; [0, 2^{m+2}-1] \subseteq [0, 2^{1+E+m}-1]$.
  This observation implies that $e, e', \hat{e}, \hat{e}'$ are representable as $(1+E+m)$-bit {\em signed} integers
  and all the other values ($f, f', \ldots$ and $n_1, n_2, \ldots$) are representable as $(1+E+m)$-bit {\em unsigned} integers.

  Lastly, we consider the remaining case: $x = \frac{1}{2}$.
  In this case, \cref{eq:exactsubtract2-case1b,eq:exactsubtract2-case2} still hold except that $-\hat{e}-2 = -1$ is now less than $0$;
  and $g_\mathrm{hi} < 2^{m+1}$ holds since $1-(x+x') < 1$ (by assumption).
  Thus,
  \begin{align}
    1-(x + x')
    &= 0.
    \underbrace{\boxed{\,\;\; g_\mathrm{hi} \vphantom{b_1 g_\mathrm{hi}} \,\;\;}}_{m+1 \text{ bits}}
    \underbrace{\boxed{b_2 \ldots b_2       \vphantom{b_1 g_\mathrm{hi}}       }}_{n_2 \text{ bits}}
    \underbrace{\boxed{\,\;\; g_\mathrm{lo} \vphantom{b_1 g_\mathrm{hi}} \,\;\;}}_{n_\mathrm{lo} \text{ bits}}.
  \end{align}
  Since \crefrange{algline:exactsubtract2-final-st}{algline:exactsubtract2-final-ed} of \cref{alg:exactsubtract2}
  compute $(n_1, n_\mathrm{hi}) = (0, m+1)$, the output of \cref{alg:exactsubtract2} corresponds to the above equation.
  This implies that all the claims still hold for $x=\frac{1}{2}$.
\end{proof}

%!TEX root=./paper.tex
\begin{algorithm}[t]
\captionsetup{hypcap=false}
\caption{Quantile of a Finite-Precision \WCDF{}}
\label{alg:quantile}
\begin{algorithmic}[1]
\Require{%
  \WCDF{} $F: \bool^n \to \floatEm \cap [0,1]$
  over number format $\bfmt = (n, \gamma_\bfmt, \phi_\bfmt)$\\
  Float $q \in \floatEm \cap [0,1]$
}
\Ensure{$\min_{<_\bfmt}\set{b \in \bool^n \mid q \le F(b)}$}
\Function{\Quantile}{$F, q$}
\State $(n, l, h) \gets (1 + E + m, 0, 2^n-1)$
\While{$l \le h$}
  \State $s \gets \floor{(l+h)/2}$
  \State $s' \gets \phi_\bfmt(\gamma^{-1}_{\mathbb{U}_n}(s))$
  \If{$q \le F(s')$}
    \; $h \gets s - 1$; $t \gets s'$
  \Else
    \; $l \gets s + 1$
  \EndIf
\EndWhile
\State \Return $t$
\EndFunction
\end{algorithmic}
\end{algorithm}

%!TEX root=./paper.tex

\begin{figure}[t]
\setlength{\intextsep}{0pt}
\setlength{\textfloatsep}{0pt}
\renewcommand{\hl}[1]{#1}
\begin{minipage}[t]{.495\linewidth}
\begin{algorithm}[H]
\captionsetup{hypcap=false}
\caption{Preprocessing for \hyperref[alg:getbit]{\GetBit}}
\label{alg:exactsubtract1}
\begin{algorithmic}[1]
\Require{\hl{$x, x' \in \floatEm \cap [0,1]$ with $0 < x -_\real x' < 1$}}
\Ensure{$(n_1, n_2, n_\mathrm{hi}, n_\mathrm{lo}, b_1, b_2, g_\mathrm{hi}, g_\mathrm{lo})$}
\Function{\ExactSubtractOne{}}{$x$, $x'$}
  \Statex $~$
  \State $(s\,e_E \ldots e_1\, f_1\dots f_m)_{\floatEm} \gets x$%
    \label{algline:exactsubtract1-ef-st}
  \State $(s\,e'_E \ldots e'_1\, f'_1\dots f'_m)_{\floatEm} \gets x'$
  \State $e \gets (e_E \ldots e_1)_2$
  \State $e' \gets (e'_E \ldots e'_1)_2$
  \State $\hat{e} \gets e - (2^{E-1}-1) + \mathbf{1}[e=0]$
  \State $\hat{e}' \gets e' - (2^{E-1}-1) + \mathbf{1}[e'=0]$
  \State $f \gets (1\,f_1 \ldots f_m)_2 - (\mathbf{1}[e = 0] \ll m)$
  \State $f' \gets (1\,f'_1 \ldots f'_m)_2 - (\mathbf{1}[e' = 0] \ll m)$%
    \label{algline:exactsubtract1-ef-ed}
  \State $f'_\mathrm{hi} \gets f' \gg \min\set{\hat{e} - \hat{e}', E+m}$%
    \label{algline:exactsubtract1-fpdecomp-st}
  \State $f'_\mathrm{lo} \gets f'\, \&\, ((1 \ll \min\set{\hat{e}-\hat{e}', m+1}) - 1)$%
    \label{algline:exactsubtract1-fpdecomp-ed}
  \State \hl{$n_1 \gets -\hat{e}-1 + \mathbf{1}[x=1]$}%
    \label{algline:exactsubtract1-final-st}
  \State $n_2 \gets \max\set{(\hat{e}-\hat{e}') - (m+1), 0}$
  \State \hl{$n_\mathrm{hi} \gets m+1 - \mathbf{1}[x=1]$}
  \State $n_\mathrm{lo} \gets \min\set{\hat{e}-\hat{e}', m+1}$
  \State \hl{$b_1 \gets 0$}
  \State \hl{$b_2 \gets \mathbf{1}[f'_\mathrm{lo} > 0]$}
  \State \hl{$g_\mathrm{hi} \gets f - f'_\mathrm{hi} - b_2$}
  \State $g_\mathrm{lo} \gets (b_2 \ll n_\mathrm{lo}) - f'_\mathrm{lo}$%
    \label{algline:exactsubtract1-final-ed}
  \State \Return $(n_1, n_2, n_\mathrm{hi}, n_\mathrm{lo}, b_1, b_2, g_\mathrm{hi}, g_\mathrm{lo})$
\EndFunction
\end{algorithmic}
\end{algorithm}
\end{minipage}\hfill
\begin{minipage}[t]{.495\linewidth}
\begin{algorithm}[H]
\captionsetup{hypcap=false}
\caption{Preprocessing for \hyperref[alg:getbit]{\GetBit}}
\label{alg:exactsubtract2}
\begin{algorithmic}[1]
\Require{\hl{$x, x' \in \floatEm \cap [0,\frac{1}{2}]$ with $0 < x +_\real x' < 1$}}
\Ensure{$(n_1, n_2, n_\mathrm{hi}, n_\mathrm{lo}, b_1, b_2, g_\mathrm{hi}, g_\mathrm{lo})$}
\Function{\ExactSubtractTwo{}}{$x$, $x'$}
  \State \hl{$(x,x') \gets (x \geq x') \;\textbf{?}\; (x,x') \;\textbf{:}\; (x',x)$}%
    \label{algline:exactsubtract2-x-order}
  \State $(s\,e_E \ldots e_1\, f_1\dots f_m)_{\floatEm} \gets x$%
    \label{algline:exactsubtract2-ef-st}
  \State $(s\,e'_E \ldots e'_1\, f'_1\dots f'_m)_{\floatEm} \gets x'$
  \State $e \gets (e_E \ldots e_1)_2$
  \State $e' \gets (e'_E \ldots e'_1)_2$
  \State $\hat{e} \gets e - (2^{E-1}-1) + \mathbf{1}[e=0]$
  \State $\hat{e}' \gets e' - (2^{E-1}-1) + \mathbf{1}[e'=0]$
  \State $f \gets (1\,f_1 \ldots f_m)_2 - (\mathbf{1}[e = 0] \ll m)$
  \State $f' \gets (1\,f'_1 \ldots f'_m)_2 - (\mathbf{1}[e' = 0] \ll m)$%
    \label{algline:exactsubtract2-ef-ed}
  \State $f'_\mathrm{hi} \gets f' \gg \min\set{\hat{e} - \hat{e}', E+m}$%
    \label{algline:exactsubtract2-fpdecomp-st}
  \State $f'_\mathrm{lo} \gets f'\, \&\, ((1 \ll \min\set{\hat{e}-\hat{e}', m+1}) - 1)$%
    \label{algline:exactsubtract2-fpdecomp-ed}
  \State \hl{$n_1 \gets -\hat{e}-2 + \mathbf{1}[x=\frac{1}{2}]$}%
    \label{algline:exactsubtract2-final-st}
  \State $n_2 \gets \max\set{(\hat{e}-\hat{e}') - (m+1), 0}$
  \State \hl{$n_\mathrm{hi} \gets m+2 - \mathbf{1}[x=\frac{1}{2}]$}
  \State $n_\mathrm{lo} \gets \min\set{\hat{e}-\hat{e}', m+1}$
  \State \hl{$b_1 \gets 1$}
  \State $b_2 \gets \mathbf{1}[f'_\mathrm{lo} > 0]$
  \State \hl{$g_\mathrm{hi} \gets (1 \ll {n_\mathrm{hi}}) - f - f'_\mathrm{hi} - b_2$}
  \State $g_\mathrm{lo} \gets (b_2 \ll n_\mathrm{lo}) - f'_\mathrm{lo}$%
    \label{algline:exactsubtract2-final-ed}
  \State \Return $(n_1, n_2, n_\mathrm{hi}, n_\mathrm{lo}, b_1, b_2, g_\mathrm{hi}, g_\mathrm{lo})$
\EndFunction
\end{algorithmic}
\end{algorithm}
\end{minipage}

\begin{minipage}[t]{.495\linewidth}
\begin{algorithm}[H]
\captionsetup{hypcap=false}
\caption{Preprocessing for \hyperref[alg:getbit]{\GetBit}}
\label{alg:exactsubtractd}
\begin{algorithmic}[1]
\Require{$(d, f), (d', f')$ with $d, d' \in \set{0,1}$\\
  and $f,f' \in \floatEm \cap [0,\frac{1}{2}]$}
\Ensure{$(n_1, n_2, n_\mathrm{hi}, n_\mathrm{lo}, b_1, b_2, g_\mathrm{hi}, g_\mathrm{lo})$}
\Function{\ExactSubtract}{$d, f, d', f'$}
  \If{$d = d' = 0$}
    \LComment{$f -_\real f'$}
    \State \Return $\hyperref[alg:exactsubtract1]{\ExactSubtractOne}(f,f')$ \EndIf
  \If{$d = d' = 1$}
    \LComment{$f' -_\real f$}
    \State \Return $\hyperref[alg:exactsubtract1]{\ExactSubtractOne}(f',f)$ \EndIf
  \If{$d = 1 \mbox{\bfseries\ and } d' = 0$}
    \LComment{$1 -_\real (f +_\real f')$}
    \State \Return $\hyperref[alg:exactsubtract2]{\ExactSubtractTwo}(f,f')$ \EndIf
  \State \textbf{error}
\EndFunction
\end{algorithmic}
\end{algorithm}
\end{minipage}
\hfill
\begin{minipage}[t]{.495\linewidth}
\begin{algorithm}[H]
\captionsetup{hypcap=false}
\caption{Extract Binary Digit}
\label{alg:getbit}
\begin{algorithmic}[1]
\Require{%
  $\beta {\defas} (n_1, n_2, n_\mathrm{hi}, n_\mathrm{lo}, b_1, b_2, g_\mathrm{hi}, g_\mathrm{lo}),\ell{\ge}1$;
  where
  $\begin{array}[t]{@{}l}
  n_1, n_2, n_\mathrm{hi}, n_\mathrm{lo} \ge 0,\;
  \text{\hl{$b_1, b_2 \in \set{0,1}$}},
  \\
  0 \leq g_\mathrm{hi} < 2^{n_\mathrm{hi}},\; 0 \leq g_\mathrm{lo} < 2^{n_\mathrm{lo}},\;
  \end{array}$
  \\
  are from
  \\
  $\hyperref[alg:exactsubtractd]{\ExactSubtract}((d,f),(d',f'))$
  }
\Ensure{$b' \in \bool$; such that\\
  if $b_1 = 0$, then $b'$ is bit $\ell$ of $(x -_\real x')$;\\
  if $b_1 = 1$, then $b'$ is bit $\ell$ of $1 -_\real (x +_\real x')$;\\
  where $\begin{aligned}[t]
    x &\defas(1-d)f+d(1-f)\\
    x' &\defas(1-d')f'+d'(1-f')
    \end{aligned}$.
}
\Function{\GetBit}{$\beta,\ell$}
  \If{$\ell \leq n_1$}
    \Return \hl{$b_1$}
  \EndIf
  \If{$\ell \leq n_1 + n_\mathrm{hi}$}
    \Return $g_{\mathrm{hi}, \,\ell - n_1}$
  \EndIf
  \If{$\ell \leq n_1 + n_\mathrm{hi} + n_2$}
    \Return \hl{$b_2$}
  \EndIf
  \If{$\ell \leq n_1 + n_\mathrm{hi} + n_2 + n_\mathrm{lo}$}
    \State \Return $g_{\mathrm{lo}, \,\ell - (n_1 + n_\mathrm{hi} + n_2)}$
  \EndIf
  \State \Return 0
\EndFunction
\end{algorithmic}
\end{algorithm}
\end{minipage}
\end{figure}

% Samplers
\clearpage

\begin{algorithm}[t]
\caption{Extended-Accuracy Conditional Bit Sampling}
\label{alg:sampler-naive-ext-impl}
\begin{algorithmic}[1]
\Require{\ECDF{} $G: \bool^n \to \set{0,1} \times (\floatEm \cap [0,\frac{1}{2}])$ \\
  over binary number format $\bfmt = (n, \gamma_\bfmt, \phi_\bfmt)$ \\
  \color{gray}{String $b \in \bool^{\le n}$;}
  \color{gray}{Pairs $(\dL, \fL), (\dR, \fR) \in \set{0,1} \times (\floatEm \cap [0,\frac{1}{2}])$}
  }
\Ensure{Exact random variate $X \sim G$}
\Function{\SampleNaiveImpl}{%
    $G$, $b = \varepsilon$,
    $\dL = 0$, $\fL = 0$,
    $\dR = 0$, $\fR = 0$}
  \If{$\abs{b} = n$}
    \Comment{Base Case}
    \State \Return $\phi_\bfmt(b)$
    \Comment{String in Format ${\bfmt}$}
  \EndIf
  \State $b' \gets b01^{n - \abs{b} - 1}$
  \State $(\dM, \fM) \gets G(\phi_\bfmt(b'))$
  \If{$(\dM, \fM) = (\dR, \fR)$}
    \Comment{Leaf}
    \State \Return $\SampleNaiveImpl(G, b0, \dL, \fL, \dM, \fM)$
    \Comment{0}
  \EndIf
  \If{$(\dM, \fM) = (\dL, \fL)$}
    \Comment{Leaf}
    \State \Return $\SampleNaiveImpl(G, b1, \dM, \fM, \dR, \fR)$
    \Comment{1}
  \EndIf
  \LComment{$\displaystyle \frac{i}{k} \defas \frac
      {\big((1-\dR)\fR+\dR(1-\fR)\big) - \big((1-\dM)\fM+\dM(1-\fM)\big)}
      {\big((1-\dR)\fR+\dR(1-\fR)\big) - \big((1-\dL)\fL+\dL(1-\fL)\big)}$}
      \label{algline:sampler-naive-ext-impl-kn}
  \State $(i, k) \gets \ExactRatio(\dL, \fL, \dM, \fM, \dR, \fR)$
  \State $z \gets \hyperref[alg:sampler-bernoulli-impl]{\Bernoulli}(i, k)$
  \Comment{Refine Subtree}
  \If{$z=0$}
    \State \Return $\SampleNaiveImpl(G, b0, \dL, \fL, \dM, \fM)$
      \Comment{0}
  \Else
    \State \Return $\SampleNaiveImpl(G, b1, \dM, \fM, \dR, \fR)$
      \Comment{1}
  \EndIf
\EndFunction
\end{algorithmic}
\end{algorithm}

\begin{algorithm}[t]
\caption{Extended-Accuracy Entropy-Optimal Random Variate Generation}
\label{alg:sampler-opt-ext-impl}
\begin{algorithmic}[1]
\Require{{\ECDF{} $G: \bool^n \to \set{0,1} \times (\floatEm \cap [0,\frac{1}{2}])$} \\
  over binary number format $\bfmt = (n, \gamma_\bfmt, \phi_\bfmt)$ \\
  \color{gray}{String $b \in \bool^{\le n}$; \#Flips $\ell \ge 0$;} %\\
  \color{gray}{Pairs $(\dL, \fL), (\dR, \fR) \in \set{0,1} \times (\floatEm \cap [0,\frac{1}{2}])$}}
\Ensure{Exact random variate $X \sim G$} % $v \sim F$
\Function{\SampleOptImpl}{$G$, $b=\varepsilon$, $\ell=0$, $\dL=0$, $\fL=0$, $\dR=1$, $\fR=0$}
  \If{$\abs{b} = n$}
    \Comment{Base Case}
    \State \Return $\phi_\bfmt(b)$ % \gamma_\bfmt(...)
    \Comment{String in Format ${\bfmt}$} % \Comment{Number in $\realext_{\bfmt}$}
  \EndIf
  \State $b' \gets b01^{n - \abs{b} - 1}$
  \State $(\dM, \fM) \gets G(\phi_\bfmt(b'))$
  \If{$(\dM, \fM) = (\dR, \fR)$}
    \Comment{Leaf}
    \State \Return $\SampleOptImpl(G, b0, \ell, \dL, \fL, \dM, \fM)$
      \Comment{0}
  \EndIf
  \If{$(\dM, \fM) = (\dL, \fL)$}
    \Comment{Leaf}
    \State \Return $\SampleOptImpl(G, b1, \ell, \dM, \fM, \dR, \fR)$
      \Comment{1}
  \EndIf
  % Recursive Case.
  \LComment{$r_0 \defas \big((1-\dM)\fM+\dM(1-\fM)\big) - \big((1-\dL)\fL+\dL(1-\fL)\big) \in \real$}
  \LComment{$r_1 \defas \big((1-\dR)\fR+\dR(1-\fR)\big) - \big((1-\dM)\fM+\dM(1-\fM)\big) \in \real$}
  \State $\beta_0 \gets \hyperref[alg:exactsubtractd]{\ExactSubtract}(\dM, \fM, \dL, \fL)$
  \State $\beta_1 \gets \hyperref[alg:exactsubtractd]{\ExactSubtract}(\dR, \fR, \dM, \fM)$
  \If{$\ell > 0$}
    \State $a_0 \gets \hyperref[alg:getbit]{\GetBit}(\beta_0, {\ell})$
    \Comment{$[r_0]_\ell$}
    \State $a_1 \gets \hyperref[alg:getbit]{\GetBit}(\beta_1, {\ell})$
    \Comment{$[r_1]_\ell$}
    \If{$a_0 = 1$ and $a_1 = 0$}
      \Comment{Leaf}
      % \Comment{Refine 0}
        \State \Return \SampleOptImpl($G$, $b0$, $\ell$, $\dL,$ $\fL$, $\dM$, $\fM$)
          \Comment{0}
    \EndIf
    \If{$a_0 = 0$ and $a_1 = 1$}
      \Comment{Leaf}
      % \Comment{Refine 1}
        \State \Return \SampleOptImpl($G$, $b1$, $\ell$, $\dM$, $\fM$, $\dR$, $\fR$)
          \Comment{1}
    \EndIf
  \EndIf
  \While{$\textbf{true}$} \Comment{Refine Subtree}
    \State $x \gets \Flip()$; $\ell \gets \ell + 1$
    \State $a_0 \gets \hyperref[alg:getbit]{\GetBit}(\beta_0, {\ell})$
    \Comment{$[r_0]_\ell$}
    \State $a_1 \gets \hyperref[alg:getbit]{\GetBit}(\beta_1, {\ell})$
    \Comment{$[r_1]_\ell$}
    \If{$x = 0$ and $a_0 = 1$}
      \Comment{Leaf}
      \State \Return \SampleOptImpl($G$, $b0$, $\ell$, $\dL$, $\fL$, $\dM$, $\fM$)
        \Comment{0}
    \EndIf
    \If{$x = 1$ and $a_1 = 1$}
      \Comment{Leaf}
      \State \Return \SampleOptImpl($G$, $b1$, $\ell$, $\dM$, $\fM$, $\dR$, $\fR$)
        \Comment{1}
    \EndIf
  \EndWhile
\EndFunction
\end{algorithmic}
\end{algorithm}

\begin{algorithm}[t]
\captionsetup{hypcap=false}
\caption{Quantile of an \ECDF{}}
\label{alg:quantile-ext}
\begin{algorithmic}[1]
\Require{\ECDF{} $G: \bool^n \to \set{0,1} \times (\floatEm \cap [0,\frac{1}{2}])$ \\
over binary number format $\bfmt = (n, \gamma_\bfmt, \phi_\bfmt)$ \\
Pair $(d, f) \in \bool \times (\floatEm \cap [0,\frac{1}{2}])$
}
\Ensure{%
  $\min_{<_\bfmt}\set{b \in \bool^n \mid (d,f) \preceq G(b)}$, \\
  where $(d,f) \preceq (d',f') \iff (1-d)f + d(1-f) \le_{\real} (1-d')f' + d'(1-f')$
}
\Function{\QuantileExt}{$G, d, f$}
\State $(n, l, h) \gets (1 + E + m, 0, 2^n-1)$
\While{$l \le h$}
  \State $s \gets \floor{(l+h)/2}$
  \State $s' \gets \phi_\bfmt(\gamma^{-1}_{\mathbb{U}_n}(s))$
  \State $(d', f') \gets G(s')$
  \If{$\hyperref[algline:CompareLte]{\textsc{CompareLte}}(d,f,d',f')$}
    \State $h \gets s - 1$; $t \gets s'$
  \Else
    \State $l \gets s + 1$
  \EndIf
\EndWhile
\State \Return $t$
\EndFunction
\Require{$d,d' \in \set{0,1}$, $f, f' \in \floatEm \cap [0,\frac{1}{2}]$}
\Ensure{$(d,f) \preceq (d',f')$}
\Function{CompareLte}{$d, d', f, f'$} \label{algline:CompareLte}
\If{$d<d'$} \Return \textbf{true} \EndIf
\If{$d = d' = 0 \,\mathbf{and}\, f  \le f'$} \Return \textbf{true} \EndIf
\If{$d = d' = 1 \,\mathbf{and}\, f' \le f$} \Return \textbf{true} \EndIf
\State \Return \textbf{false}
\EndFunction
\end{algorithmic}
\end{algorithm}

\clearpage
\enlargethispage{5\baselineskip}

\section{Survey of Numerical Errors in Python Random Variate Generation Libraries}
\label{appx:survey}
%!TEX root=./paper.tex
% \smallskip
\begin{center}
\begin{adjustbox}{max width=.92\linewidth}
\begin{tabular}{|l|ll|}
\hline\hline
NumPy   & BUG: random: Problems with hypergeometric with ridiculously large arguments                                 & \url{https://github.com/numpy/numpy/issues/11443} \\
NumPy   & Possible bug in random.laplace                                                                              & \url{https://github.com/numpy/numpy/issues/13361} \\
NumPy   & Bias of random.integers() with int8 dtype                                                                   & \url{https://github.com/numpy/numpy/issues/14774} \\
NumPy   & Geometric, negative binomial and poisson fail for extreme arguments                                         & \url{https://github.com/numpy/numpy/issues/1494} \\
NumPy   & numpy.random.hypergeometric: error for some cases                                                           & \url{https://github.com/numpy/numpy/issues/1519} \\
NumPy   & numpy.random.logseries - incorrect convergence for k=1, k=2                                                 & \url{https://github.com/numpy/numpy/issues/1521} \\
NumPy   & Von Mises draws not between -pi and pi [patch]                                                              & \url{https://github.com/numpy/numpy/issues/1584} \\
NumPy   & Negative binomial sampling bug when p=0                                                                     & \url{https://github.com/numpy/numpy/issues/15913} \\
NumPy   & default\_rng.integers(2**32) always return 0                                                                & \url{https://github.com/numpy/numpy/issues/16066} \\
NumPy   & Beta random number generator can produce values outside its domain                                          & \url{https://github.com/numpy/numpy/issues/16230} \\
NumPy   & OverflowError for np.random.RandomState()                                                                   & \url{https://github.com/numpy/numpy/issues/16695} \\
NumPy   & binomial can return unitialized integers when size is passed with array values for a or p                   & \url{https://github.com/numpy/numpy/issues/16833} \\
NumPy   & np.random.geometric(10**-20) returns negative values                                                        & \url{https://github.com/numpy/numpy/issues/17007} \\
NumPy   & numpy.random.vonmises() fails for kappa > ~10\^8                                                            & \url{https://github.com/numpy/numpy/issues/17275} \\
NumPy   & Wasted bit in random float32 generation                                                                     & \url{https://github.com/numpy/numpy/issues/17478} \\
NumPy   & test\_pareto on 32-bit got even worse                                                                       & \url{https://github.com/numpy/numpy/issues/18387} \\
NumPy   & Silent overflow error in numpy.random.default\_rng.negative\_binomial                                       & \url{https://github.com/numpy/numpy/issues/18997} \\
NumPy   & Possible mistake in distribution.c::rk\_binomial\_btpe                                                      & \url{https://github.com/numpy/numpy/issues/2012} \\
NumPy   & mtrand.beta does not handle small parameters well                                                           & \url{https://github.com/numpy/numpy/issues/2056} \\
NumPy   & random.uniform gives inf when using finfo('float').min, finfo('float').max as intervall                     & \url{https://github.com/numpy/numpy/issues/2138} \\
NumPy   & BUG: numpy.random.Generator.dirichlet should accept zeros.                                                  & \url{https://github.com/numpy/numpy/issues/22547} \\
NumPy   & numpy.random.randint(-2147483648, 2147483647) raises ValueError: low >= high                                & \url{https://github.com/numpy/numpy/issues/2286} \\
NumPy   & BUG: random: beta (and therefore dirichlet) hangs when the parameters are very small                        & \url{https://github.com/numpy/numpy/issues/24203} \\
NumPy   & BUG: random: dirichlet(alpha) can return nans in some cases                                                 & \url{https://github.com/numpy/numpy/issues/24210} \\
NumPy   & BUG: random: beta can generate nan when the parameters are extremely small                                  & \url{https://github.com/numpy/numpy/issues/24266} \\
NumPy   & BUG: Inaccurate left tail of random.Generator.dirichlet at small alpha                                      & \url{https://github.com/numpy/numpy/issues/24475} \\
NumPy   & Cannot generate random variates from noncentral chi-square distribution with dof = 1                        & \url{https://github.com/numpy/numpy/issues/5766} \\
NumPy   & Bug in np.random.dirichlet for small alpha parameters                                                       & \url{https://github.com/numpy/numpy/issues/5851} \\
NumPy   & numpy.random.poisson(0) should return 0                                                                     & \url{https://github.com/numpy/numpy/issues/827} \\
NumPy   & Could random.hypergeometric() be made to match behavior of random.binomial() when sample or n = 0?          & \url{https://github.com/numpy/numpy/issues/9237} \\
NumPy   & BUG: np.random.zipf hangs the interpreter on pathological input                                             & \url{https://github.com/numpy/numpy/issues/9829} \\
PyTorch & torch.distributions.categorical.Categorical samples indices with zero probability                           & \url{https://github.com/pytorch/pytorch/issues/100884} \\
PyTorch & Torch randperm with device mps does not sample exactly uniformly from all possible permutations             & \url{https://github.com/pytorch/pytorch/issues/104315} \\
PyTorch & torch.distributions.Pareto.sample sometimes gives inf                                                       & \url{https://github.com/pytorch/pytorch/issues/107821} \\
PyTorch & torch.multinomial - Unexpected (incorrect) results when replacement=True in version 2.1.1+cpu               & \url{https://github.com/pytorch/pytorch/issues/114945} \\
PyTorch & Strange behavior of randint using device=cuda                                                               & \url{https://github.com/pytorch/pytorch/issues/125224} \\
PyTorch & Beta Distribution values wrong for a=b---> 0                                                                & \url{https://github.com/pytorch/pytorch/issues/15738} \\
PyTorch & Very poor Uniform() sampling near floating 0.0                                                              & \url{https://github.com/pytorch/pytorch/issues/16706} \\
PyTorch & Full-range random\_() generation broken for cuda.IntTensor, cuda.LongTensor and LongTensor.                 & \url{https://github.com/pytorch/pytorch/issues/16944} \\
PyTorch & RelaxedBernoulli produces samples on the boundary with NaN log\_prob                                        & \url{https://github.com/pytorch/pytorch/issues/18254} \\
PyTorch & torch.distributions.Binomial.sample() uses a massive amount of memory                                       & \url{https://github.com/pytorch/pytorch/issues/20343} \\
PyTorch & Weird sampling from multinomial\_alias\_draw                                                                & \url{https://github.com/pytorch/pytorch/issues/21257} \\
PyTorch & CUDA implementation of alias multinomial doesn’t work correctly                                             & \url{https://github.com/pytorch/pytorch/issues/21508} \\
PyTorch & Wrong distribution sampled by torch.multinomial on CUDA                                                     & \url{https://github.com/pytorch/pytorch/issues/22086} \\
PyTorch & torch.nn.functional.gumbel\_softmax yields NaNs                                                             & \url{https://github.com/pytorch/pytorch/issues/22442} \\
PyTorch & torch.distributions.Normal cuda sampling broken                                                             & \url{https://github.com/pytorch/pytorch/issues/22529} \\
PyTorch & CPU torch.exponential\_ function may generate 0 which can cause downstream NaN                              & \url{https://github.com/pytorch/pytorch/issues/22557} \\
PyTorch & got nan when gumbel\_softmax calculated in GPU                                                              & \url{https://github.com/pytorch/pytorch/issues/22586} \\
PyTorch & torch.bernoulli() randomly returns "1" for 0 inputs on CPU                                                  & \url{https://github.com/pytorch/pytorch/issues/26807} \\
PyTorch & Uniform random generator generates too many zeros compared to NumPy                                         & \url{https://github.com/pytorch/pytorch/issues/26973} \\
PyTorch & torch.multinominal ignores elements from cumulative distribution                                            & \url{https://github.com/pytorch/pytorch/issues/28390} \\
PyTorch & Tensor.random\_ should be able to generate all 64 bit numbers including min and max value                   & \url{https://github.com/pytorch/pytorch/issues/33299} \\
PyTorch & torch.multinomial behaves abnormally with CUDA tensor                                                       & \url{https://github.com/pytorch/pytorch/issues/37403} \\
PyTorch & Investigate using -cospi(u) / sinpi(u) instead of tan(pi * (u - 0.5)) in transformation::cauchy             & \url{https://github.com/pytorch/pytorch/issues/38611} \\
PyTorch & Investigate exponential distribution improvements                                                           & \url{https://github.com/pytorch/pytorch/issues/38612} \\
PyTorch & [bug] Binomial distribution has small chance of returning -1                                                & \url{https://github.com/pytorch/pytorch/issues/42153} \\
PyTorch & torch.multinomial with replacement=True produces inaccurate results for large number of categories          & \url{https://github.com/pytorch/pytorch/issues/43115} \\
PyTorch & torch.multinomial behave unexpectedly on float16 GPU input tensor                                           & \url{https://github.com/pytorch/pytorch/issues/46702} \\
PyTorch & torch.multinomial selects elements with zero weight                                                         & \url{https://github.com/pytorch/pytorch/issues/48841} \\
PyTorch & Multinomial without replacement produces samples that have zero probability                                 & \url{https://github.com/pytorch/pytorch/issues/50034} \\
PyTorch & Cauchy samples inf values on CUDA                                                                           & \url{https://github.com/pytorch/pytorch/issues/59144} \\
PyTorch & [Bug] cuda version of torch.randperm(n) generate all zero/negative/large positive values for large n        & \url{https://github.com/pytorch/pytorch/issues/59756} \\
PyTorch & a problem happened in torch.randperm                                                                        & \url{https://github.com/pytorch/pytorch/issues/63726} \\
PyTorch & Gamma distribution returns some wrong extreme values                                                        & \url{https://github.com/pytorch/pytorch/issues/71414} \\
PyTorch & Dirichlet with small concentration                                                                          & \url{https://github.com/pytorch/pytorch/issues/76030} \\
PyTorch & torch.randperm uses too much cpu, but not efficient.                                                        & \url{https://github.com/pytorch/pytorch/issues/77140} \\
PyTorch & torch.randint should accept high=2**63                                                                      & \url{https://github.com/pytorch/pytorch/issues/81446} \\
PyTorch & Beta distribution behaves incorrectly for small parameters                                                  & \url{https://github.com/pytorch/pytorch/issues/84625} \\
PyTorch & Poisson sampling on GPU fails for high rates                                                                & \url{https://github.com/pytorch/pytorch/issues/86782} \\
PyTorch & Hang: sampling VonMises distribution gets stuck in rejection sampling for small kappa                       & \url{https://github.com/pytorch/pytorch/issues/88443} \\
PyTorch & torch.normal(...) on MPS sometimes produces NaN's                                                           & \url{https://github.com/pytorch/pytorch/issues/89127} \\
PyTorch & torch.randn and torch.normal sometimes produce NaN on mps device                                            & \url{https://github.com/pytorch/pytorch/issues/89283} \\
PyTorch & torch.Categorical samples indexes with 0 probability when given logits as argument                          & \url{https://github.com/pytorch/pytorch/issues/91863} \\
PyTorch & [Inductor] philox randn doesn't follow standard normal distribution                                         & \url{https://github.com/pytorch/pytorch/issues/91944} \\
PyTorch & distributions.Beta returning incorrect results at 0 and 1                                                   & \url{https://github.com/pytorch/pytorch/issues/92260} \\
PyTorch & torch.distributions.kumaraswamy.Kumaraswamy generates samples outside its support (0,1)                     & \url{https://github.com/pytorch/pytorch/issues/95548} \\
PyTorch & torch.rand can sample the upper bound for lower precision floating point dtypes on CUDA                     & \url{https://github.com/pytorch/pytorch/issues/96947} \\
SciPy   & overflow in truncnorm.rvs                                                                                   & \url{https://github.com/scipy/scipy/issues/10092} \\
SciPy   & truncnorm.rvs Weird Behaviors                                                                               & \url{https://github.com/scipy/scipy/issues/11769} \\
SciPy   & truncnorm.rvs is painfully slow on version 1.5.0rc2                                                         & \url{https://github.com/scipy/scipy/issues/12370} \\
SciPy   & Levy Stable Random Variates Code has a typo                                                                 & \url{https://github.com/scipy/scipy/issues/12870} \\
SciPy   & truncnorm.rvs still crashes when sampling from extreme tails                                                & \url{https://github.com/scipy/scipy/issues/13966} \\
SciPy   & truncnorm rvs() produces junk for ranges in the tail                                                        & \url{https://github.com/scipy/scipy/issues/1489} \\
SciPy   & BUG: Levy stable                                                                                            & \url{https://github.com/scipy/scipy/issues/14994} \\
SciPy   & BUG: scipy.stats.multivariate\_hypergeom.rvs raises ValueError when at least the last two populations are 0 & \url{https://github.com/scipy/scipy/issues/16171} \\
SciPy   & BUG: truncnorm rvs sometimes returns nan (float32 issue?)                                                   & \url{https://github.com/scipy/scipy/issues/19554} \\
SciPy   & binomial (in numpy.random) and scipy.stats.binom are not defined for n=0                                    & \url{https://github.com/scipy/scipy/issues/2213} \\
SciPy   & stats.truncnorm.rvs() does not give symmetric results for negative \& positive regions                      & \url{https://github.com/scipy/scipy/issues/2477} \\
SciPy   & scipy.stats.rice.rvs(b) returns bad random numbers (zeros) for b greater than 10                            & \url{https://github.com/scipy/scipy/issues/3282} \\
SciPy   & scipy.stats.ncx2 fails for nc=0                                                                             & \url{https://github.com/scipy/scipy/issues/5441} \\ \hline\hline
\end{tabular}
\end{adjustbox}
\end{center}

}

\end{document}